\documentclass[11pt]{article}
\usepackage{amsmath,amssymb,amsfonts,latexsym,graphicx,amsthm}
\usepackage{fullpage,color}
\usepackage{url}
\usepackage[hidelinks]{hyperref}
\usepackage{comment}
\usepackage[linesnumbered,boxed,ruled,vlined]{algorithm2e}
\usepackage{framed}
\usepackage[shortlabels]{enumitem}
\usepackage{cleveref}
\usepackage{tikz}
\usepackage{thm-restate}
\usetikzlibrary{calc}
\usetikzlibrary{intersections}
\usepackage{tcolorbox}

\usepackage[margin=1in]{geometry}

\newtheorem{theorem}{Theorem}[section]
\newtheorem{lemma}{Lemma}[section]
\newtheorem{definition}{Definition}[section]
\newtheorem{corollary}{Corollary}[section]
\newtheorem{claim}{Claim}[section]
\newtheorem{question}{Question}[section]
\newtheorem{invariant}{Invariant}[section]

\newtheorem{observation}{Observation}[section]

        {\medskip}

\newcommand{\eps}{\epsilon}

\newcommand{\ceil}[1]{\lceil #1 \rceil}

\newcommand{\brac}[1]{\left(#1\right)}
\newcommand{\opt}{\mathsf{opt}}
\newcommand{\mst}{\mathsf{mst}}
\newcommand{\MST}{\mathrm{MST}}
\newcommand{\diam}{\mathrm{diam}}
\newcommand{\ball}{B} 

\newcommand{\proj}{\mathsf{proj}}
\newcommand{\dist}{\mathsf{dist}}
\newcommand{\real}{\mathbb{R}}
\newcommand{\cost}{\mathbf{C}}
\newcommand{\sparse}{\mathrm{spa}}
\newcommand{\light}{\mathrm{light}}
\newcommand{\weight}{\mathbf{w}}
\definecolor{BrickRed}{rgb}{.72,0,0}
\def\EMPH#1{\emph{\textcolor{BrickRed} {#1}}}

\begin{document}

\begin{titlepage}
	\title{Towards Instance-Optimal Euclidean Spanners}
	\author{
		Hung Le\thanks{University of Massachusetts Amherst, \href{}{hungle@cs.umass.edu}}\and
		Shay Solomon\thanks{Tel Aviv University, \href{}{solo.shay@gmail.com}}\and	
		Cuong Than\thanks{University of Massachusetts Amherst, \href{}{cthan@umass.edu}}\and
		Csaba D. T\'oth\thanks{California State University Northridge and Tufts University, \href{}{csaba.toth@csun.edu}}\and
		Tianyi Zhang\thanks{Tel Aviv University, \href{}{tianyiz21@tauex.tau.ac.il}}}
	
	\date{}

	\maketitle
	\thispagestyle{empty}

	\begin{abstract}
		Euclidean spanners are important geometric objects that have been extensively studied since the 1980s.
		The two most basic ``compactness'' measures of a Euclidean spanner $E$ \footnote{We shall identify a graph $H = (X,E)$ with its edge set $E$. All edge weights are given by the Euclidean distances.}
        are the size (number of edges) $|E|$ and the weight (sum of edge weights) $\|E\|$. The state-of-the-art constructions of Euclidean $(1+\eps)$-spanners in $\mathbb{R}^d$ have $O_d\brac{n\cdot \epsilon^{-d+1}}$ edges (or \emph{sparsity}
        $O_d(\epsilon^{-d+1})$)
        and weight $O_d\brac{\epsilon^{-d}\log\epsilon^{-1}} \cdot \|E_\mst\|$  (or \emph{lightness} $O_d(\epsilon^{-d}\log\epsilon^{-1})$); here $O_d$ suppresses a factor of $d^{O(d)}$ and  $\|E_\mst\|$ denotes the weight of a minimum spanning tree of the input point set. 
        
        Importantly, these two upper bounds are (near-)optimal (up to the $d^{O(d)}$ factor and disregarding the factor of $\log(\epsilon^{-1})$ in the lightness bound) for some {\em extremal
        instances} [Le and Solomon, 2019], and therefore they are (near-)\EMPH{optimal in an existential sense}.
        Moreover, both these upper bounds are attained by the same construction---the classic \EMPH{greedy spanner}, whose sparsity and lightness are not only existentially optimal, but they also significantly outperform those of any other Euclidean spanner construction studied in an experimental study by [Farshi-Gudmundsson, 2009]
        for various practical point sets in the plane.
        This raises the natural question of whether the greedy spanner is (near-) optimal for \EMPH{any point set instance}? 
        
  Motivated by this question, we initiate the study of \EMPH{instance optimal Euclidean spanners}. Our results are two-fold.
		\begin{itemize}[leftmargin=*]
        \item Rather surprisingly (given the aforementioned experimental study), we demonstrate that the \EMPH{greedy spanner is  far from being instance optimal}, even when allowing its stretch to grow. More concretely, we design two hard instances of point sets in the plane, where the greedy $(1+x \epsilon)$-spanner (for basically any parameter $x \ge 1$) has $\Omega_x(\epsilon^{-1/2}) \cdot |E_\sparse|$ edges and weight $\Omega_x(\epsilon^{-1}) \cdot \|E_\light\|$, where 
        $E_\sparse$ and $E_\light$ denote the per-instance sparsest and lightest $(1+\eps)$-spanners, respectively,
        and the $\Omega_x$ notation suppresses a polynomial dependence on $1/x$.

		\item  As our main contribution, we design a new construction 
        of Euclidean spanners, which is inherently different from known constructions, achieving the following bounds:  a stretch of $1+\epsilon\cdot 2^{O(\log^*(d/\epsilon))}$
        with $O(1) \cdot |E_\sparse|$ edges and weight  
        $O(1) \cdot \|E_\light\|$.
        In other words, we show that a slight increase to the stretch suffices for obtaining \EMPH{instance optimality up to an absolute constant for both sparsity and lightness}. Remarkably, there is only a log-star dependence on the dimension in the stretch,
        and there is no dependence on it whatsoever in the number of edges and weight. In general, for any integer $k \ge 1$, we can construct a Euclidean spanner in $\mathbb{R}^d$ of stretch $1+\epsilon\cdot 2^{O(k)}$ with $O\brac{\log^{(k)}(\eps^{-1}) + \log^{(k-1)}(d)} \cdot |E_\sparse|$ edges and weight $O\brac{\log^{(k)}(\eps^{-1}) + \log^{(k-1)}(d)} \cdot \|E_\light\|$, where $\log^{(k)}$ denotes the $k$-iterated logarithm.	
		\end{itemize}		
	\end{abstract}

\end{titlepage}

\tableofcontents

\thispagestyle{empty}
\clearpage
\pagenumbering{arabic}
\setcounter{page}{1}

\newpage

\section{Introduction}

A \EMPH{Euclidean $t$-spanner} of a set of point $X\subset \real^d$ is an edge-weighted graph $H = (X,E,\weight)$ having $X$ as the vertex set, such that the weight $\weight(x,y)$ of any edge $xy \in E$ is given by the Euclidean distance $\|xy\|$ between the endpoints $x$ and $y$, and for any $x,y\in X$,
there is a path in $H$ 
of {\em weight}
at most $t\cdot \|xy\|$; such a path is called a {\em $t$-spanner path}. The parameter $t$ is called the \EMPH{stretch} of the spanner.  The most important ``compactness'' measure of a Euclidean spanner is its number of edges or its \EMPH{sparsity}, which is the ratio of the number of edges of the spanner to the number of points. Another basic compactness measure of a Euclidean spanner is its weight (i.e., total edge weight) or its \EMPH{lightness}, which is the ratio of its weight to the weight of a Euclidean minimum spanning tree (MST).

The pioneering work by Chew~\cite{Chew86} showed that one could construct a spanner with \emph{constant} stretch and \emph{constant} sparsity for any point set in the plane.  Since then, Euclidean spanners have been extensively and intensively studied. Over more than three decades, Chew's result has been strengthened and generalized, culminating with the following result: For any parameter $\eps\in (0,1)$ and for any finite point set in $\real^d$, one can construct a $(1+\eps)$-spanner with sparsity $O_d(\eps^{-d+1})$ and/or lightness $O_d(\eps^{-d} \log(\eps^{-1}))$; here $O_d$ suppresses a factor of $d^{O(d)}$. The sparsity upper bound of $O_d(\eps^{-d+1})$ is realized by various classic constructions, such as the Theta-graph~\cite{Clarkson87,Keil88,RS91}, Yao graph~\cite{Yao82}, and the greedy spanner~\cite{althofer1993sparse}; the upper bound arguments, which were given already in the 90s, are short and simple. (We refer the readers to the book by Narasimhan and Smid~\cite{NS07C} for a  comprehensive coverage of spanner constructions that achieve a similar sparsity bound.) On the other hand, the lightness upper bound of $O_d(\eps^{-d} \log(\eps^{-1}))$
is much more complex, and it was proven rather recently by Le and Solomon~\cite{le2022truly}; in contrast to the sparsity upper bound, the lightness upper bound was proved only for the greedy spanner. In the same paper \cite{le2022truly}, it was shown that, for any  $d \geq 1$ and for any $n = \Omega_d(\eps^{-d} \log(\eps^{-1}))$, there is a set $X$ of $n$ points in $\real^d$ such that any $(1+\eps)$-spanner for $X$ must have sparsity $\Omega_d(\eps^{-d+1})$ and lightness $\Omega_d(\eps^{-d})$. Their results imply that the greedy spanner is \EMPH{existentially optimal} for \emph{both} sparsity and lightness
and the Theta- and Yao-graphs are {existentially optimal} just for sparsity; by existentially optimal we mean (near-)optimal in an existential sense to be formally defined below. 

\paragraph{Existential optimality.~} Let $A$ be a polynomial-time algorithm that takes a point set $X \in \real^d$ and $\eps\in (0,1)$ as input, and outputs a Euclidean $(1+\eps)$-spanner for $X$, denoted by $A(X,\eps)$.  Let $\cost(G)$ be a cost function imposed on a graph $G$. For our purposes, function $\cost$ either counts the number of edges (corresponding to sparsity) or the total edge weight (corresponding to lightness). 
We say that algorithm $A$ is \EMPH{existentially optimal for $\cost$ with optimality ratio $\kappa$} (for the Euclidean space $\real^d$) if, for every positive integer $n$,  \underline{\textbf{there exists}}
an $n$-point set $P_n$ in $\real^d$ such that 
\begin{center}
for any $n$-point set $X_n\in \real^d$, we have $\cost(A(X_n,\eps)) \leq \kappa \cdot \opt_\cost(P_n,\eps)$,    
\end{center}
 where $\opt_\cost(P_n,\eps)$ is the cost of the optimal $(1+\eps)$-spanner for $P_n$ under $\cost$. 
 Here the point set $P = P_n$ serves as a ``hard'' or ``extreme'' instance; that is, for existential optimality of algorithm $A$, it suffices to show the existence of a single hard instance (for any $n$), where the cost
of $A$ on any $n$-point set is no worse (by more than a factor of $\kappa$) than the {\em optimal} cost on the hard instance $P$.
 The notion of existential optimality in the context of graph spanners was explicitly formulated by Filtser and Solomon~\cite{FiltserS20}, though the general idea of existential optimality was implicitly used long before.  Here, we tailor their definition to the Euclidean space $\real^d$.  
 
 Ideally, we would like to design an existentially optimal algorithm $A$ with $\kappa = 1$, but this is too much to ask for: No known (polynomial time) algorithm in the spanner literature is existentially optimal with $\kappa = 1$ (or even close to that). The known constructions are existentially optimal with optimality ratio $\kappa = d^{O(d)}$ 
 or $\kappa = 2^{O(d)} \cdot O_\eps(1)$, where $O_\eps(1)$ is a constant that  depends only on $\eps$; thus the optimality ratio of all known constructions is a constant (and typically a large one) that depends at least exponentially on the dimension $d$.
 We shall henceforth say that algorithm $A$ is \EMPH{existentially optimal for $\cost$} if the optimality ratio $\kappa$ is a constant that  depends only on the dimension $d$; importantly, $\kappa$ \emph{must not} depend on $\eps$ and on the size $n$ of the input point sets.

The hard instance used by Le and Solomon~\cite{LeS23}, denoted here by $S_n$, is basically a set of $n$ evenly spaced points on the boundary of a $d$-dimensional unit sphere.\footnote{For sparsity, one has to use multiple vertex-disjoint copies of the sphere that are well-separated from each other.} Let us start with the sparsity
(so that the function $\cost$ counts the number of edges).  Le and Solomon~\cite{LeS23} showed that $\opt_{\cost}(S_n,\eps) \geq \alpha_d\cdot \eps^{-d+1} |S_n|$  for some $\alpha_d = 2^{-O(d)}$. For the Theta-graph construction, denoted here by $\Theta$, it is known~\cite{Clarkson87,Keil88,RS91} that for every $n$-point set $X \in \real^d$, we have $\cost(\Theta(X_n,\eps)) \leq \beta_{d} \cdot \eps^{-d+1} |X_n|$ for some  $\beta_d = 2^{O(d)}$. Thus, for every $n$-point set $X \in \real^d$:
\begin{equation*}
 \cost(\Theta(X_n,\eps))   \leq  \beta_{d}\cdot (1/\alpha_d) \cdot \opt_{\cost}(P_n,\eps) = 2^{O(d)} \opt_{\cost} \cdot (P_n,\eps)~,
\end{equation*}
implying that Theta-graphs are existentially optimal for sparsity. By the same argument, both Yao-graphs and greedy spanners are existentially optimal for sparsity~\cite{althofer1993sparse}. 
Next, for lightness (now the function $\cost$ counts the weight), 
 Le and Solomon~\cite{LeS23} showed that $\opt_{\cost}(S_n,\eps) \geq \alpha_d\cdot \eps^{-d}\,\|\MST(S_n)\|$  for some $\alpha_d = 2^{-O(d)}$, where $\MST(S_n)$ is a Euclidean minimum spanning tree of $S_n$. Combining this lower bound with the lightness upper bound of the greedy spanner by \cite{LeS23}, it follows that the greedy $(1+\eps)$-spanner algorithm is existentially optimal for lightness with \emph{optimality ratio $2^{O(d)}\log(\eps^{-1})$}. It is worth noting that removing the $\log(\eps^{-1})$ factor in the lightness optimality ratio of the greedy algorithm remains an open problem. 

Thus, the hard instance  by Le and Solomon~\cite{LeS23} allows one to ``declare'' that their spanner construction, if achieving $O_d(\eps^{-d+1})$ sparsity and/or $O_d(\eps^{-d})$ lightness, is (existentially) optimal and hence they could redirect their effort on optimizing other properties.  
However, existential optimality, while interesting in its own right, is a rather weak notion of optimality: it only requires that the algorithm has to perform as well as the optimal spanners (up to a factor of $\kappa$) for \EMPH{the hard(est) instance}. The hard instance might be impractical; indeed, this is the case with the aforementioned hard instance by Le and Solomon~\cite{LeS23} (basically a set of evenly spaced points on the boundary of a sphere),
which is very unlikely to appear in practice.
On the other hand, 
on {\em practical point sets},
an existentially optimal algorithm $A$ might perform  poorly compared to the optimal spanner. In fact, it is conceivable that for \EMPH{a wide range of point sets}, including various random distributions of points, the spanners produced by an existentially optimal algorithm $A$ have more edges than the \EMPH{instance optimal spanners}
by a factor of $\Omega_d(\eps^{-d+1})$. In extreme cases where $\eps$ is chosen so that  $\eps^{-d+1} = \Theta_d(n)$, $A$ could end up having $\Theta_d(n^2)$ edges while an optimal spanner has only $O_d(n)$ edges. Indeed, the experimental work by Farshi and Gudmundsson~\cite{FG09} showed that several existentially optimal algorithms (with respect to sparsity), including the Theta- and Yao-graphs, produce spanners with a much larger number of edges than the optimal spanners\footnote{In~\cite{FG09}, the optimal spanners (in terms of sparsity and other cost functions) were not computed explicitly; however for every instance,  one can take the best spanners over the collection of algorithms studied by Farshi and Gudmundsson as an upper bound for the cost of the optimal spanners.}. This naturally calls for a focus on a stronger notion of optimality, namely instance optimal spanners.

\paragraph{Instance optimality.~}  The notion of instance optimality was introduced by Fargin, Lotem, and Naorc \cite{FLN03} in the problem of choosing top $k$ items in sorted lists. We adapt this notion in our context as follows. We say that a polynomial-time algorithm $A$ is \EMPH{instance optimal for $\cost$ with optimality ratio $\kappa$} if \underline{\textbf{for every}} finite point set $X\in \real^d$, it holds that
\begin{center} 
$\cost(A(X,\eps)) \leq \kappa \cdot \opt_\cost(X,\eps)$,    
\end{center}
 where $\opt_\cost(X,\eps)$ is the cost of the optimal $(1+\eps)$-spanner for $X$ under  $\cost$.  
 
 An instance optimal algorithm with ratio $\kappa$ is, in fact, an {\em approximation algorithm} with an approximation factor $\kappa$. 
 The reason we chose to use the terminology of instance optimality rather than that of approximation is that the existential bounds in low-dimensional Euclidean spaces are already ``good'' in the sense that they are independent of the metric size, namely $\eps^{-O(d)}$ for stretch $1+\eps$ and dimension $d$. This stands in contrast to the existential bounds in general graphs that depend on the graph size, and for which there is a long line of influential work on approximation algorithms from both the upper and lower bounds sides; see \Cref{subsec:related} for more detail. We stress that improving bounds that are independent of the metric size is a completely different challenge, requiring new tools and techniques. By using the ``instance-optimality'' terminology, we wish to deviate from the line of work on approximation algorithms in general graphs and put the main focus on (1) the side of the upper bound (i.e., the ``optimality'' aspect) and (2) how the instance-optimal bounds compare to the existential-optimal bounds. Moreover, we hope that our work will initiate a systematic study on instance-optimal spanners in other graph families for which the existential bounds are independent of the graph size, such as unit-disk graphs, planar graphs, and bounded treewidth graphs.
 
 By definition, an instance optimal algorithm is also existentially optimal, but the converse is of course not true; indeed, designing an instance optimal algorithm appears to be significantly  harder. 
 First, designing an instance optimal algorithm with ratio $\kappa = 1$ is NP-hard for both sparsity~\cite{KK06,GKKKM10} and lightness~\cite{CC13}. Thus, to get a polynomial time algorithm,
 we shall allow the ratio $\kappa$ to grow beyond 1, to any constant that depends only on the dimension $d$; we refer to such an algorithm as \EMPH{instance optimal}. 
Even under this relaxation, no such algorithm is known to date, even in the plane!
This leads naturally to a fundamental
question in Euclidean spanners, summarized below.

\begin{tcolorbox}  \begin{question}\label{ques:approx} 
Can one design a polynomial-time instance optimal $(1+\eps)$-spanner algorithm with a ratio $\kappa = O_d(1)$ for sparsity and/or lightness? 
\end{question}
\end{tcolorbox}

\paragraph{The greedy spanner.}
A natural candidate of an instance optimal algorithm is the (path) greedy spanner~\cite{althofer1993sparse}.
As mentioned above, in the Euclidean space $\real^d$, it was shown to be \EMPH{existentially optimal} for sparsity, and also for lightness up to a factor of $\log(\eps^{-1})$  \cite{le2022truly}. In addition, Filtser and Solomon~\cite{FiltserS20} showed that greedy spanners are existentially optimal for very broad classes of graphs: those that are closed under edge deletions, which include general graphs and minor-closed families. They also showed that greedy spanners (and also an approximate version of the greedy spanner \cite{DN97,GLN02}) are existentially (near-)optimal for both sparsity and lightness in the family of {\em doubling metrics}, which is wider than that of Euclidean spaces. Experimental results~\cite{SZ04,FG09,ChimaniS22} also showed that the greedy spanners achieve the best quality in multiple aspects.

For general graphs, it was known that the approximation factor of the greedy spanners is $\Theta(n)$ for stretch $t < 3$~\cite{FiltserS20,ABSHLKS20} while an algorithm with a better approximation was known~\cite{KP94}. However, for higher stretch ($t\geq 5$), the greedy algorithm achieves the best-known approximation ratio.  Furthermore, there are works showing that the greedy spanner provides rather good {\em bicriteria approximation} algorithms in general graphs \cite{gudmundsson2022improving,wong23,BBGW24} for a related problem, called \emph{minimum dilation graph augmentation}: augmenting 
a graph as few edges as possible to reduce the dilation (a.k.a.\ stretch). 

Consequently, it appears that a preponderance of work on the greedy spanners all point in the same direction, i.e., to a
\EMPH{conjecture that the greedy algorithm gives a good bicriteria instance optimal spanner}. In this work, we demonstrate that the \EMPH{greedy spanner is far from being bicriteria instance optimal} for points in $\mathbb{R}^d$. 
More concretely, we design two hard instances of point sets in the plane, where the greedy $(1+x \epsilon)$-spanner (for basically any parameter $x \ge 1$) has $\Omega_x(\epsilon^{-1/2}) \cdot |E_\sparse|$ edges and weight $\Omega_x(\epsilon^{-1}) \cdot \|E_\light\|$, where 
        $E_\sparse$ and $E_\light$ denote the per-instance sparsest and lightest $(1+\eps)$-spanners, respectively,
        and the $\Omega_x$ notation suppresses a polynomial dependence on $1/x$. 

\paragraph{Bicriteria instance optimality.~}
Given that the greedy spanner is far from being instance optimal, even in the plane, and \EMPH{even when allowing its stretch to grow from $1+\eps$ to $1+x\eps$} for $x \ge 1$, and given that the greedy spanner appears to outperform any other known spanner construction in terms of both sparsity and lightness, the natural conclusion is that a \EMPH{new spanner construction is in order}. In light of our hardness result for the greedy spanner, it seems acceptable to allow the stretch to increase from $1+\eps$ to $1+x \eps$, for some reasonably small $x$.
This leads to the following question: Could standard techniques in geometric optimization, such as Arora's technique~\cite{Arora98}, be applied to construct instance optimal spanners? Arora's technique has been instrumental in solving problems such as the Euclidean TSP and Steiner tree. 
In our problem, however, the major difficulty is that we aim at optimizing the cost of the spanner {\em while  guaranteeing a stretch bound of $1+\eps$}. Alas, Arora's technique, as well as other known techniques, are not suitable for achieving \emph{both criteria}. 
Thus, it appears that \EMPH{a new technique for achieving both criteria is in order}.

The above discussion motivates us to consider \EMPH{bicriteria instance optimal spanners}: We say that an algorithm $A$ is \EMPH{$(c,\kappa)$-instance optimal} if for every point set $X\in \real^d$, $A(X,\eps)$ is a $(1+c\cdot \eps)$-spanner for $X$, and  $\cost(A(X,\eps)) \leq \kappa \cdot \opt_\cost(X,\eps)$.   

\begin{tcolorbox}  
\begin{question}\label{ques:bi-approx} 
Can one design a polynomial-time $(c,\kappa)$-instance optimal spanner algorithm with $c$ and $\kappa$ both bounded by  $O_d(1)$ (independent of $\eps$) for sparsity and/or lightness?
\end{question}
\end{tcolorbox}

While \Cref{ques:bi-approx} asks for constants $c$ and $\kappa$ (depending only on the dimension $d$), what was previously known is embarrassingly little. Even in the basic setting of the Euclidean plane,
the only positive result is a direct corollary of existentially optimal spanners, which gives a bicriteria $(c,\kappa)$-instance optimal spanner algorithm with $c = O(1)$ and $\kappa = O(\eps^{-1})$ for sparsity and  $\kappa = O(\eps^{-2})$ for lightness\footnote{Note that in the regime of parameters derived from existentially optimal spanners,  having $c = 1$ is the same as having $c = O(1)$ since one can apply the standard scaling trick: $\eps \leftarrow \eps/c$. Scaling reduces the stretch from $1+c\eps$ to $1+\eps$ while adding a factor of $c$ for sparsity and $c^2$ for lightness to $\kappa$.  The same trick, however, does not work for instance optimal spanners: if one scales $\eps$ to $\eps/c$, then one now essentially compares against the optimal $(1+\eps/c)$-spanner instead of the optimal $(1+\eps)$-spanner.}. To the best of our knowledge, no prior result achieves sublinear (respectively, subquadratic) dependence of $\kappa$ on $\eps^{-1}$ for sparsity (resp., lightness) in the Euclidean plane.  On the other hand, there is a lot of work on approximating spanners in general graphs, which we will review in more detail in \Cref{subsec:related}. A short takeaway is that the approximation factors in this setting depend on $n$ (the number of vertices), while in our setting, the approximation factor is independent of $n$ (the number of points). Therefore, the techniques for general graphs do not seem applicable to $\mathbb{R}^d$.

\subsection{Our Contribution}

 \paragraph{Hard instances.~}
We first construct point sets in $\real^2$ for which the \EMPH{greedy} $(1+\eps)$-spanner is far from being instance optimal in terms of sparsity or lightness, even if we relax the stretch from $1+\eps$ to a larger value $1+x\eps$.
\begin{restatable}[Sparsity lower bound for greedy]{theorem}{lbsparse} 
    \label{thm:sparsityLB+}
    For every sufficiently small $\eps>0$ and $1\leq x\leq o(\eps^{-1/3})$, there exists a finite set $S\subset \mathbb{R}^2$ such that 
\[
|E_{{\rm gr}(x)}| \geq 
\Omega\left(\frac{\eps^{-1/2}}{x^{3/2}}\right)\cdot |E_{\sparse}|,
\]
where $E_{{\rm gr}(x)}$ is the edge set of the greedy $(1+x\eps)$-spanner, and $E_{\sparse}$ is the edge set of a sparsest $(1+\eps)$-spanner for $S$.
\end{restatable}

\begin{restatable}[Lightness lower bound for greedy]{theorem}{lblight} 
    \label{thm:weightLB+}
For every sufficiently small $\eps>0$ and $x\in [2,\eps^{-1/2}/48]$, there exists a finite set $S\subset \mathbb{R}^2$ such that 
\[
\|E_{{\rm gr}(x)}\| \geq 
\Omega\left(\frac{\eps^{-1}}{x^{2}\cdot \log x}\right)\cdot \|E_{\light}\|,
\]
where $E_{{\rm gr}(x)}$ is the edge set of the greedy $(1+x\eps)$-spanner, and $E_{\light}$ is the edge set of a minimum-weight $(1+\eps)$-spanner for $S$.
\end{restatable}

 \paragraph{Main results.~}
Our main contribution is on the algorithmic front: we design a bicriteria $(c,\kappa)$-instance optimal spanner algorithm with $c = 2^{O(\log^*(d/\eps))}$ and $\kappa = O(1)$, thus resolving \Cref{ques:bi-approx} up to the exponential log-star term $2^{O(\log^*(d/\eps))}$, which is bounded by $O(\log^{(k)})(d/\eps)$ for \emph{any constant $k$};  here and throughout $\log^{(k)}$ denotes the $k$-iterated logarithm function (i.e.,  $\log^{(1)}(x) = \log(x)$ and  $\log^{(k)}(x) = \log(\log^{(k-1)}(x))$ for any integer $k\geq 2$). This result is obtained as a direct corollary of the following general theorem. 

\begin{theorem}[General tradeoff upper bound]\label{thm:tech}   Let $X\subset \mathbb{R}^d$ be any set of $n$ points, and $k\geq 1$ an integer. For any $\eps\in (0,1)$, there is an algorithm that returns an Euclidean $(1+2^{O(k)}\eps)$-spanner $H=(X,E)$ for $X$, such that
	\begin{equation*}
		\begin{split}
			|E| &= O\left(\log^{(k)}(\eps^{-1})+\log^{(k-1)}(d)\right) |E_\sparse| \qquad \text{and}\\
			\|E\| &= O\left(\log^{(k)}(\eps^{-1})+\log^{(k-1)}(d)\right) \|E_\light\|,
		\end{split}
	\end{equation*}
	where $E_\sparse, E_\light\in \binom{X}{2}$ are the edge sets of the optimal $(1+\eps)$-spanners of $X$ for sparsity and lightness, respectively.  In other words, our algorithm is bicriteria $(2^{O(k)}, \log^{(k)}(\eps^{-1}) +\log^{(k-1)}(d))$-instance optimal for both sparsity \emph{and} lightness.\\
	Furthermore, our algorithm can be implemented in $\eps^{-O(d)}n\log^2(n)$ time.
\end{theorem}

We now highlight the two extreme points on the tradeoff curve. First, by setting $k = O(1)$ (for any constant) in \Cref{thm:tech}, we obtain an $(O(1), \log^{(k)}(\eps^{-1}) + \log^{(k-1)}(d))$-instance optimal spanner algorithm. 
This result is also interesting from the perspective of \EMPH{existential optimality}. 
By scaling $\eps \leftarrow \eps/c$, we get an existentially optimal algorithm for both sparsity and lightness with ratio $\kappa = O(\log^{(k)}(\eps^{-1}) + \log^{(k-1)}(d))$ for any constant $k$. 
Recall that all known existentially optimal spanner algorithms have an \EMPH{exponential dependence on $d$} in the optimality ratio or the stretch blow-up, while we achieve a sublogarithmic dependence on $d$.
In particular, recall that the state-of-the-art existential optimality ratio for \EMPH{lightness} is $O_d(\log(\eps^{-1})) = d^{O(d)} \cdot \log(\eps^{-1})$; our result \EMPH{improves exponentially} both the dependence on $\eps^{-1}$ and the dependence on $d$. 

By setting $k = \log^{*}(d/\eps)$ in \Cref{thm:tech}, we obtain the following main corollary of the general tradeoff. 
Remarkably, our optimality ratio $\kappa$ in \Cref{thm:main} does not depend on the dimension $d$, while our stretch blow-up $c$ only depends \EMPH{sublogarithmically} on $d$.

\begin{corollary} [Almost instance optimality] \label{thm:main} 
There is an algorithm for constructing spanners that are $(c,\kappa)$-instance optimal for both sparsity \emph{and} lightness, where $c = 2^{O(\log^*(d/\eps))}$ and $\kappa = O(1)$.  Furthermore, our algorithm can be implemented in $\eps^{-O(d)}n\log^2(n)$ time.
\end{corollary}

\paragraph{Remark.~} An important feature of the classic greedy spanner algorithm is that it is existentially optimal with respect to both sparsity and lightness.
Our spanner algorithms (provided by
\Cref{thm:tech}
and \Cref{thm:main}) 
also have this feature, but in the stronger sense of instance optimality:
it \EMPH{simultaneously approximates both the sparsest spanner and the lightest spanner}. We note that for the same point set $X$ and stretch factor, the  sparsest spanner for $X$ and the lightest spanner could be completely different. In fact, the sparsest spanner could have a huge lightness while the lightest spanner could have a huge sparsity. Thus, a priori, it is unclear whether there exists an algorithm that is instance optimal with respect to \emph{both} sparsity and lightness.

\subsection{Related Work on Approximate Spanners in General Graphs}\label{subsec:related}

There is a long line of work on approximating minimum 
$t$-spanners of both undirected and directed graphs.
We briefly review the current best results in these regimes.
For undirected graphs, the greedy algorithm \cite{althofer1993sparse} provides a $t$-spanner with approximation ratio $n^{2/(t + 1)}$ for odd $t$ and $n^{2/t}$ for even $t$. Dinitz, Kortsarz, and Raz \cite{DKR15} showed that for any $t \geq 3$ and for any constant $\eps>0$, there is no polynomial-time algorithm approximating a $t$-spanner with ratio better than $2^{\log{n}^{1 - \epsilon}/t}$ assuming $NP \not \subseteq BPTIME(2^{\mathrm{polylog}(n)})$.  For $t = 2$, Kortsarz and Peleg \cite{KP94} (see also \cite{EP01}, \cite{DK11}) designed an algorithm with approximation ratio $O(\log n)$, matching the lower bound given by Kortsarz \cite{Kor01} (assuming $P \neq NP$). For $t = 3$, Berman, Bhattacharyya, Makarychev, Raskhodnikova and Yaroslavtsev \cite{BBM+11} achieved approximation ratio $\tilde{O}(n^{1/3})$. Dinitz and Zhang \cite{DZ16} obtained the same approximation ratio for stretch $t = 4$. For $t\geq 5$, the approximation ratio achieved by the greedy algorithm remains the state-of-the-art. 

For directed graphs, Dinitz and Krauthgamer \cite{DK11} gave an $\tilde{O}(\sqrt{n})$-approximation algorithm when $t = 3$. The approximation ratio for $t = 3$ was later improved to $O(n^{1/3}\log{n})$ for the unit-weight case~\cite{BBM+11}. For general $t$, there have been significant efforts to improve the approximation ratio. Bhattacharyya, Grigorescu, Jung, Raskhodnikova, and Woodruff \cite{BGJ+12} provided a $\tilde{O}(n^{1 - 1/t})$-approximation algorithm for directed graphs for $t > 2$. Berman, Raskhodnikova, and Ruan~\cite{BRR10} improved the approximation ratio to $\tilde{O}(t \cdot n^{1 - \frac{1}{\lceil t/2\rceil}})$. For $t > 3$, Dinitz and Krauthgamer \cite{DK11} achieved the approximation ratio of $\tilde{O}(n^{2/3})$. Berman, Bhattacharyya, Makarychev, Raskhodnikova and Yaroslavtsev \cite{BBM+11} later improved the approximation ratio to $O(\sqrt{n}\log{n})$.

\section{Technical Overview}

\paragraph{Our construction and sparsity analysis.} Let us begin with analyzing the sparsity of an arbitrary $(1+\delta)$-spanner $(X, E)$ against a sparsest $(1+\eps)$-spanner $(X, E_\sparse)$; one can think of the parameter $\delta$ as $O(\eps)$, though in practice it is a parameter used by our algorithm that starts at around $\eps$ and ultimately grows to $O(\eps)$ and even beyond $O(\eps)$. A basic approach is to {\em charge} the edges in $E$ to edges in $E_\sparse$. Our first idea is to design a \EMPH{fractional} charging scheme, rather than an integral one. 
Taking any edge $(s, t)\in E$, since $(X, E_\sparse)$ is a $(1+\epsilon)$-spanner, there exists a path $\pi$ in $(X, E_\sparse)$ between $s$ and $t$ such that $\|\pi\|\leq (1+\epsilon)\|st\|$. Thus, we can charge the edge $(s, t)\in E$ to each edge $e \in E_\sparse$ along the path $\pi$ with a \emph{fractional cost} of $\|e\|/ \|st\|$, meaning that the sum of fractional costs of these edges amounts to at least 1 (and also at most $1+\eps$), thus the total fractional costs of all edges in $E_\sparse$ is at least $|E|$. Therefore, if we could argue that the fractional costs received by any edge $e \in E_\sparse$ is  at most $\lambda$, that would directly imply that $|E|\leq \lambda\cdot |E_\sparse|$. 

However, such a charging scheme by itself is insufficient, unless it is accompanied with a ``good'' spanner. Alas, there are hard instances for which the known spanner constructions fail. 
As a (simplistic) example, consider a point set $X \in \mathbb{R}^2$
that contains two point sets $\{x_i\}_{0\leq i\leq 1/\sqrt{\epsilon}}, \{y_i\}_{0\leq i\leq 1/\sqrt{\epsilon}}$, where $x_i = (0, i\cdot \epsilon), y_i = (10, i\cdot\epsilon)$, as well as a pair of middle points $z = (3, \frac{\sqrt{\epsilon}}{2}), w=(3, \frac{\sqrt{\epsilon}}{2})$. See \Cref{bad-example} for an illustration. Our spanner could wastefully include edges $x_iy_j$ for all $i, j$, creating a bi-clique between $\{x_i\}_{0\leq i\leq 1/\sqrt{\epsilon}}$ and $\{y_i\}_{0\leq i\leq 1/\sqrt{\epsilon}}$; we will show later on that the greedy spanner exhibits this kind of wasteful behavior (on more subtle instances). On the other hand, the instance optimal $(1+\epsilon)$-spanner for this point set will only include edges $x_iz, zw, wy_i, \forall i$. 
Getting back to the aforementioned charging scheme, we see that all bi-clique edges would charge to the same middle edge $zw$, and so this charging scheme would inevitably fail. The main problem here is not the charging scheme, but rather the wasteful behavior of known spanner constructions.

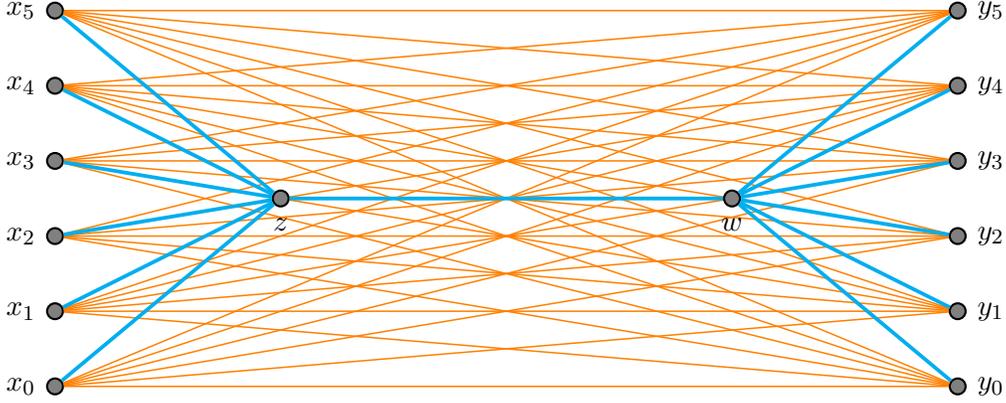
\begin{figure}
	\centering
	\begin{tikzpicture}[thick,scale=1]

	\foreach \y in {0,...,5}{
			\draw (0, \y) node(\y)[circle, draw, fill=black!50,inner sep=0pt, minimum width=6pt, label = {180 : {$x_{\y}$}}] {};
	}
	\foreach \y in {0,...,5}{
		\draw (12, \y) node(\y+6)[circle, draw, fill=black!50,inner sep=0pt, minimum width=6pt, , label = {0 : {$y_{\y}$}}] {};
	}
	
	\foreach \x in {0,...,5}{
		\foreach \y in {0,...,5}{
			\draw [line width = 0.2mm, color=orange] (\x) to (\y+6); 
		}
	}
	
	\draw (3, 2.5) node(12)[circle, draw, fill=black!50,inner sep=0pt, minimum width=6pt,label = {270 : {$z$}}] {};
	\draw (9, 2.5) node(13)[circle, draw, fill=black!50,inner sep=0pt, minimum width=6pt,label = {270 : {$w$}}] {};
	
	\draw [line width = 0.5mm, color=cyan] (12) to (13); 
	\foreach \y in {0,...,5}{
		\draw [line width = 0.5mm, color=cyan] (\y) to (12); 
		\draw [line width = 0.5mm, color=cyan] (\y+6) to (13); 
	}
\end{tikzpicture}
	\caption{In this example, a spanner might include the orange edges, while the optimal spanner takes the blue edges. Then the middle blue edge would receive a large amount of charges from the orange edges.}\label{bad-example}
\end{figure}

To outperform the known spanner constructions on hard instances, we design a new spanner algorithm, which deviates significantly from the known constructions.
Our construction is guided by a novel charging scheme;
we build on the basic idea of fractional charging as explained before, but in a much more nuanced way, 
which takes into account the angles formed by the edges, the locations of the endpoints, 
and other geometric parameters.  To keep this technical overview simple, we will not go into most details of the charging scheme; refer to \Cref{sparsityanal} for the full details.

Start with an arbitrary $(1+\delta)$-spanner. To improve sparsity (and lightness), our basic strategy is to look for \EMPH{helper edges}, such as edges $zw$  in the example above, add them to our spanner, and then {\em prune} unnecessary edges whose distances are already well-preserved by $zw$ together with other existing shorter edges in the spanner. More specifically, we will go over all edges $st\in E$ of the original spanner in a non-decreasing order of their lengths (i.e., weights), and look for helper edges $zw$ in the $(1+\epsilon)\|st\|$-ellipsoid around edge $st$ that satisfy the following two properties:
\begin{enumerate}
	\item $\|zw\|\geq \Omega(\|st\|)$; and
	\item both $z$ and $w$ are bounded away from the endpoints $s$ and $t$; that is, $\|sz\|, \|wt\|\geq \Omega(\|st\|)$.
\end{enumerate}

If such a helper edge $zw$ exists, then we add it to $E$. The key observation is that any other edge $s't'\in E$ with roughly the same length as $st$ and charging to common edges as $st$ {\em can now be pruned from $E$}.
Indeed, due to our charging scheme, if $s't'$ and $st$ charge to the same edge, it means that these two edges should be similar in a strong geometric sense, which allows us to reason that edge $zw$ could serve as a helper edge for $s't'$, as well. Consequently, by adding a single helper edge $zw$ while processing edge $st$, we are able to prune away from our spanner all 
edges that are ``similar'' to edge $st$, 
which leads to a significant saving. See \Cref{overview-prune} for an illustration.
So far, we have only discussed the pruning of edges $s't'$ that are of roughly the same length as $st$. We then generalize the above insight to show that for {\em every length scale} and for every edge $e$ charged by $st$, at most one edge from that length scale
may charge to edge $e$.
Since the fractional charge of edge $s't'$ to edge $e$, namely $\|e\| / \|s't'\|$, decays with the length of $s't'$,
we can bound the total contribution of all such edges $s't'$ to the fractional cost of edge $e$, over all length scales, by a geometric sum. 

\begin{figure}
	\centering
	\begin{tikzpicture}[thick,scale=1]
	\draw (0, 0) node(1)[circle, draw, fill=black!50,
	inner sep=0pt, minimum width=6pt, label = $s$] {};
	\draw (12, 0) node(2)[circle, draw, fill=black!50,
	inner sep=0pt, minimum width=6pt,label = $t$] {};
	
	\def\stellipse{(6, 0) ellipse (8 and 3)};
	\def\firstrec{(4.26, -4) rectangle (4.74, 4)};
	\def\secondrec{(7.26, -4) rectangle (7.74, 4)};

	\draw [black, dashed] \stellipse;
	\draw [line width = 0.5mm] (1) to (2);
	
	\draw (4.5, 2) node(3)[circle, draw, fill=black!50,
	inner sep=0pt, minimum width=6pt, label = $z$] {};
	\draw (7.5, 1.5) node(4)[circle, draw, fill=black!50,
	inner sep=0pt, minimum width=6pt,label = $w$] {};
	
	\draw [black, dashed] \stellipse;
	\draw [line width = 0.5mm, color=red] (3) to (4);
	
	\draw (5.7, 1) node(5)[circle, draw, fill=black!50,
	inner sep=0pt, minimum width=6pt] {};
	\draw (6.3, 0.9) node(6)[circle, draw, fill=black!50,
	inner sep=0pt, minimum width=6pt] {};
	
	\draw (6, 0.3) node[black, label={$e$}]{};
	
	\draw [black, dashed] \stellipse;
	\draw [line width = 0.5mm, color=cyan] (5) to (6);
	\draw [line width = 0.5mm, color=cyan, dashed] (1) to (5);
	\draw [line width = 0.5mm, color=cyan, dashed] (6) to (2);
	
	\draw (-2, 2) node(7)[circle, draw, fill=black!50,
	inner sep=0pt, minimum width=6pt, label = $s'$] {};
	\draw (10, -3.5) node(8)[circle, draw, fill=black!50,
	inner sep=0pt, minimum width=6pt,label = {0: {$t'$}}] {};
	
	\draw [black, dashed] \stellipse;
	\draw [line width = 0.5mm, dashed, color=orange] (7) to (8);
	\draw [line width = 0.5mm, dashed, color=red] (3) to (7);
	\draw [line width = 0.5mm, dashed, color=red] (4) to (8);
\end{tikzpicture}
	\caption{The blue path between $s$ and $t$ is a $(1+\eps)$-spanner path $\pi$ in $(X, E_\sparse)$.  
 The solid blue edge $e$ is an edge of $\pi$, so it receives a fractional charge of $\|e\| / \|st\|$ from $st$. When processing edge $st$, we find a helper edge $zw$ in the 
$(1+\epsilon)\|st\|$-ellipsoid around $st$. After adding the helper edge $zw$ to $E$, we are able to prune other spanner edges $s't'$ which are charging to $e$, as well, because we can now connect $s'$ and $t'$ using the red path in $(X, E)$.}\label{overview-prune} 
\end{figure}
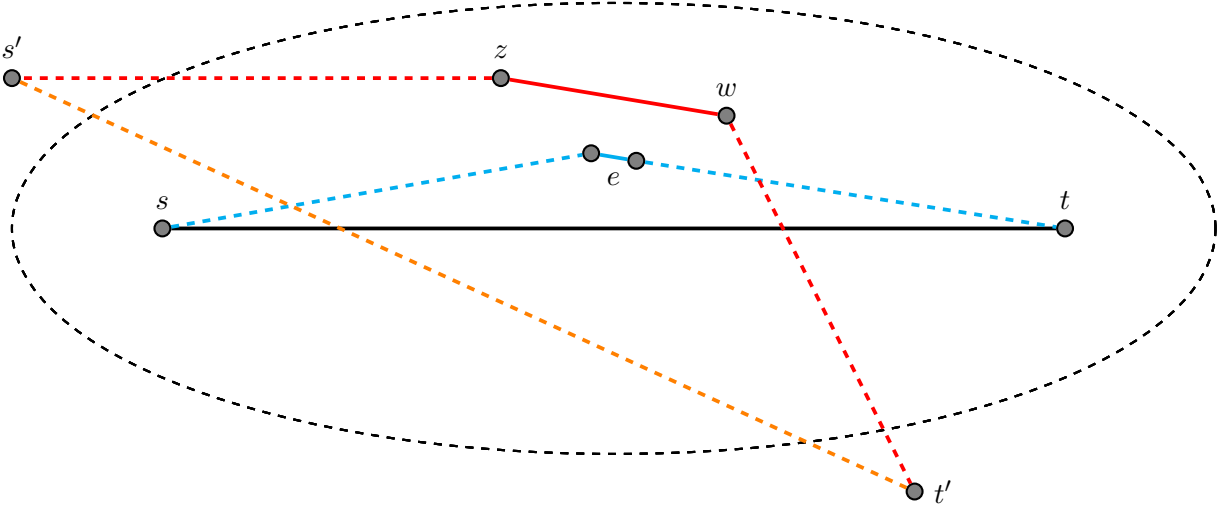

This strategy of using helper edges $zw$ as described above is effective when helper edges exist; {\em alas, they do not always exist}. 
When helper edges do not exist, we will show that there must exist an edge $e$ along the $(1+\eps)$-spanner path $\pi$ in $E_\sparse$ such that $\|e\|\geq \Omega(\|st\|)$; that is, $\pi$ is making at least one long stride at some point when connecting $s$ and $t$. In this case, we will charge the edge $st$ only to this single long edge $e$. However, since this long edge $e$ does not qualify as a helper edge (property~2 of a helper edge does not hold), we cannot apply the aforementioned argument, and it is now possible that many different edges $s't'\in E$ of length $\Theta(\|e\|)$ will charge to the same edge $e$. In this case, we will apply a more aggressive pruning, which greedily finds all edges $e$ whose addition to $E$ could prune sufficiently many edges $st$ with length $\Theta(\|e\|)$. Roughly speaking, this greedy pruning procedure works as follows: for any index $j\geq 0$, let $L_j\subseteq E$ be the set of edges whose lengths are in the range $[1.01^j, 1.01^{j+1})$. We start with an upper bound on the sparsity $\alpha \geq |E| / |E_\sparse|$, and for each length scale $[1.01^j, 1.01^{j+1})$, we will greedily add \EMPH{substitute} edges $e = zw\in L_j$ to the spanner, whose inclusion in $E$ could help in removing at least $\alpha /100$ existing edges from $L_j$ without increasing the stretch by too much. After such a pruning procedure, we are able to decrease our upper bound $\alpha$  on the ratio $|E| / |E_\sparse|$ by a constant factor, ignoring the new substitute edges, whose sparsity is constant. 
By working carefully, we demonstrate that this pruning procedure can be iterated until the sparsity upper bound has reduces from $\alpha$ to $O(\log \alpha)$, while the stretch increases only slightly.

Thus far we have presented the informal descriptions of two different {\em pruning phases}. Our spanner construction employs these two pruning phases on top of the original $(1+\delta)$-spanner $(X,E)$ that we start from, one after another, to reduce the sparsity upper bound 
from $\alpha$ to $O(\log \alpha)$ while increasing the stretch from $1+ \delta$ to $1 + O(\delta)$.
We remark that there is a delicate interplay between the two pruning phases, which requires our algorithm to pay special care to additional subtleties, and in particular to explicitly distinguish between two types of edges; aiming for brevity, we will not discuss such subtleties in this high-level overview, but the details appear in \Cref{thealg} (refer to \Cref{edgeclass} in particular for the definition of the two types of edges).  
Finally, to achieve the stretch and sparsity bounds claimed in  \Cref{thm:main}, we apply the two pruning phases iteratively until the sparsity upper bound reduces to a constant. 
The full details of the construction, the stretch analyais and the sparsity analysis appear in Sections~\ref{thealg}, \ref{sec:stretch} and~\ref{sparsityanal}, respectively.

Our spanner algorithm and the underlying charging scheme heavily rely on Euclidean geometry. However, surprisingly perhaps, both of them are completely unaffected by the \EMPH{dimension} of the input point set. (Our algorithm runs in low polynomial time in any dimension; however, the fast implementation of our algorithm described below does depend on the dimension.) The dependencies on the dimension 
in \Cref{thm:main} (ignoring the fast implementation running time)
stem only from the sparsity and lightness dependencies on the dimension in the original spanner $(X,E)$.

\paragraph{Lightness analysis.} It turns out that the same spanner construction algorithm described above also approximates the instance-optimal lightness. More specifically, relying on the fact that the sparsity and lightness bounds of the initial (e.g., greedy) spanner $(X,E)$ are both $\eps^{-\Theta(d)}$, we can prove that the instance-optimal lightness bound in the resulting spanner is (basically) the same as the sparsity bound.
While our lightness analysis builds on the sparsity analysis,
it has to drill quite a bit deeper in order to overcome another significant technical challenge, which we highlight next.
The full details are in \Cref{sec:light}.
When comparing our $(1+\delta)$-spanner $(X, E)$ against the optimal (now in terms of weight) spanner $(X, E_\light)$, for any edge $st\in E$, let $\pi$ denote a $(1+\eps)$-spanner path in $(X, E_\light)$ between $s$ and $t$. 
While in the sparsity analysis we charged the edge $st \in E$ to each edge $e$ along $\pi$ with a fractional cost of $\frac{\|e\|}{\|st\|}$, so that the total fractional costs of these edges is at least 1, such a charging is not suitable for the lightness analysis.
Instead, we shall charge edge $st$ 
to each edge  $e$ along $\pi$ with a cost of $\|e\|$, so that the total costs of these edges is at least $\|st\|$, 
thus the total costs of the edges in $E_\light$ is at least $\|E\|$. Therefore, if we could argue that the total costs received by any edge $e \in E_\light$ is  at most $\lambda$, that would directly imply that $\|E\|\leq \lambda\cdot \|E_\light\|$.

By adapting our techniques from the sparsity analysis, 
we can bound by $O(1)$ the number of times any edge $e \in E_\light$ gets charged from the same edge set $L_j$ (i.e., the edges whose lengths are in the range $[1.01^j, 1.01^{j+1})$). However, the total amount of charges over all length scales (or indices $j$) increases by a factor of $O(\log \Phi)$, where $\Phi$ is the spread of the point set. Using a standard trick, this blowup can be reduced to $O(\log n)$, but this is still insufficient for achieving instance-optimality.
This is how lightness is different from sparsity, as charges of sparsity are decreasing geometrically as $j$ increases, while charges of lightness remain the same for all scales $j$.

To bypass the logarithmic blowup, let us go over all edges in $E$ in a non-decreasing order of lengths and allocate their charges to edges in $E_\light$ in an ``adaptive'' manner. Consider any edge $st \in E$ as well as a $(1+\eps)$-spanner path $\pi$ between $s$ and $t$ in $(X, E_\light)$. If most edges of $\pi$ have only been charged a small number of times by edges in $E$ shorter than $st$, then we can safely charge $st$ proportionally only to those edges (we charge any such edge $e$ a cost of say $10 \|e\|$). Otherwise, there exists a set of edges $z_1w_1, z_2w_2, \ldots, z_lw_l$ on $\pi$ such that:
\begin{itemize}
	\item $\sum_{i=1}^l\|z_iw_i\|\geq \Omega(\|st\|)$; and
	\item each edge $z_iw_i$ has been charged multiple times already.
\end{itemize}
In this case, we can prove that for each such edge $z_iw_i$ of $\pi$, there exists an edge $a_ib_i\in E$, which charges to $z_iw_i$ and is much longer than $\|z_iw_i\|$ yet much shorter than $\|st\|$. Using this key insight carefully, we prove that existing edges in $E$ can be used to ``stitch together'' all the edges $a_1b_1, a_2b_2, \ldots, a_lb_l$  into a good spanner path (of small stretch) between $s$ and $t$ in $(X, E)$, making the direct edge $st$ unnecessary in $E$, and so we can show that $st$ could not have been added to $E$. See \Cref{overview-stitch} for an illustration.

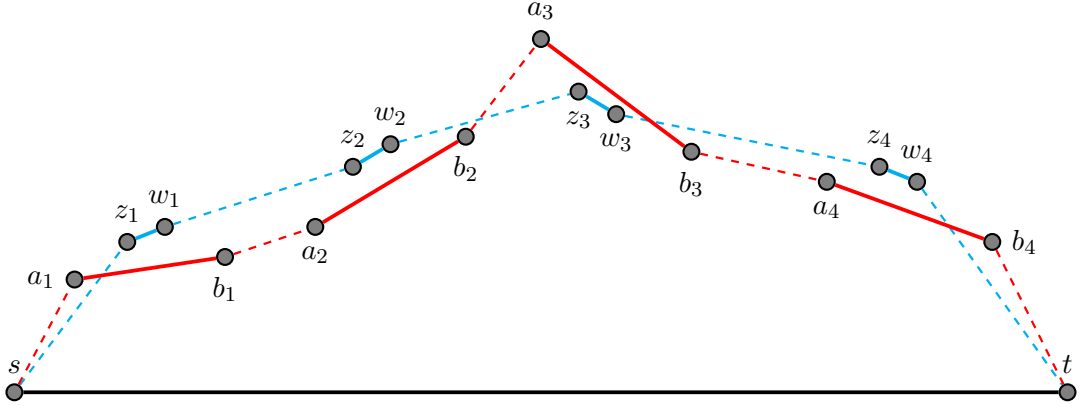
\begin{figure}
	\centering
	\begin{tikzpicture}[thick,scale=1]
	\draw (0, 0) node(1)[circle, draw, fill=black!50,
	inner sep=0pt, minimum width=6pt, label = $s$] {};
	\draw (14, 0) node(2)[circle, draw, fill=black!50,
	inner sep=0pt, minimum width=6pt,label = $t$] {};
	
	\draw [line width = 0.5mm] (1) to (2);
	
	\draw (1.5, 2) node(3)[circle, draw, fill=black!50, inner sep=0pt, minimum width=6pt, label = $z_1$] {};
	\draw (2, 2.2) node(4)[circle, draw, fill=black!50, inner sep=0pt, minimum width=6pt, label = $w_1$] {};
	
	\draw (4.5, 3) node(5)[circle, draw, fill=black!50, inner sep=0pt, minimum width=6pt, label = $z_2$] {};
	\draw (5, 3.3) node(6)[circle, draw, fill=black!50, inner sep=0pt, minimum width=6pt, label = $w_2$] {};
	
	\draw (7.5, 4) node(7)[circle, draw, fill=black!50, inner sep=0pt, minimum width=6pt, label = {-90: {$z_3$}}] {};
	\draw (8, 3.7) node(8)[circle, draw, fill=black!50, inner sep=0pt, minimum width=6pt, label = {-90: {$w_3$}}] {};
	
	\draw (11.5, 3) node(9)[circle, draw, fill=black!50, inner sep=0pt, minimum width=6pt, label = $z_4$] {};
	\draw (12, 2.8) node(10)[circle, draw, fill=black!50, inner sep=0pt, minimum width=6pt, label = $w_4$] {};
	
	\draw [line width = 0.5mm, color=cyan] (3) to (4);
	\draw [line width = 0.5mm, color=cyan] (5) to (6);
	\draw [line width = 0.5mm, color=cyan] (7) to (8);
	\draw [line width = 0.5mm, color=cyan] (9) to (10);
	\draw [line width = 0.3mm, color=cyan, dashed] (1) to (3);
	\draw [line width = 0.3mm, color=cyan, dashed] (4) to (5);
	\draw [line width = 0.3mm, color=cyan, dashed] (6) to (7);
	\draw [line width = 0.3mm, color=cyan, dashed] (8) to (9);
	\draw [line width = 0.3mm, color=cyan, dashed] (10) to (2);

	\draw (0.8, 1.5) node(11)[circle, draw, fill=black!50, inner sep=0pt, minimum width=6pt, label = {180: {$a_1$}}] {};
	\draw (2.8, 1.8) node(12)[circle, draw, fill=black!50, inner sep=0pt, minimum width=6pt, label = {-90: {$b_1$}}] {};
	
	\draw (4, 2.2) node(13)[circle, draw, fill=black!50, inner sep=0pt, minimum width=6pt, label = {-90: {$a_2$}}] {};
	\draw (6, 3.4) node(14)[circle, draw, fill=black!50, inner sep=0pt, minimum width=6pt, label = {-90: {$b_2$}}] {};
	
	\draw (7, 4.7) node(15)[circle, draw, fill=black!50, inner sep=0pt, minimum width=6pt, label = $a_3$] {};
	\draw (9, 3.2) node(16)[circle, draw, fill=black!50, inner sep=0pt, minimum width=6pt, label = {-90: {$b_3$}}] {};
	
	\draw (10.8, 2.8) node(17)[circle, draw, fill=black!50, inner sep=0pt, minimum width=6pt, label = {-90: {$a_4$}}] {};
	\draw (13, 2) node(18)[circle, draw, fill=black!50, inner sep=0pt, minimum width=6pt, label = {0:{$b_4$}}] {};
	
	\draw [line width = 0.5mm, color=red] (11) to (12);
	\draw [line width = 0.5mm, color=red] (13) to (14);
	\draw [line width = 0.5mm, color=red] (15) to (16);
	\draw [line width = 0.5mm, color=red] (17) to (18);
	\draw [line width = 0.3mm, color=red, dashed] (1) to (11);
	\draw [line width = 0.3mm, color=red, dashed] (12) to (13);
	\draw [line width = 0.3mm, color=red, dashed] (14) to (15);
	\draw [line width = 0.3mm, color=red, dashed] (16) to (17);
	\draw [line width = 0.3mm, color=red, dashed] (18) to (2);

\end{tikzpicture}
	\caption{The blue path represents a $(1+\eps)$-spanner path $\pi$ between $s$ and $t$ in $(X, E_\light)$, and the blue solid edges are the ones already receiving heavy charges from shorter edges in $E$. Then, for each blue solid edge $z_iw_i$, we can find a longer red solid edge $a_ib_i\in E$ that charged to $z_iw_i$. Thus, we can stitch together all these red solid edges with existing edges in $E$ to create a good spanner path between $s$ and $t$, and so edge $st$ does not need to stay in $E$ anymore.}\label{overview-stitch}
\end{figure}

\paragraph{Fast implementation.} Next, we highlight three main technical difficulties behind achieving a near-linear time 
implementation of our spanner construction (in low-dimensional spaces).
Recall that our implementation takes $\eps^{-O(d)}n\log^2(n)$ time; 
see \Cref{sec:fast}
for the full details.

The first difficulty is to efficiently locate all the substitute edges. (Recall that the substitute edges are needed if helper edges do not exist, and they are used for a more aggressive pruning of edges.) 
A straightforward implementation would enumerate all possible choices of a substitute edge $xy$, and count how many edges any such edge $xy$ can prune from the current spanner if it is added to the spanner. Such an implementation would require quadratic time to find a substitute edge each time. Instead, we will build a {\em hierarchy of nets} (as in \cite{CGMZ16}), and for each edge $st\in E$, its substitute edge will be restricted to {\em $\eps \|st\|$-net points} nearby. Using the standard packing bound
in $\real^d$, the total number of possible substitute edges for $st$ is bounded by $\eps^{-O(d)}$. Then, we can initialize and maintain an efficient data structure in  $\eps^{-O(d)}n\log(n)$  total time, which counts 
the number of edges assigned to each possible substitute edge; using this data structure, we can repeatedly select the best substitute edges in $\eps^{-O(d)}n\log(n)$  total time.

The second difficulty stems from the  need to find a helper edge $zw$ for any existing edge $st\in E$. A straightforward implementation of this task takes quadratic time by checking all possible helper edges $zw$. To narrow down our search space, we will again build an {$\epsilon\|st\|$-net} and only look for {net points} around the segment $st$ in search of a helper edge $zw$. Using the packing bound, we can show that the total number of candidate edges is at most $\eps^{-O(d)}$.

As the third difficulty, we need to go over all edges in $st\in E$ in a non-decreasing order of lengths and, for each one, decide if we should keep it in the new spanner. Here, we need to quickly determine if the distance between the two endpoints $s$ and $t$ in the new spanner is already well-preserved, which entails an (approximate) shortest path computation. If we directly apply Dijkstra's algorithm to compute shortest paths, then the total time for such computations over all edges $st\in E$ would be quadratic. To reduce the runtime, we will adopt the approach from \cite{DN97} used for constructing an {\em approximate} version of the greedy spanner. Roughly speaking, to compute an (approximate) shortest path between $s$ and $t$, we build a {\em cluster graph} that contracts clusters of radius $\eps\|st\|$, and then look at the cluster centers near $s$ and $t$; importantly, we only apply Dijkstra's algorithm locally on these cluster centers. Relying again on the standard packing bound in $\mathbb{R}^d$, we can prove that the number of cluster centers around $s$ and $t$ is also at most $\eps^{-O(d)}$, which allows us to reduce the running time of an (approximate) shortest path computation to $\eps^{-O(d)}$.

\paragraph{Lower bounds for greedy spanners.} 
Our hard instances for the poor performance of the greedy algorithm 
build on the dramatic difference between the sparsity of Euclidean spanners with or without Steiner points~\cite{bhore2022euclidean,le2022truly}: $O(\eps^{(-d+1)/2})$ versus $O(\eps^{-d+1})$. We modify the previously known worst-case instances for Euclidean Steiner spanners: Two large sets uniformly distributed along two parallel lines, and a small set of middle points  between the two lines (a schematic example is in \Cref{overview-lower-bnd}; see \Cref{fig:LB,fig:LB+} for more accurate illustrations). In previous work~\cite{bhore2022euclidean,le2022truly}, the middle points played the role of ``Steiner points'' (recall that a Steiner  spanner does not have to maintain the stretch factor for paths to and from Steiner points). We include the middle points in the input. Importantly, the optimal (i.e., sparsest) spanner can take advantage of the middle points as convenient ``via'' points between the two large point sets (similarly to Steiner spanners), but the greedy algorithm misses their potential and is forced to add the complete bipartite graph between the two large point sets on the two lines; see \Cref{thm:sparsityLB,thm:sparsityLB+} for details. The main challenge is the analysis of the greedy spanner. We heavily use geometry in the design and analysis of the lower bound instances to maintain a tight approximation of the stretch between point pairs during the greedy process. 

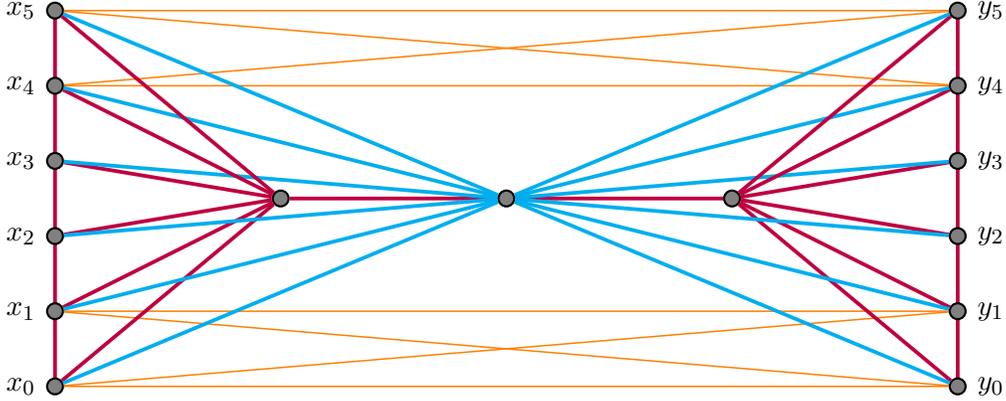
\begin{figure}
	\centering
	\begin{tikzpicture}[thick,scale=1]
         \draw [line width = 0.5mm, color=purple] (0,0) to (0,5); 
        \draw [line width = 0.5mm, color=purple] (12,0) to (12,5);

	\foreach \y in {0,...,5}{
			\draw (0, \y) node(\y)[circle, draw, fill=black!50,inner sep=0pt, minimum width=6pt, label = {180 : {$x_{\y}$}}] {};
	}
	\foreach \y in {0,...,5}{
		\draw (12, \y) node(\y+6)[circle, draw, fill=black!50,inner sep=0pt, minimum width=6pt, , label = {0 : {$y_{\y}$}}] {};
	}
	
	\foreach \x in {0,1}{
		\foreach \y in {0,1}{
			\draw [line width = 0.2mm, color=orange] (\x) to (\y+6); 
		}
	}

         \foreach \x in {4,5}{
		\foreach \y in {4,5}{
			\draw [line width = 0.2mm, color=orange] (\x) to (\y+6); 
		}
	}
	
	\draw (3, 2.5) node(12)[circle, draw, fill=black!50,inner sep=0pt, minimum width=6pt] {};
	\draw (9, 2.5) node(13)[circle, draw, fill=black!50,inner sep=0pt, minimum width=6pt] {};

        \draw (6, 2.5) node(14) [circle, draw, fill=black!50,inner sep=0pt, minimum width=6pt] {};
 
	\draw [line width = 0.5mm, color=purple] (12) to (14);
        \draw [line width = 0.5mm, color=purple] (13) to (14);   
	\foreach \y in {0,...,5}{
		\draw [line width = 0.5mm, color=purple] (\y) to (12); 
		\draw [line width = 0.5mm, color=purple] (\y+6) to (13);
            \draw [line width = 0.5mm, color=cyan] (\y) to (14);
            \draw [line width = 0.5mm, color=cyan] (\y+6) to (14);
	}
\end{tikzpicture}
	\caption{The red edges are in both the optimal and the greedy spanner; the blue edges are only in the optimal spanner, and the orange edges are only in the greedy spanner. This point set fools the greedy algorithm not to add the blue edges, so that it later has to add all orange edges,  which form two bi-cliques, incurring quadratic sparsity.}\label{overview-lower-bnd}
\end{figure}

Our lower bounds for \emph{sparsity} use edges of comparable weight, and immediately give the same lower bounds for lightness. However, we can obtain stronger lower bounds for \emph{lightness} with a surprisingly simple point set: Uniformly distributed points along a circular arc (not the entire circle!). Due to the uniform distribution, we can easily analyze the greedy algorithm with stretch $(1+x\eps)$ for any $x$, $1\leq x\leq O(\eps^{1/2})$: The greedy algorithm includes a path along the circular arc; and then at a certain threshold, it adds a large number of heavy edges of equal weight. However, an optimum (i.e., lightest) spanner can use a much smaller number of ``shortcut'' edges instead (similarly to the \emph{helper} edges in our upper bound construction). In our basic example (\Cref{thm:weightLB}), a single shortcut edge of $G_{\light}$ trades off against $O(\eps^{-1})$ almost diametric edges of $G_{\rm gr}$. For a lower bound for $(1+x\eps)$-spanners (\Cref{thm:weightLB+}), we use a hierarchy of shortcut edges.

\section{Preliminaries}
Let $\kappa = 10^4$ be a large but fixed constant independent of $n$, $d$, and $\epsilon$; and assume that $\epsilon>0$ is relatively small compared to $\kappa^{-1}$ and $d^{-1}$. More specifically, we assume the following relation: 
\begin{equation} \label{eq:relationeps}
\epsilon\cdot 2^{O(\log^*(d / \epsilon))} < \kappa^{-5}.
\end{equation}

For any vector $e\in\mathbb{R}^{d}$, let $\|e\|$ denote the the Euclidean length (i.e., $\ell_2$-norm) of $e$. For a graph $H$, let $V(H)$ and $E(H)$, resp., denote the vertex set and the edge set of $H$. For any set of edges $E\subseteq \binom{X}{2}$, let $\|E\|$ be the total length of the edges in $E$; and the weight of a graph $H=(X,E)$ is defined as $\|H\| = \|E\|$.

For any pair of points $s, t\in \mathbb{R}^{d}$, let $st$ denote the segment that connects $s$ and $t$, and let $\overrightarrow{st}$ or $t-s$ be the vector directed from $s$ to $t$. We will sometimes use the notation $s \rightsquigarrow t$ for the shortest (spanner) path between $s$ and $t$ in a graph $H$ when $H$ is clear from context.

For any pair of vectors $e_1, e_2\in \mathbb{R}^d$, their angle $\angle(e_1, e_2)$ is defined as:
$$\angle(e_1, e_2)\overset{\text{def}}{=}\arccos\brac{\frac{|e_1\cdot e_2|}{\|e_1\|\cdot\|e_2\|}}$$

For a polygonal path $\pi$ and two vertices $p, q$ of $\pi$, let $\pi[p, q]$ be the sub-path of $\pi$ between $p$ and $q$. For an edge $e$ and a line $st$, let $\proj_{st}(e)$ be the orthogonal projection of $e$ onto the line $st$. We generalize \cite[Lemma~4]{bhore2022euclidean}, originally stated for a polygonal path $\pi$, to the setting where $\pi$ is a sequence of edges; we include the proof for completeness. 

\begin{lemma}[Lemma~4 in \cite{bhore2022euclidean}]\label{angle-bound}
	Consider a sequence of edges $\pi$ (not necessarily a polygonal path) whose projection on a line is segment $ab$. Let $E(\pi, ab, \theta)$ be the set of edges $e\in \pi$ such that $\angle(e, ab)\leq \theta$. If $\|\pi\| \leq (1+\epsilon)\|ab\| $, then $\|E(\pi, ab, 2\sqrt{\epsilon})\| \geq 0.5\|ab\|$. 
\end{lemma}
\begin{proof}
	Assume otherwise that $\|E(\pi, ab, 2\sqrt{\epsilon})\| < 0.5\|ab\|$. Then, since the projection of $\pi$ on line $ab$ is equal to the segment $ab$, we have:
	$$\sum_{e\in \pi}\|\proj_{ab}(e)\| \geq \|ab\|$$
	which implies
	$$\begin{aligned}
		\sum_{e\in \pi\setminus E(\pi, ab, 2\sqrt{\epsilon})}\|\proj_{ab}(e)\| &\geq \|ab\| - \sum_{e\in E(\pi, ab, 2\sqrt{\epsilon})}\|\proj_{ab}(e)\|\\
		&\geq \|ab\| - \sum_{e\in E(\pi, ab, 2\sqrt{\epsilon})}\|e\|\\
		&= \|ab\| - \|E(\pi, ab, 2\sqrt{\epsilon})\|.
	\end{aligned}$$
	Recall that for every edge $e\in \pi \setminus E(\pi, ab, 2\sqrt{\epsilon})$, we have $\angle(e, ab)\geq 2\sqrt{\epsilon}$. Using Taylor estimate $1 / \cos(x)\geq 1+x^2/2$, and thus for any edge $e\in \pi\setminus E(\pi, ab, 2\sqrt{\epsilon})$, we have
	$$\|e\|\geq \frac{\|\proj_{ab}(e)\|}{\cos(2\sqrt{\epsilon})}\geq \|\proj_{ab}(e)\|\cdot (1+2\epsilon).$$
	Combined with the previous inequality, we obtain
\begin{align*}
		\|\pi\| &= \sum_{e\in E(\pi, ab, 2\sqrt{\epsilon})}\|e\| + \sum_{e\in \pi\setminus E(\pi, ab, 2\sqrt{\epsilon})}\|e\|\\
		&\geq \|E(\pi, ab, 2\sqrt{\epsilon})\| + (1+2\epsilon)\cdot \brac{\|ab\| - \|E(\pi, ab, 2\sqrt{\epsilon})\|}\\
		&\geq (1+2\epsilon)\|ab\| - 2\epsilon \|E(\pi, ab, 2\sqrt{\epsilon})\|\\
		&> (1+\epsilon)\|ab\|,
\end{align*}
	which is a contradiction.
\end{proof}

Given a finite point set $X\subset \mathbb{R}^d$, let $G_\sparse = (X, E_\sparse)$ and $G_\light = (X, E_\light)$ be the $(1+\epsilon)$-stretch Euclidean spanners of $X$ with the minimum number of edges and the minimum weight, respectively. We will use the following statement, which bounds the sparsity and lightness  of the greedy spanner; although more precise bounds are known (which explicate the constants in the $O$-notation), as mentioned, the   bounds in the following statement will suffice for our purposes.
\begin{lemma}[\cite{ChandraDNS95, RS98, NS07C}]
\label{greedy}
	If $H = (X,E)$ is the greedy $(1+\epsilon)$-spanner for $X\subset \mathbb{R}^d$, then $|E|\leq n\cdot \epsilon^{-O(d)}$, and $\|E\|\leq \|\MST(X)\|\cdot \epsilon^{-O(d)}$, where $\MST(X)$ is a Euclidean minimum spanning tree of $X$.
\end{lemma}

\section{Instance-Optimal Euclidean Spanners}

\subsection{A Greedy-Pruning Algorithm} \label{thealg}

We are given a set $X\subset \mathbb{R}^d$ of $n$ points and
a sufficiently small $\eps>0$ satisfying \Cref{eq:relationeps}.
Initially, let $H=(X,E)$ be an arbitrary $(1+\epsilon)$-spanner on $X$ (for example, the greedy spanner). According to \Cref{greedy}, we may assume that $H$ has $n \cdot \epsilon^{-O(d)}$ edges and weight at most $\|\MST(X)\|\cdot \epsilon^{-O(d)}$. Next, we will successively modify $H$ to improve its  sparsity and lightness while keeping the stretch under control. In each iteration, we will prune some of the edges in $H$ and add new edges. In each iteration, we maintain an upper bound on the current stretch of $H$ and its approximation ratio  compared to optimal $(1+\eps)$-spanners w.r.t.\ sparsity and lightness. Specifically, we assume that $H$ is a $(1+\delta)$-spanner for some $\delta \geq \epsilon$; the assumption $\delta \geq \epsilon$ (or $\delta = \Omega(\eps)$) will be crucial in the analysis of stretch and lightness later on. Furthermore, we assume that $\alpha \geq \max\left\{|E| / |E_\sparse|,  \|E\|/ \|E_\light\|\right\}$; initially when $E$ is a greedy spanner, we have $\delta = \epsilon$ and $\alpha = \epsilon^{-O(d)}$. The total number of  iterations will be $\log^*(d / \epsilon) + O(1)$.

One iteration consists of two phases which construct two spanners: First $H_1 = (X, E_1)$ and then $H_2 = (X, E_2)$. At the end of an iteration, we will reassign $E\leftarrow E_2$. During the process, edges in $E_1$ and $E_2$ that come from $E$ will be called \emph{old} edges, and all other edges in $E_1$ and $E_2$ will be called \emph{new} edges.

\paragraph{Classification of edges in $E$.} As a preliminary step, the algorithm will distinguish between two types of edges in $E$. 
For an edge $st\in E$ and a point $x\in X$, let $\proj(x)$ denote the orthogonal projection of $x$ on the straight line passing through $s$ and $t$. Let $\Gamma_{s, t}$ be the ellipsoid defined by
$$\Gamma_{s, t} = \left\{x\in \mathbb{R}^{d} : \|sx\|  + \|xt\|  \leq (1+\epsilon)\|st\|  \right\}$$
with foci $s$ and $t$.
We define the following two regions (see \Cref{ellipsoid} for an illustration): 
$$A_{s, t} \overset{\text{def}}{=} \Gamma_{s, t}\cap\left\{x\in \mathbb{R}^{d} : \frac{\|s - \proj_{st}(x)\| }{\|st\| } \in \left[\frac{3}{8} - \frac{1}{50}, \frac{3}{8} + \frac{1}{50}\right], \frac{\|t - \proj_{st}(x)\| }{\|st\| } \in \left[\frac{5}{8} - \frac{1}{50}, \frac{5}{8} + \frac{1}{50}\right]\right\},$$
$$B_{s, t} \overset{\text{def}}{=} \Gamma_{s, t}\cap\left\{x\in \mathbb{R}^{d} : \frac{\|s - \proj_{st}(x)\| }{\|st\| } \in \left[\frac{5}{8} - \frac{1}{50}, \frac{5}{8} + \frac{1}{50}\right], \frac{\|t - \proj_{st}(x)\| }{\|st\| } \in \left[\frac{3}{8} - \frac{1}{50}, \frac{3}{8} + \frac{1}{50}\right]\right\}.$$
\begin{definition} \label{edgeclass}
If $A_{s, t}\cap X$ or $B_{s, t}\cap X$ is empty for an edge $st\in E$, then $s t$ is called a \EMPH{type-(\romannumeral1)} edge; otherwise, it is called a \EMPH{type-(\romannumeral2)} edge. Let $E^{(\romannumeral1)}$ and $E^{(\romannumeral2)}$, resp., denote the set of type-(\romannumeral1) and type-(\romannumeral2) edges in $E$.
 \end{definition}

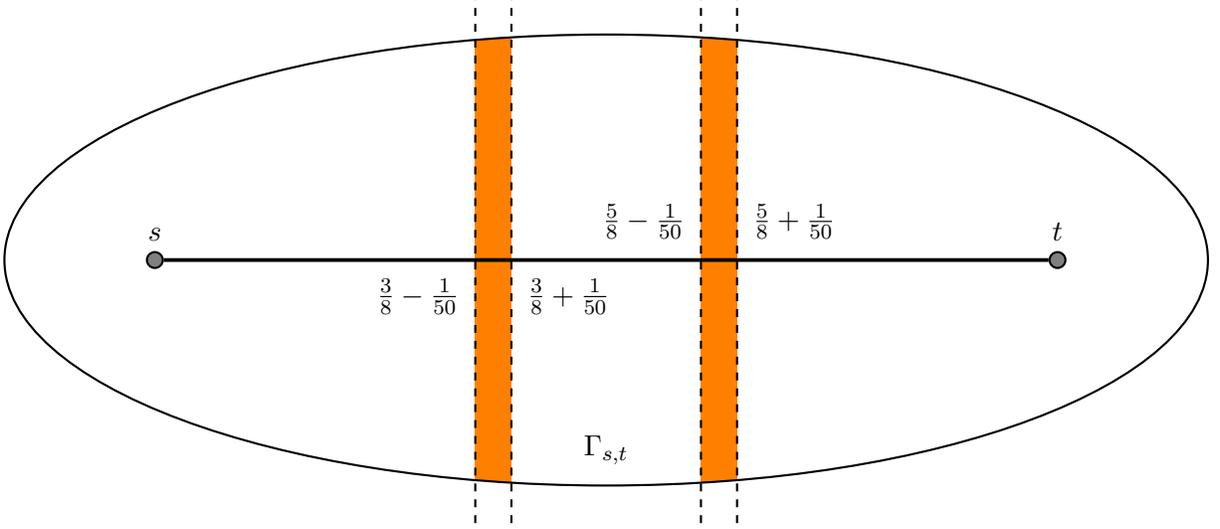
\begin{figure}
	\centering
	\begin{tikzpicture}[thick,scale=1]
	\draw (0, 0) node(1)[circle, draw, fill=black!50,
	inner sep=0pt, minimum width=6pt, label = $s$] {};
	\draw (12, 0) node(2)[circle, draw, fill=black!50,
	inner sep=0pt, minimum width=6pt,label = $t$] {};
	
	\def\stellipse{(6, 0) ellipse (8 and 3)};
	\def\firstrec{(4.26, -4) rectangle (4.74, 4)};
	\def\secondrec{(7.26, -4) rectangle (7.74, 4)};
		
	\begin{scope}
		\clip \stellipse;
		\fill[orange] \firstrec;
	\end{scope}
	
	\begin{scope}
		\clip \stellipse;
		\fill[orange] \secondrec;
	\end{scope}
	
	\draw [black] \stellipse;
	\draw (6, -3) node[black, label={$\Gamma_{s, t}$}]{};
	
	\draw [line width = 0.5mm] (1) to (2); 
	\draw [dashed] (4.26, -3.5) to (4.26, 3.5);
	\draw [dashed] (4.74, -3.5) to (4.74, 3.5);
	\draw [dashed] (7.26, -3.5) to (7.26, 3.5);
	\draw [dashed] (7.74, -3.5) to (7.74, 3.5);
	
	\draw (3.5, -1) node[black, label={$\frac{3}{8}-\frac{1}{50}$}]{};
	\draw (5.5, -1) node[black, label={$\frac{3}{8}+\frac{1}{50}$}]{};
	
	\draw (6.5, 0) node[black, label={$\frac{5}{8}-\frac{1}{50}$}]{};
	\draw (8.5, 0) node[black, label={$\frac{5}{8}+\frac{1}{50}$}]{};
	
\end{tikzpicture}
	\caption{The two sets $A_{s, t}$ and $B_{s, t}$ are drawn as two orange regions in the ellipsoid around $st$. For simplicity, in this figure we assume that $s$ and $t$ lie on the $x$-axis with coordinate 0 and 1, respectively (and all the numbers shown designate $x$-coordinates).}\label{ellipsoid}
\end{figure}

\paragraph{First pruning phase.} Initially, we set $E_1 \leftarrow E$ and $\beta \overset{\text{def}}{=} 1.01$. Without loss of generality, assume that all pairwise distances in $X$ are at least $1$. Next, we perform $O(\log\alpha)$ sub-iterations of pruning. In the $i$-th sub-iteration, go over all indices $j = 0, 1, 2, \ldots$. For each index $j$, let $L_j\subseteq E$ be the set of old edges whose lengths are in the range $[\beta^j, \beta^{j+1})$. In the $j$-th iteration, enumerate all pairs of vertices $\{x, y\}\in \binom{X}{2}$ such that $\|xy\| \geq \beta^j / 25$, and define a set of type-(\romannumeral1) edges 
$$P_{x,y} = \left\{st\in L_j\cap E_1\cap E^{(\romannumeral1)}\ : \|sx\|+ \|xy\|  + \|yt\|\leq (1+\epsilon)\cdot\|st\|  \right\}.$$
If $|P_{x, y}| \geq \frac{\alpha}{2^i\kappa}$, then add $xy$ to $E_1$ as a new edge, and remove all type-(\romannumeral1) edges in $P_{x, y}$ from $E_1$.

\paragraph{Second pruning phase.} To construct $E_2$, initially set $E_2\leftarrow E_1\setminus E^{(\romannumeral2)}$. Then, enumerate all type-(\romannumeral2) edges of $E_1$ in an increasing order of edge weights. For each such edge $st\in E_1\cap E^{(\romannumeral2)}$, first check whether the stretch between $s$ and $t$ is approximately preserved in the graph $(X,E_2)$; that is, whether 
$$\dist_{H_2}(s, t)\leq (1+\kappa^2\delta)\cdot \|st\|.$$
If so, move on to the next old edge in $E_1$. Otherwise, since $s t$ is type-(\romannumeral2), there exists a pair of vertices $a\in A_{s, t}, b\in B_{s, t}$. Then add one such edge $ab$ to $E_2$ as a new edge, and $s t$ to $E_2$ as an old edge; we call $a b$ the \emph{helper} edge associated with $s t$. After these procedures, move on to the next type-(\romannumeral2) edge in $E_1$.

\paragraph{Updating the parameters.} After the two pruning phases, reassign $E\leftarrow E_2$. Before moving on to the next iteration, we need to update the upper bound on the stretch $\delta$ and approximation ratio $\alpha$. Specifically, update $\delta \leftarrow \Delta(\kappa, \delta) \overset{\text{def}}{=} (1+\delta)\cdot(1+\kappa\delta)\cdot (1+\kappa^2\delta) - 1$, and $\alpha\leftarrow O(\log\alpha)$. The whole algorithm is summarized in Algorithm \ref{greedy-prune}.

\begin{algorithm}\label{greedy-prune}
	\caption{$\mathsf{GreedyPrune}(X, \epsilon)$}
	$\kappa\leftarrow 5000, \beta\leftarrow 1.01, \alpha\leftarrow (1/\epsilon)^{O(d)}, \delta\leftarrow \epsilon$\;
	let $H = (X, E)$ be a greedy $(1+\epsilon)$-spanner on the point set $X\subseteq \mathbb{R}^{d}$\;
	\For{$k = 1, 2, \ldots, O(\log^*(d/\epsilon))$}{
		\tcc{the first pruning phase}
		$E_1\leftarrow E$\;
		\For{$i = 1, 2, \ldots, O(\log\alpha)$\label{iter-i}}{
			\For{$j = 0, 1, 2, \ldots$}{
				define $L_j = \{e\in E : \|e\|\in [\beta^j, \beta^{j+1})\}$\;
				while there exists $xy\in \binom{X}{2}$ such that $\|xy\|\geq \beta^j / 25$, and $|P_{x, y}|\geq \frac{\alpha}{2^i\kappa}$, where $P_{x,y} = \left\{s t\in L_j\cap E_1\cap E^{(\romannumeral1)} : \|sx\|+ \|xy\|  + \|yt\|\leq (1+\epsilon)\cdot\|st\|  \right\}$ (see \Cref{edgeclass})\;
				add $xy$ to $E_1$ as a new edge, and remove $P_{x, y}$ from $E_1$\;
			}
		}
		\tcc{the second pruning phase}
		$E_2\leftarrow E_1\setminus E^{(\romannumeral2)}$\;
		\For{edge $st\in E_1\cap E^{(\romannumeral2)}$ in non-decreasing order in terms of of norm}{
			\If{$\dist_{H_2}(s, t) > (1+\kappa^2\delta)\|st\|$}{
				$E_2\leftarrow E_2\cup \{s t\}$\;
				find $a\in A_{s, t}, b\in B_{s, t}$ and add $a b$ to $E_2$ as a new helper edge\;
			}
		}
		$E\leftarrow E_2, \delta\leftarrow \Delta(\kappa, \delta), \alpha\leftarrow O(\log\alpha)$\;
	}
	\Return $E$\;
\end{algorithm}

\subsection{Stretch Analysis} \label{sec:stretch}
Before analyzing the stretch of our spanner, we first need to bound the value of $\delta$ throughout all $O(\log^*(d/ \epsilon))$ iterations of pruning.

\begin{claim}\label{small-delta}
	Throughout all $O(\log^*(d / \epsilon))$ iterations of pruning, we have $\delta < \kappa^{-5}$ and $\Delta(\kappa, \delta) < (\kappa+1)^2\delta$.
\end{claim}
\begin{proof}
	We show by induction on $i$ that right before the $i$-th iteration of pruning, we have $\delta \leq (\kappa+1)^{2(i-1)}\epsilon$. For the basis step, recall that at the beginning when $E$ was initialized as a greedy $(1+\epsilon)$-spanner, we have $\delta = \epsilon$. For the induction step, suppose that $\delta \leq (\kappa+1)^{2(i-2)}\epsilon < \kappa^{-5}$ (by \Cref{eq:relationeps}) right before the $(i-1)$-st iteration. Then, at the end of the $(i-1)$-st iteration, we have updated $\delta$ as:
	$$\begin{aligned}
		\Delta(\kappa, \delta) &= (1+\delta)\cdot(1+\kappa\delta)\cdot (1+\kappa^2\delta) - 1\\
		&= \brac{1+(\kappa+1)\delta + \kappa\delta^2}\cdot(1+\kappa^2\delta) - 1\\
		&< \brac{1 + (\kappa+2)\delta}\cdot (1+\kappa^2\delta)-1\\
		&= (\kappa^2+\kappa+2)\delta + \kappa^2(\kappa+2)\delta^2\\
		&< (\kappa+1)^2\delta \\
        &\le (\kappa+1)^2 \cdot (\kappa+1)^{2(i-2)}\epsilon \\
        &= (\kappa+1)^{2(i-1)}\epsilon.
	\end{aligned}$$
	Therefore, at the end of the $(i-1)$-st iteration, and so also at the beginning of the $i$th iteration, we have $\delta < (\kappa+1)^{2(i-1)}\epsilon$, which completes the induction step.
It follows that throughout all $O(\log^*(d / \epsilon))$ iterations we have $\Delta(\kappa, \delta) \le (\kappa+1)^{O(\log^*(d / \epsilon))} \eps < \kappa^{-5}$, where the last inequality follows by  employing \Cref{eq:relationeps} again. 
\end{proof}

Next, let us show an upper bound on the stretch during the execution of the first pruning phase.
\begin{claim}\label{edge-stretch}
	At the end of the first pruning phase, for any edge $s t\in E$, we are guaranteed that $\dist_{H_1}(s, t)\leq (1+\kappa\delta)\cdot \|st\|$. Also, at the end of the second phase, for any edge $st\in E_1\cap E$, we have that $\dist_{H_2}(s, t)\leq (1+\kappa^2\delta)\cdot \|st\| $.
\end{claim}
\begin{proof}  
    The first assertion holds trivially for all edges that remain in $E_1$ until the end of the first phase; this includes all type-(ii) edges and possibly some type-(i) edges. We henceforth restrict the attention only to edges that get pruned from $E_1$ during the first phase, and for each such edge $e$ we consider the index $j$ such that $e \in L_j \cap E_1 \cup E^{(i)}$. We will prove that the stretch of each edge $e = st$ that is pruned from $E_1$ is in check (i.e., $\dist_{H_1}(s, t)\leq (1+\kappa\delta)\cdot \|st\|$) by induction on $j$, $j \ge 0$. We stress that the induction is applied on the final edge set $E_1$ at the end of the first pruning phase (when the last sub-iteration ends).
    
	As the basis when $j = 0$, all edges in $E^{(\romannumeral1)}\cap L_0$ would be added to $E_1$ and never pruned (otherwise their stretch would be at least $2$), so the stretch between the endpoints of any type-(\romannumeral1) edge in $L_0$ is $1$. Now, for $j\geq 1$, let us assume $\dist_{H_1}(s, t)\leq (1+\kappa\delta)\cdot \|st\|$ for every edge $s t\in L_{k}$ for every $k<j$. Consider any type-(\romannumeral1) edge $s t\in L_j$ which was removed from $E_1$ during the first phase. Then, by the algorithm, there must exist an edge $xy\in E_1$ such that $\|xy\|\geq \beta^j / 25$ and
	$$\|sx\| + \|xy\| +    \|yt\|\leq (1+\epsilon)\|st\|.$$
    Note that this edge will not be removed later on during the first phase since we only prune old edges.
    As $\|xy\| \geq \beta^j/25$ and $\|st\|<\beta^{j+1}$, we have
	$$\max\{\|sx\|, \|yt\|\}\leq (1+\epsilon)\|st\| - \|xt\| \leq \brac{1+\epsilon - \frac{1}{25\beta}}\|st\| < \frac{\|st\|}{\beta^2}.$$
	Since we assumed $(X, E)$ is a $(1+\delta)$-spanner for $X$, there exist two paths $\gamma_1$ and $\gamma_2$ in $E$ connecting $s, x$ and $y, t$ such that
	$$\|\gamma_1\|\leq (1+\delta)\|sx\|,$$
	$$\|\gamma_2\|\leq (1+\delta)\|yt\|.$$
	Therefore, every edge $e = s' t'$ on the paths $\gamma_1$ or $\gamma_2$ has length at most
	$$\|e\|\leq (1+\delta)\max\{\|sx\|, \|yt\|\}\leq \frac{1+\delta}{\beta^2} \|st\| < \frac{\|st\|}{\beta}.$$
	In other words, $e$ belongs to some set $L_k$ for $k<j$. Using the inductive hypothesis, $\dist_{H_1}(s', t')\leq (1+\kappa\delta)\|s't'\|$, and consequently we have the following inequalities which conclude the induction
    \begin{align*}
		\dist_{H_1}(s, t)&\leq \dist_{H_1}(s, x) + \|xy\| + \dist_{H_1}(y, t)\\
		&\leq (1+\kappa\delta)\brac{(1+\epsilon)\|st\| - \|xy\|} + \|xy\|\\
		&= (1+\epsilon)(1+\kappa\delta)\|st\| - \kappa\delta\|xy\|\\
		&\leq \brac{(1+\epsilon)(1+\kappa\delta) - \frac{\kappa\delta}{25\beta}}\|st\| < (1+\kappa\delta)\|st\|, 
    \end{align*}
    where the penultimate inequality holds as $\|xy\| \ge \frac{1}{25\beta}\|st\|$ and the last inequality follows from \Cref{eq:relationeps} and \Cref{small-delta}, which yield $\eps(1+\kappa\delta) < \frac{\kappa\delta}{25\beta}$.

	During the second pruning phase, if an old edge $s t\in E_1$ was not added to $E_2$, there must exist a path $\pi$ between $s$ and $t$ consisting of edges in $E_2$ such that $\|\pi\| \leq (1+\kappa^2\delta)\cdot\|st\| $. Since the set $E_2$ grows monotonically, the path $\pi$ is preserved for the remainder of the phase.
\end{proof}

\begin{corollary}\label{spanner-stretch}
	For every $s t\in \binom{X}{2}$, we have:
	$$\dist_{H_2}(s, t)\leq \brac{1+\Delta(\kappa, \delta)}\cdot \|st\| < \brac{1 + (\kappa+1)^2\delta}\cdot \|st\|.$$
\end{corollary}
Note that he second inequality in the corollary above holds as $\delta < \kappa^{-5}$ by \Cref{small-delta}.

\subsection{Sparsity Analysis} \label{sparsityanal}
We first analyze how the number of edges in $E$ changes in one iteration, and then consider $k$ consecutive iterations at the end of this subsection.
To analyze sparsity, we devise a charging scheme $\Psi_0$ that maps (fractionally) edges in $E$ to edges in $E_\sparse$. 
In fact, the charging scheme charges (possibly fractionally) edges from $E$ to edges in a subdivision of $E_\sparse$. We also stress that the total number of edges in the subdivision of $E_\sparse$ to which we charge is at most $\kappa\cdot |E_\sparse| = O(|E_\sparse|)$.

\paragraph{Charging scheme from $E$ to $E_\sparse$.} For each edge $s t\in E$, find a path $\pi_{s, t}$ in $E_\sparse$ between $s$ and $t$ such that $\|\pi_{s, t}\| \leq (1+\epsilon)\cdot \|st\| $. If $s t\in E^{(\romannumeral1)}$, then by definition $A_{s, t}$ or $B_{s, t}$ is empty. Since $\pi_{s, t}$ is a path connecting $s, t$ with total weight at most $(1+\epsilon)\|st\|$, the entire path $\pi_{s, t}$ should lie within the ellipsoid $\Gamma_{s, t}$ and thus there must be a single edge $e$ on $\pi_{s, t}$ that crosses the region $A_{s, t}$ or $B_{s, t}$. In this case, $\Psi_0$ charges the edge $s t$ to edge $e\in E_\sparse$.

Next, assume $s t\in E^{(\romannumeral2)}$. To charge type-(\romannumeral2) edges to $E_\sparse$, we subdivide each edge $e\in E_\sparse$ evenly into $\kappa$ sub-segments with at most $\kappa-1$ Steiner points. Let $Y\supseteq X$ be the point set containing all original points and Steiner points, and 
let $E_\sparse^Y\subseteq \binom{Y}{2}$ denote the set of subdivided edges (clearly, $|E_\sparse| = O(|E_\sparse^Y|)$). Our charging scheme will be from type-(\romannumeral2) edges to edges in $E_\sparse^Y$. We distinguish between two cases, depending on the path $\pi_{s, t}$.

\begin{enumerate}[(a),leftmargin=*]
	\item Suppose $\pi_{s, t}\cap A_{s, t}$ or $\pi_{s, t}\cap B_{s, t}$ is empty; that is, the polygonal path $\pi_{s, t}$ does not contain vertices in $A_{s, t}$ or $B_{s, t}$. Then, since $\pi_{s, t}$ lies in $\Gamma_{s, t}$ entirely, there must be an edge $e = s' t'$ in $\pi_{s, t}$ that crosses $A_{s, t}$ or $B_{s, t}$. If $e$ only crosses one of the two regions (say $A_{s, t}$), then we have
	$$\|s - \proj_{st}(s')\| < \brac{\frac{3}{8} - \frac{1}{50}}\cdot \|st\|,$$
	$$\brac{\frac{3}{8} - \frac{1}{50}}\cdot \|st\| < \|\proj_{st}(t') - t\| < \brac{\frac{5}{8} - \frac{1}{50}}\cdot \|st\|.$$
	Let $z\in Y\cap e$ be the Steiner point in $A_{s, t}$ on the segment $e$ that is closest to $s'$; such a point $z$ must exist since each sub-segment of $e$ has length at most $\frac{\|e\|}{\kappa} < \frac{\|st\|}{25}$. Then, $\Psi_0$ charges $s t$ to segment $z t'$ which has length at least $\brac{\frac{1}{25} - \frac{1}{\kappa}}\|st\| > \frac{\|st\|}{26}$.
	
	If $e$ crosses both regions $A_{s, t}$ and $B_{s, t}$, then let $z_1\in Y\cap e$ be the Steiner point in $A_{s, t}$ that is closest to $s'$, and let $z_2\in Y\cap e$ be the Steiner point in $B_{s, t}$ that is closest to $t'$. Then, $\Psi_0$ charges $s t$ to segment $z_1 z_2$ which has length at least $\brac{\frac{1}{4} + \frac{1}{25} - \frac{2}{\kappa}}\|st\| > \frac{\|st\|}{4}$.
	
	\item Otherwise, we will charge $s t$ {\em fractionally} to a set of edges in $E_\sparse^Y$. Move along $\pi_{s, t}$ from $s$ to $t$ and let $p$ be the last vertex in $A_{s, t}$ and let $q$ be the first vertex in $B_{s, t}$. As $\|\pi_{s, t}\|\leq (1+\epsilon)\cdot \|st\|$ and $\|\proj_{st}(p)-\proj_{st}(q)\| \ge (1/4 - 1/25) \|st\|$, we know that:
	$$\|\pi_{s, t}[p, q]\| \leq \|\proj_{st}(p)-\proj_{st}(q)\|  + \epsilon\cdot \|st\| \leq (1+10\epsilon)\cdot \|\proj_{st}(p)-\proj_{st}(q)\|. $$
	Therefore, applying \Cref{angle-bound},
	we know that
	$$\|E(\pi_{s, t}[p, q], st, 2\sqrt{10\epsilon})\|> 0.5\cdot \|\proj_{st}(p)-\proj_{st}(q)\|. $$
	Then, for each edge $e\in E(\pi_{s, t}[p, q], st, 2\sqrt{10\epsilon})$, $\Psi_0$  charges a fraction of $\frac{2\cdot\|e\| }{\|\proj_{st}(p)-\proj_{st}(q)\| }$ of edge $s t$ to edge $e$.
\end{enumerate}

By design of our charging scheme $\Psi_0$, we can upper bound the angle between any type-(\romannumeral2) edge $s t$ and the edge $e$ it charges to.
\begin{claim}\label{angle}
	If a type-(\romannumeral2) edge $s t$ charges to an edge $x y\in E_\sparse^Y$, then angle $\angle(st, xy)$ is at most $15\sqrt{\epsilon}$. Furthermore, for all $z\in \{x, y\}$, the projection $\proj_{st}(z)$ of $z$ onto line $st$ lies on the segment $st$ and satisfies $\frac{\|s - \proj_{st}(z)\|}{\|st\|}\in \left[\frac{3}{8} -\frac{1}{50}, \frac{5}{8} +\frac{1}{50}\right]$.
\end{claim}
\begin{proof}
	The second assertion of the statement holds by the design of our charging scheme, so let us focus on the first assertion. If $st$ charges to $x y$ as in case-(b), then the assertion holds by design. Otherwise, assume $s t$ charges to $x y$ as in case-(a). Define $\theta = \angle(st, xy)$. Since both $x$ and $y$ are on the path $\pi_{s, t}$, let us assume w.l.o.g.\ that $x$ lies between $s$ and $y$ on $\pi_{s, t}$, and then we have
\begin{align*}
		\epsilon\|st\| &\geq \|\pi_{s, t}\| - \|st\| \geq \|sx\| + \|xy\| + \|yt\| - \|st\|\\
		&\geq \|xy\| - \|\proj_{st}(xy)\| = (1 - \cos\theta)\|xy\|\\
		&\geq \frac{\theta^2}{8}\cdot \|xy\| > \frac{\theta^2}{208}\cdot\|st\|.
\end{align*}
	Here we have used the fact that $1-\cos\theta = 2\sin^2(\theta/2) > \theta^2/8$ and $\|xy\| > \frac{\|st\|}{26}$. Therefore, $\theta\leq 15\sqrt{\epsilon}$.
\end{proof}

We need to argue that $\Psi_0$ is a valid charging scheme from $E$ to $E_\sparse\cup E_\sparse^Y$.
\begin{claim}
    \label{clm:fully-charge}
	Every edge in $E$ is fully charged to edges in a subdivision of $E_\sparse$; that is, the total charges produced by any edge in $E$ is at least $1$.
\end{claim}
\begin{proof}
	According to the charging scheme, every type-(\romannumeral1) edge is fully charged to an edge in $E_\sparse$. As for type-(\romannumeral2) edges, in case (a), we are also using an integral charging (i.e., charging to a single edge) to edges in $E_\sparse^Y$. In case (b), the fractions of $st$ charged to various edges in $E(\pi_{s, t}[p, q], 2\sqrt{10\epsilon})$ sum to at least $2\cdot \frac12\,\|pq\| / \|pq\| = 1$ by \Cref{angle-bound}, and so $s t$ is also fully charged fractionally to some Steiner edges in $E_\sparse$.
\end{proof}

Next, let us analyze the sparsity of the first pruning phase.

\begin{claim}\label{sparsity-phase1}
	During the first pruning phase, the number of new edges added to $E_1$ is at most $O(|E_\sparse|\log\alpha)$. After the first pruning phase, the number of type-(\romannumeral1) edges in $E_1$ is at most $O(|E_\sparse|)$.
\end{claim}
\begin{proof}
	During the first pruning phase, we show by induction on $i$ that at the beginning of the $i$-th sub-iteration, $|E_1\cap E^{(\romannumeral1)}|$ is at most $|E_\sparse|\,\alpha / 2^{i-1}$. For the basis when $i = 1$, this bound holds as $\alpha \geq |E| / |E_\sparse|$.
    
	For the inductive step, to bound $\left|E_1\cap E^{(\romannumeral1)}\right|$ after the $i$-th sub-iteration, we need to utilize our charging scheme for type-(\romannumeral1) edges. Suppose, for the sake of contradiction, that there are more than $|E_\sparse|\alpha / 2^{i-1}$ type-(\romannumeral1) edges remaining in $E_1\cap E^{(\romannumeral1)}$ after the $i$-th sub-iteration. Then, by the pigeon-hole principle, there exists a set $F$ of more than $\alpha / 2^{i-1}$ type-(\romannumeral1) edges currently in $E_1$ charging to the same edge $xy\in E_\sparse$. By our charging scheme, for each such edge $e\in F$, we have 
 $$\|xy\| 
    \le (1+\eps)\|e\| 
    \le (1+\eps)25\cdot \|xy\| 
    \le 25\beta\cdot \|xy\|.$$
 Therefore, by the pigeonhole principle there exists an index $j\geq 0$ such that
 $$|F\cap L_j|
     \geq \frac{\alpha}{2^{i-1}\cdot \log_\beta{25\beta}} 
     > \frac{\alpha}{2^i\kappa}.$$
 The last inequality holds since $\kappa = 10^4 > \log_\beta(25\beta)$.
	Now, consider the sub-iteration during the first pruning phase when we were processing edges in $L_j$.

	To reach a contradiction, it suffices to show that at that time, it must be that $P_{x, y}\supseteq F\cap L_j$. In fact, for any edge $st\in F\cap L_j$, by the charging scheme, we know that $xy$ is on a $(1+\eps)$-spanning path of $E_\sparse$ between $s$ and $t$, hence the triangle inequality yields
	$$\|sx\|  + \|xy\| + \|yt\| \leq (1+\epsilon)\|st\| .$$
    Note that for each edge $e = st \in F \cap L_j$, we have $\|xy\| \ge \beta^j / 25$. Therefore, the algorithm could have added $x y$ to $E_1$ as a new edge and remove the entire set $P_{x, y}$, leading to a contradiction. This completes the proof of the induction step.
	
    It remains to bound the number of new edges added to $E_1$. In the $i$-th sub-iteration, each time we add a new edge to $E_1$, we  decrease $\left|E_1\cap E^{(\romannumeral1)}\right|$ by at least $\frac{\alpha}{2^i\kappa}$. Since $\left|E_1\cap E^{(\romannumeral1)}\right|$ was at most $|E_\sparse|\alpha / 2^{i-1}$ at the beginning of the $i$-th sub-iteration, we could add at most $O(|E_\sparse|)$ edges to $E_1$ in this sub-iteration. It follows that at most $O(|E_\sparse|\log\alpha)$ new edges are added to $E_1$ in total.
\end{proof}

Next, let us analyze the sparsity of the second pruning phase. Let us begin with a basic observation which follows directly from the algorithm description.

\begin{observation}\label{helper}
	Every type-(\romannumeral2) edge $s t\in E_2\cap E^{(\romannumeral2)}$ must have a helper edge $a b$ when it was processed and added to $E_2$ during the second pruning phase.
\end{observation}

The following lemma is the key behind the charging argument.
\begin{lemma}\label{one-edge-per-level}
	Fix any edge $e\in E_\sparse^Y$ and level index $j\geq 0$. Then, after the second pruning phase, there is at most one type-(\romannumeral2) edge in $E_2\cap L_j$ that is charged to $e$.
\end{lemma}
\begin{proof}
	Assume for contradiction that there are two distinct type-(\romannumeral2) edges $s_1t_1, s_2 t_2\in E_2\cap L_j$ charging to the same edge $e\in E_\sparse^Y$. Let $r$ be an arbitrary endpoint of $e$. Without loss of generality, assume $\|s_1t_1\|\leq \|s_2t_2\|$, and so edge $s_1 t_1$ was processed before edge $s_2 t_2$. By \Cref{helper}, a helper edge of $s_1 t_1$, say edge $a b$, was added to $E_2$.

	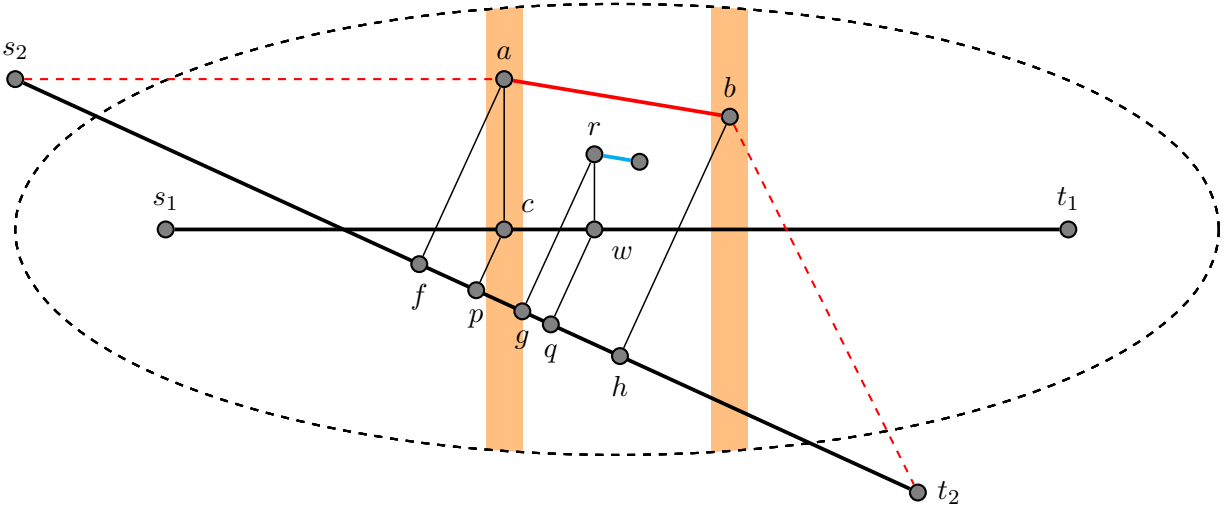
\begin{figure}
		\centering
		\begin{tikzpicture}[thick,scale=1]
	\draw (0, 0) node(1)[circle, draw, fill=black!50,
	inner sep=0pt, minimum width=6pt, label = $s_1$] {};
	\draw (12, 0) node(2)[circle, draw, fill=black!50,
	inner sep=0pt, minimum width=6pt,label = $t_1$] {};
	
	\def\stellipse{(6, 0) ellipse (8 and 3)};
	\def\firstrec{(4.26, -4) rectangle (4.74, 4)};
	\def\secondrec{(7.26, -4) rectangle (7.74, 4)};
	
	\begin{scope}
		\clip \stellipse;
		\fill[orange, opacity=0.5] \firstrec;
	\end{scope}
	
	\begin{scope}
		\clip \stellipse;
		\fill[orange, opacity=0.5] \secondrec;
	\end{scope}
	
	\draw [black, dashed] \stellipse;
	\draw [line width = 0.5mm] (1) to (2);
	
	\draw (4.5, 2) node(3)[circle, draw, fill=black!50,
	inner sep=0pt, minimum width=6pt, label = $a$] {};
	\draw (7.5, 1.5) node(4)[circle, draw, fill=black!50,
	inner sep=0pt, minimum width=6pt,label = $b$] {};
	
	\draw [black, dashed] \stellipse;
	\draw [line width = 0.5mm, color=red] (3) to (4);
	
	\draw (5.7, 1) node(5)[circle, draw, fill=black!50,
	inner sep=0pt, minimum width=6pt, label = $r$] {};
	\draw (6.3, 0.9) node(6)[circle, draw, fill=black!50,
	inner sep=0pt, minimum width=6pt] {};
	
	\draw [black, dashed] \stellipse;
	\draw [line width = 0.5mm, color=cyan] (5) to (6);
	
	\draw (-2, 2) node(7)[circle, draw, fill=black!50,
	inner sep=0pt, minimum width=6pt, label = $s_2$] {};
	\draw (10, -3.5) node(8)[circle, draw, fill=black!50,
	inner sep=0pt, minimum width=6pt,label = {0: {$t_2$}}] {};
	
	\draw [black, dashed] \stellipse;
	\draw [line width = 0.5mm] (7) to (8);

	\draw ($(7)!(3)!(8)$) node(9)[circle, draw, fill=black!50, inner sep=0pt, minimum width=6pt, label = {-90:{$f$}}] {};
	\draw [line width = 0.2mm] (9) -- (3);
	
	\draw ($(1)!(3)!(2)$) node(10)[circle, draw, fill=black!50, inner sep=0pt, minimum width=6pt, label = {45:{$c$}}] {};
	\draw [line width = 0.2mm] (3) -- (10);
	
	\draw ($(7)!(4)!(8)$) node(11)[circle, draw, fill=black!50, inner sep=0pt, minimum width=6pt, label = {-90:{$h$}}] {};
	\draw [line width = 0.2mm] (4) -- (11);
	
	\draw ($(7)!(10)!(8)$) node(12)[circle, draw, fill=black!50, inner sep=0pt, minimum width=6pt, label = {-90:{$p$}}] {};
	\draw [line width = 0.2mm] (10) -- (12);
	
	\draw ($(1)!(5)!(2)$) node(13)[circle, draw, fill=black!50, inner sep=0pt, minimum width=6pt, label = {-45:{$w$}}] {};
	\draw [line width = 0.2mm] (5) -- (13);
	
	\draw ($(7)!(5)!(8)$) node(14)[circle, draw, fill=black!50, inner sep=0pt, minimum width=6pt, label = {-90:{$g$}}] {};
	\draw [line width = 0.2mm] (5) -- (14);
	
	\draw ($(7)!(13)!(8)$) node(15)[circle, draw, fill=black!50, inner sep=0pt, minimum width=6pt, label = {-90:{$q$}}] {};
	\draw [line width = 0.2mm] (13) -- (15);
	
	\draw [dashed, color=red] (7) -- (3);
	\draw [dashed, color=red] (4) -- (8);
\end{tikzpicture}
		\caption{Both $s_1 t_1$ and $s_2 t_2$ are charging to $e$ which is drawn as the cyan edge. If $s_1 t_1$ was added to $E_2$ together with a helper edge $a b$, then $s_2 t_2$ could not be added to $E_2$ later on since $(s_2\rightsquigarrow a)\circ ab \circ (b\rightsquigarrow t_2)$ will be a good path  in $H_2$.}\label{charge}
	\end{figure}
 
	Write $D = \|s_1t_1\|$. Let $c$ and $w$ be the orthogonal projections of $a$ and $r$ on the line $s_1t_1$; let $f$, $g$, and $h$ be the projections of $a$, $r$, and $b$ on line $s_2t_2$; and let $p$ and $q$ be the projections of $c$ and $w$ on line $s_2t_2$; see \Cref{charge} for an illustration in the 2-dimensional case. By the design of our charging scheme and \Cref{angle}, we know that $g$ should land on segment $s_2t_2$ and
  $\|cw\|\leq 0.29\cdot D$, and 
  $\angle(s_1t_1, s_2t_2)
  \leq \angle(s_1t_1, e)+\angle(e, s_2t_2)
  \leq 30\sqrt{\epsilon}$.
	As $s_2 t_2$ is also charged to $e$, we know that $\|s_2g\| / \|s_2t_2\|\in \left[\frac{3}{8}-\frac{1}{50}, \frac{5}{8} + \frac{1}{50}\right]$. Thus, by the triangle inequality, we obtain
\begin{align*}
		\|s_2p \| &\geq \|s_2 g\| - \|gq\| - \|pq\|\\
            &\geq 0.355D - \sin\brac{\angle(s_1t_1, s_2t_2)}\cdot\|rw\| - \|cw\|\\
		&\geq \brac{0.355 - 2\sqrt{\epsilon}\cdot \sin(30\sqrt{\epsilon}) - 0.29}D\\
		&> 0.06D .
\end{align*}
Note that the last inequality holds for sufficiently small $\eps$. Thus, $\|s_2f\| \geq \|s_2p\| - \|pf\| > 0.06D - \sin(\angle(s_1t_1, s_2t_2))\cdot 2\sqrt{\epsilon}D > 0.06D - 60\epsilon D > 0.05D$. This also shows that both $f$ and $p$ should land on the segment $s_2t_2$.
	
	On the other hand, using the triangle inequality for the projections of segments $cp, cw, wq$ on the hyperplane orthogonal to line $s_2t_2$, we have:
\begin{align*}
		\|cp\| &\leq \|wq\| + \|cw\|\cdot \sin(30\sqrt{\epsilon})\\
		&< \|wq\| + 30\sqrt{\epsilon}\|cw\|\\
		&< \|wq\| + 8.7\sqrt{\epsilon}D\\
		&\leq \|rg\| + \|rw\| + 8.7\sqrt{\epsilon}D\\
		&\leq 12.7\sqrt{\epsilon}D.
\end{align*}
	Here, we have used the fact that $\|cw\|\leq 0.29D$ and $\|rg\|,\|rw\|\leq 2\sqrt{\epsilon}D$.
    Therefore, $\|af\| \leq \|ac\| + \|cp\| \leq 14.7\sqrt{\epsilon}D$. Symmetrically, we can show that $\|bh\|\leq 14.7\sqrt{\epsilon}D$ and $\|t_2h\| > 0.05D$.
    
	Finally, let us show that $(s_2\rightsquigarrow a)\circ ab \circ (b\rightsquigarrow t_2)$ can make a good path from $s_2$ to $t_2$ in $H_2$. By the above calculations, as $\|s_2f\|>0.05D$ and $\|af\| \leq 14.7\sqrt{\epsilon}D$, we have
\begin{equation}
\label{eq217a}
\begin{aligned}
		\|s_2a\| - \|s_2f\| &= \sqrt{\|s_2f\|^2 + \|af\|^2} - \|s_2f\| = \frac{\|af\|^2}{\sqrt{\|s_2f\|^2 + \|af\|^2} + \|s_2f\|}\\
		& < 217\epsilon D^2 / 0.1D = 2170\epsilon D.
\end{aligned}
\end{equation}
	Symmetrically, we can show that $h$ also lands on the segment $s_2t_2$ and
    \begin{equation}\label{eq217b}
	\|bt_2\| - \|ht_2\|\leq  2170\epsilon D .
    \end{equation}
	Also, since $\angle(ab, fh)\leq 15\sqrt{\epsilon}$ by \Cref{angle}, we have 
	\begin{align*}
	    \|ab\| &\leq \frac{\|fh\|}{\cos(15\sqrt{\eps})} \leq \brac{1 + 112.5\eps + O(\eps^2)}\cdot\|fh\|\\
        &<  (1+115\epsilon)\|fh\|
		< \|fh\| + 120\epsilon D .
	\end{align*}
    The last inequality is because $\|fh\| \leq \|s_2t_2\|< \beta D$.
	
	By \Cref{spanner-stretch}, we have
	$$\dist_{H_2}(s_2, a)\leq \brac{1 + (\kappa+1)^2\delta}\cdot\|s_2a\|,$$
	$$\dist_{H_2}(b, t_2)\leq \brac{1 + (\kappa+1)^2\delta}\cdot\|bt_2\|.$$
	Therefore, using \Cref{eq217a,eq217b} combined with  \Cref{small-delta} (i.e., $\delta  < k^{-5}$), at the time when edge $s_2 t_2$ was being processed by the second pruning phase, we have 
\begin{align*} 
\dist_{H_2}(s_2, t_2) 
    &\leq \dist_{H_2}(s_2, a) + \|ab\| + \dist_{H_2}(b, t_2)\\
    &\leq \brac{1 + (\kappa+1)^2\delta}\cdot\|s_2a\| + \brac{1 + (\kappa+1)^2\delta}\cdot\|bt_2\| + \|ab\|\\
    &\leq \brac{1 + (\kappa+1)^2\delta}\cdot(\|s_2t_2\| + 2860\epsilon D - \|ab\|) + \|ab\|\\
    &<(1+\kappa^2\delta)\|s_2t_2\| + \brac{(2\kappa+1)\delta\|s_2t_2\|+\brac{1+(\kappa+1)^2\delta}\cdot 2860\epsilon D - (\kappa+1)^2\delta\|ab\|}\\
    &<(1+\kappa^2\delta)\|s_2t_2\| + \brac{2\kappa+3000 - 0.2(\kappa+1)^2}\delta \|s_2t_2\|\\
    &<(1+\kappa^2\delta)\|s_2t_2\| .
\end{align*}
	The penultimate inequality holds as $\|ab\| \geq 0.21\|s_1t_1\| > 0.21\beta^{-1}\|s_2t_2\| > 0.2\|s_2t_2\|$, $\delta \geq \eps$ and $\kappa=10^4$. Therefore, $s_2 t_2$ could not have been added to $E_2$, leading to a contradiction.
\end{proof}

\begin{corollary}
    \label{cor:total-type-ii}
	The total number of type-(\romannumeral2) edges added to $E_2$ is at most $O(|E_\sparse|)$.
\end{corollary}
\begin{proof}
	Consider any subdivided edge $e\in E_\sparse^Y$ within an edge in $E_\sparse$. By \Cref{one-edge-per-level}, for each $j\geq 0$, at most one type-(\romannumeral2) edge in $L_j\cap E_2$ charged to $e$. By the fractional charging scheme, the total amount of charges that $e$ receives across all $j\geq 0$ is a geometric sum which is bounded by a constant. Note that we analyze integral (case (a)) and fractional (case (b)) charges separately: Integral charges are incurred by edges on a constant number of levels, since by the charging scheme every edge charged to $e$ must have weight $\Theta(\|e\|)$; the fractional charges to $e$ could be incurred by edges on possibly many levels, but our charging scheme guarantees that the fractional charge to $e$ decays geometrically with the level, hence the sum of fractional charges is bounded by a geometric sum.
\end{proof}

\paragraph{Putting it all together.}
\Cref{cor:total-type-ii}, together with \Cref{sparsity-phase1}, yields $|E_2| = O(|E_\sparse|\log\alpha )$. Also, by \Cref{spanner-stretch}, the stretch of $H_2$ is at most $1+(\kappa+1)^2\delta$. Therefore, starting with $(X, E)$ being a $(1+\eps)$-spanner with $\delta=\eps$ and $\alpha = \epsilon^{-O(d)}$, if we iterate these two pruning phases $k$ times, we end up with a spanner with $\brac{1+\epsilon\cdot 2^{O(k)}}$-stretch and $O\brac{\log^{(k)}(1/\eps) + \log^{(k-1)}(d)}\cdot |E_\sparse|$ edges, which concludes the sparsity bound of \Cref{thm:tech}. Note that Algorithm \ref{greedy-prune} (Algorithm $\mathsf{GreedyPrune}$) repeats the two pruning phases for $k=O(\log^*(d/\epsilon))$ iterations, the result of which provides a spanner with the stretch and sparsity bounds of \Cref{thm:main}.

\subsection{Lightness Analysis} \label{sec:light}
Similarly to the analysis of sparsity, we devise a charging scheme from edges in $E$ to edges in $E_\light$.

\paragraph{Preparation.} Since we deal with edge lengths, we use more refined geometric properties. For each edge $st\in E$, let $\pi_{s, t}$ be a path in $E_\light$ between $s$ and $t$ such that $\|\pi_{s, t}\| \leq (1+\epsilon)\cdot \|st\|$. For technical reasons, we wish that all points on the path $\pi_{s, t}$ are monotonically increasing in terms of their projections on the directed line $\overrightarrow{st}$, which is not necessarily the case in general.  Therefore, as a preliminary step, we construct a sub-sequence of edges $\rho_{s, t}\subseteq\pi_{s, t}$ from $s$ to $t$ in $\mathbb{R}^d$ such that
\begin{enumerate}[(1)]
	\item the projection of $\rho_{s, t}$ on the line $st$ is a segment that covers the entire segment $st$;
	\item every point on segment $st$ is covered at most twice by the projection of $\rho_{s, t}$ on line $st$.
\end{enumerate}
To define this sub-sequence $\rho_{s, t}$, consider the following iterative procedure which uses a curser variable $z$ starting at $z\leftarrow s$, as well as a partially constructed sub-sequence $\rho\subseteq \pi_{s, t}$ from $s$ to $z$. In each iteration, let $x y$ be the last edge on $\pi_{s, t}$ whose projection on $st$ contains the projection of $z$ on $st$; namely $\proj_{st}(z)\in \proj_{st}(xy)$. Then, extend $\rho_{s, t}$ by $\rho_{s, t}\leftarrow \rho_{s, t}\cup \{x y\}$ and reassign $z\leftarrow y$. By the construction, it is easy to see that both requirements (1)--(2) on $\rho_{s, t}$ are met. Using $\rho_{s, t}$, we describe a charging scheme $\Psi_1$ that discharges the weight of the edges of $E$ to edges in $E_\light$.

\paragraph{Charging scheme from $E$ to $E_\light$.} 
The charging scheme $\Psi_1$ is almost the same as $\Psi_0$, except that we are using $\rho_{s, t}$ in lieu of $\pi_{s, t}$ for type-(\romannumeral2) edges.

If $s t\in E^{(\romannumeral1)}$, then by definition $A_{s, t}$ or $B_{s, t}$ is empty. Since $\pi_{s, t}$ is a path connecting $s, t$ with total weight at most $(1+\epsilon)\|st\|$, the entire path $\pi_{s, t}$ should lie within the ellipsoid and thus there must be a single edge $e$ on $\pi_{s, t}$ that crosses the region $A_{s, t}$ or $B_{s, t}$. In this case, $\Psi_1$ charges weight $\|st\|$ to such an edge $e\in E_\light$.

Next, assume that $s t\in E^{(\romannumeral2)}$. To charge type-(\romannumeral2) edges to $E_\light$, we subdivide each edge $e\in E_\light$ evenly into $\kappa$ sub-segments by adding at most $\kappa-1$ Steiner points on $e$. Let $Y\supseteq X$ be the point set containing all original points and Steiner points, and 
let $E_\light^Y\subseteq \binom{Y}{2}$ denote the set of subdivided edges (clearly, $\|E_\light\| = \|E_\light^Y\|$). Our charging scheme will be from type-(\romannumeral2) edges to edges in $E_\light^Y$. We distinguish between two cases, depending on the path $\rho_{s, t}$.
\begin{enumerate}[(a),leftmargin=*]
	\item Suppose that $\rho_{s, t}\cap A_{s, t}$ or $\rho_{s, t}\cap B_{s, t}$ is empty; that is, the edges in  $\rho_{s, t}$ do not have any endpoints in $A_{s, t}$ or $B_{s, t}$. Then, since $\rho_{s, t}$ lies in $\Gamma_{s, t}$ entirely and $st\subseteq \proj_{st}(\rho_{s, t})$, there must be an edge $e = s't'$ in $\rho_{s,t}\subseteq \pi_{s, t}$ that crosses $A_{s, t}$ or $B_{s, t}$. If $e$ only crosses one of the two regions (say $A_{s, t}$), then we have:
	$$\|s - \proj_{st}(s')\| < \brac{\frac{3}{8} - \frac{1}{50}}\cdot \|st\|,$$
	$$\brac{\frac{3}{8} -\frac{1}{50}}\cdot \|st\| < \|\proj_{st}(t') - t\| < \brac{\frac{5}{8} - \frac{1}{50}}\cdot \|st\|.$$
	Let $z\in Y\cap e$ be the Steiner points on segment $e$ which is in $A_{s, t}$ but the closest one from $s'$; such a point $z$ must exist since each sub-segment of $e$ has length at most $\frac{\|e\|}{\kappa} < \frac{\|st\|}{25}$. Then, $\Psi_1$ charges the weight $\|st\|$ to segment $z t'$ which has length at least $\brac{\frac{1}{25} - \frac{1}{\kappa}}\|st\| > \frac{\|st\|}{26}$.
	
	If $e$ crosses both regions $A_{s, t}$ and $B_{s, t}$, then let $z_1\in Y\cap e$ be the Steiner point in $A_{s, t}$ which is the closest one from $s'$, and let $z_2\in Y\cap e$ be the Steiner point in $B_{s, t}$ which is the closest one from $t'$. Then, $\Psi_1$ charges weight $\|s t\|$ of $st$ to segment $z_1 z_2$ which has length at least $\brac{\frac{1}{4} + \frac{1}{25} - \frac{2}{\kappa}}\|st\| > \frac{\|st\|}{4}$.
 
	\item Otherwise, we will distribute the weight $\|s t\|$ among multiple edges in $E_\light$. Move along $\rho_{s, t}$ from $s$ to $t$ and let $p$ be the last vertex in $A_{s, t}$ and let $q$ be the first vertex in $B_{s, t}$. As $\|\pi_{s, t}\|\leq (1+\epsilon)\cdot \|st\|$, we know that
	$$\|\rho_{s, t}[p, q]\| \leq \|\proj_{st}(p)-\proj_{st}(q)\|  + \epsilon\cdot \|st\| \leq (1+10\epsilon)\cdot \|\proj_{st}(p)-\proj_{st}(q)\|,$$
  where $\rho_{s, t}[p, q]\overset{\mathrm{def}}{=} \rho_{s, t}\cap \pi_{s, t}[p, q]$ refers to the sub-sequence of $\rho_{s, t}$ between $p$ and $q$. Therefore, \Cref{angle-bound} yields
	$$\|E(\rho_{s, t}[p, q], st, 2\sqrt{10\epsilon})\|> 0.5\cdot \|\proj_{st}(p)-\proj_{st}(q)\|.$$
	Then, for each edge $e\in E(\rho_{s, t}[p, q], st, 2\sqrt{10\epsilon})$, let $\Psi_1$ charge an amount of $\frac{\|e\|}{\|\rho_{s, t}[p, q]\|}\cdot \|st\|$ from $st$ to edge $e$.
\end{enumerate}

For some technical reason, we need to be more formal about the properties of charging schemes.
\begin{definition}[weight charging]\label{replete}
    A \EMPH{replete} weight charging scheme $\Psi$ from an edge set $A$ to edge set $B$ is a mapping $\Psi: A\times B\rightarrow \mathbb{R}^+\cup \{0\}$ such that:
    \begin{itemize}
        \item for any $f\in B$, if $\Psi(e, f)>0$, then $\Psi(e, f)\geq \min\{\|f\|, \|e\|\}$;
        \item $\sum_{f\in B}\Psi(e, f) = \|e\|$ for any $e\in A$.
    \end{itemize}
    We say that $e$ charges to $f$ under $\Psi$ if $\Psi(e, f)>0$.
\end{definition}

We can show that $\Psi_1: E\times \brac{E_\light\cup E_\light^Y}\rightarrow \mathbb{R}^+\cup \{0\}$ is replete under \Cref{replete} for $A = E$ and $B = E_\light\cup E_\light^Y$.

\begin{claim}\label{const-charge}
    The charging scheme $\Psi_1$ is replete, and every edge $e$ receives $O( \|e\|)$ amount of charges from any edge $s t\in E$; furthermore, both endpoints of $e$ are in $X$ if $s t$ is type-(\romannumeral2) and case-(b).
\end{claim}
\begin{proof}
	If $s t$ is a type-(\romannumeral1) edge, then by the design of $\Psi_1$ we know that $\|st\|>\|e\|\geq \frac{\|st\|}{25\beta}$. So the amount of charges $e$ receives from $s t$ is between $\|e\|$ and $25\beta\cdot \|e\|$. If $s t$ is a type-(\romannumeral2) edge in case-(a), then we have $\|st\|>\|e\|\geq \frac{\|st\|}{30}$. Finally, if $s t$ is a type-(\romannumeral2) edge in case-(b), the on one hand   notice that $\|\rho_{s, t}[p, q]\|\leq \|\pi_{s, t}\| - \|sp\|\leq \brac{1+\eps - \brac{\frac{3}{8}-\frac{1}{50}}}\|st\| < \|st\|$. Consequently, $e$ is receives at least $\|e\|$ charges from $st$ under $\Psi_1$. On the other hand, 
   by the design of $\Psi_1$, the amount of charges $e$ has receives from $s t$ is at most
    $$\frac{\|e\|}{\|\rho_{s, t}[p, q]\|}\cdot \|st\|\leq \frac{2\|e\|}{\frac{1}{4} - \frac{1}{25}} < 10\|e\|.$$
    
    The second half of the statement holds since the charging scheme in case-(b) does not involve any Steiner points in the super-set $Y$.
\end{proof}

Similar to the sparsity analysis, we show that each edge $e$ in $E_\light$ receives $O\brac{\|e\|}$ amount of charges from type-(\romannumeral1) edges after the first pruning phase, and so this phase increases the weight by $O(\|E_\light\|\log\alpha)$.
\begin{claim}
    \label{clm:E1-log-size}
	After the first pruning phase, the total weight of new edges added to $E_1$ is  $O\brac{\|E_\light\|\log\alpha}$. Furthermore, the total weight of type-(\romannumeral1) edges in $E_1$ is $O\brac{\|E_\light\|\log\alpha}$.
\end{claim}
\begin{proof}
	During the first pruning phase, we show by induction on $i$ that at the beginning of the $i$-th sub-iteration, $\left\|E_1\cap E^{(\romannumeral1)}\right\|$ is at most $\left\|E_\light\right\|\alpha / 2^{i-1}$. For the basis when $i = 1$, this bound holds as $\alpha \geq \|E\| / \|E_\light\|$.
	
	For the inductive step, to bound $\left\|E_1\cap E^{(\romannumeral1)}\right\|$ after the $i$-th sub-iteration, we need to utilize our charging scheme $\Psi_1$ of type-(\romannumeral1) edges. Suppose otherwise that the total weight of  all type-(\romannumeral1) edges remaining in $E_1\cap E^{(\romannumeral1)}$ is at least $\left\|E_\light\right\|\alpha / 2^{i-1}$ after the $i$-th sub-iteration. Then, by the pigeon-hole principle, there exists a set $F$ of more than $25\beta\cdot\alpha / 2^{i-1}$ type-(\romannumeral1) edges currently in $E_1$ charging to the same edge $x y\in E_\light$. By our charging scheme, for each such edge $e\in F$, we have
    $$\|xy\| \leq \|e\| \leq 25\beta\cdot \|xy\| .$$ 
    Therefore, there exists an index $j\geq 0$ such that
	$$|F\cap L_j|\geq \frac{\alpha}{2^{i-1}\cdot \log_\beta{25\beta}} > \frac{\alpha}{2^i\kappa}.$$
	Now, consider the time during the first pruning phase when we were processing all edges in $L_j$.
	
	To reach a contradiction, it suffices to claim that at the moment, it must be that $P_{x, y}\supseteq F\cap L_j$. In fact, for any edge $s t\in F\cap L_j$, by the charging scheme, we know that
	$$\|sx\|  + \|xy\| + \|yt\| \leq (1+\epsilon)\|st\|. $$
	Therefore, the algorithm could have added $xy$ to $E_1$ as a new edge and remove the entire set $P_{x, y}$ leading to a contradiction. This completes the proof of the induction step.
	
 It remains to bound the total weight of new edges added to $E_1$. In the $i$-th sub-iteration, each time we add a new edge $e$ to $E_1$, we have decreased the weight $\left\|E_1\cap E^{(\romannumeral1)}\right\|$ by at least $\frac{\alpha}{2^i\kappa}\cdot\|e\|$. Since $\left\|E_1\cap E^{(\romannumeral1)}\right\|$ was at most $\|E_\light\|\alpha / 2^{i-1}$ at the beginning, we could increase the weight of $E$ by at most $O(\|E_\light\|)$ in this sub-iteration, and thus by at most $O(\|E_\light\|\log\alpha)$ overall.
\end{proof}

Now, let us analyze the second pruning phase. Using the same calculation as in \Cref{one-edge-per-level}, we can prove the following claim for type-(b) edges.
\begin{claim}\label{one-edge-per-level-lightness}
    Fix an edge $e\in E_\light^Y$ and a level $j\geq 0$. After the second pruning phase, $\Psi_1$ charged the weight of at most one type-(\romannumeral2) edge in $E_2\cap L_j$ to $e$.
\end{claim}

\paragraph{Charging scheme from $E_2\cap E$ to $E_\light^Y$.} However, \Cref{one-edge-per-level-lightness} is not enough to establish a lightness bound because an edge in $E_\light^Y$ could be charged multiple times across all different levels $j\geq 0$, and the charges need not form a geometric sum as in the case of sparsity. To bound the lightness of $E_2$, we describe another charging scheme $\Psi_2$, from $E_2\cap E^{(\romannumeral2)}$ to $E_\light^Y$, that works simultaneously with the second pruning phase. During the construction of the charging scheme $\Psi_2$, we will maintain the following invariant.
\begin{invariant}\label{charge-count}
 The charging scheme $\Psi_2$ is replete, and it charges every edge $e\in E_\light^Y$ at most $\kappa$ times, and each time $e$ receives at most $20\,\|e\|$ amount of charges.
\end{invariant}

At the beginning when $E_2\cap E^{(\romannumeral2)} = \emptyset$, no edges have been charged to $E_\light$ under the new charging scheme. The algorithm goes over all edges $s t\in E_1\cap E$ and decides whether $s t$ should stay in $E_2$ or not. For an edge $st$ of type-(b), let $F\subseteq E_\light^Y$ be the set of edges to which the charging scheme $\Psi_1$ distributed a positive portion of the weight $\|s t\|$, that is, $F=\{e\in E_\light^Y: \Psi_1(st,e)>0\}$.

\begin{claim}\label{proj-len}
	$\|F\|> \frac{1}{10}\|st\|$.
\end{claim}
\begin{proof}
	Since $s t$ is type-(b), by definition, $F$ contains all edges in the sub-sequence $\rho_{s, t}$ from the last vertex $p\in A_{s, t}$ to the first vertex $q\in B_{s, t}$. Therefore, the orthogonal projection of $F$ on $st$ has length 
 $$\|\proj_{st}(F)\|
 =\frac{1}{2}\,\|\proj_{st}(p) - \proj_{st}(q)\|
 \geq \frac{1}{2}\cdot\brac{\frac{1}{4} - \frac{1}{25}}\cdot \|st\| > \frac{1}{10}\cdot \|st\|.$$
    Therefore, $\|F\|\geq \|\proj_{st}(F)\| > \frac{1}{10}\|st\|$.
\end{proof}

Let $S\subseteq F$ be the subset of edges that $\Psi_2$ has already charged $\kappa$ times by edges prior to processing edge $s t$. 
If $\|S\|\leq \frac12\,\|F\|$, then the charging scheme $\Psi_2$ distributes the weight $\|st\|$ among $F\setminus S$ proportionally to their weight. Formally, for each edge $e\in F\setminus S$, if $\Psi_1$ charged $\lambda$ of the weight $\|st\|$ to $e$, then let $\Psi_2$ charge an amount of $\|F\| / \|F\setminus S\|\cdot \lambda \leq 2\lambda$ to $e$; that is, 
$$\Psi_2(st, e)\leftarrow \begin{cases}
    \frac{|F\|}{\|F\setminus S\|} \cdot \Psi_1(st, e)   &   e\in F\setminus S\\
    0   &   e\notin F\setminus S .
\end{cases}$$
Note that $\frac{|F\|}{\|F\setminus S\|} \cdot \Psi_1(st, e)\ \leq 2\Psi_1(st, e))$. By \Cref{const-charge}, this maintains \Cref{charge-count}.

Next, let us consider the harder case where $\|S\| > \frac12\,\|F\|$. In fact, we will show that $s t$ could not have been added to $E_1$; that is, this case never occurs.
\begin{lemma}\label{lightness}
    If $\|S\| > \frac12\, \|F\|$, then edge $s t$ would not be added to $E_2$; that is, we already have $\dist_{H_2}(s, t) \leq (1+\kappa^2\delta)\cdot\|st\|$.
\end{lemma}

\paragraph{Putting it all together.} Assuming \Cref{lightness}, we can show that every edge $e\in E_\light^Y$ receives at most $20\kappa \|e\|$ charges under $\Psi_2$ (cf.~\Cref{charge-count}). Together with \Cref{clm:E1-log-size}, we can show that $\|E_2\| = O(\|E_\light\|\log\alpha )$. Also, by \Cref{spanner-stretch}, the stretch of $H_2$ is at most $1+(\kappa+1)^2\delta$. Therefore, starting with $(X, E)$ being a $(1+\eps)$-spanner with $\delta=\eps$ and $\alpha = \epsilon^{-O(d)}$, if we iterate these two pruning phases $k$ times, we end up with a spanner with $\brac{1+\epsilon\cdot 2^{O(k)}}$-stretch and $O\brac{\log^{(k)}(1/\eps) + \log^{(k-1)}(d)}\cdot \|E_\light\|$ total weight, which concludes the lightness bound of \Cref{thm:tech}. Note that Algorithm \ref{greedy-prune} (Algorithm $\mathsf{GreedyPrune}$) repeats the two pruning phases for $k=O(\log^*(d/\epsilon))$ iterations, the result of which provides a spanner with the stretch and sparsity bounds of \Cref{thm:main}.

\paragraph{Proof of \Cref{lightness}.}
We break down the proof of \Cref{lightness} into a sequence of claims.

\begin{claim} \label{clm:mid-charge-edge}
    For each $e\in E_\light^Y$ that $\Psi_2$ has charged exactly $\kappa$ times,
    there exists an edge $\chi(e)\in E_2$ that $\Psi_1$ 
    has already charged
    to $e$ such that $\kappa\|e\|\leq \|\chi(e)\|\leq \frac{1}{\kappa}\|st\|$; furthermore, the angle between $\chi(e)$ and $e$ is bounded by $\angle\brac{\chi(e), e}\leq 2\sqrt{10\epsilon}$.
\end{claim}
\begin{proof}
 	According to  \Cref{one-edge-per-level-lightness}, all edges in $E_2$ that charge to $e$ under $\Psi_1$ are on different levels $L_i$. Therefore, by our construction of $\Psi_2$, the total number of levels that conver the range $\left[\frac{1}{\kappa} \|st\|, \kappa\|e\|\right] \subseteq \left[\frac{1}{\kappa} \|st\|, \|st\|\right]\cup \big[\|e\|, \kappa\|e\|\big]$ is at most $2\ceil{\log_\beta\kappa} < \kappa$. By the pigeonhole principle, there exists an edge $\chi(e)$, as claimed.
	
	For the second half of the statement, since $\chi(e)$ charges to $e$ under $\Psi_1$ and $\|\chi(e)\|\geq \kappa \|e\|$, it must be of type-(\romannumeral2) in case-(b). Hence, we have $\angle\brac{\chi(e), e}\leq 2\sqrt{10\epsilon}$ by design of $\Psi_1$.
\end{proof}

Let $P = \{\chi(e) : e\in S\}$. By \Cref{replete} and \Cref{charge-count}, we know that $\|P\|\geq \|S\|$. Next, we will prove that $s t$ is not added to $E_2$ in the second pruning phase. The proof consists of two steps. At a high level, in the first step, we will select a subset of edges $Q\subseteq P$ such that $\|Q\|\geq \Omega(\|st\|)$; in the second step, we will stitch these edges into an $st$-path in $E_2$ of total length at most $(1+\kappa^2\delta)\cdot\|st\|$, thus implying that $s t$ would be excluded from $E_2$.

To construct $Q$ starting with $Q = \emptyset$, iterate overall all indices $i = O(\log \alpha), \hdots, 0$ decrementally while adding edges to $Q$. In the $i$-th iteration, scan the segment $st$ from $s$ to $t$. Whenever we hit the projection $\proj_{st}(\chi(e))$ of some edge $\chi(e)\in P\cap L_i$, add $\chi(e)$ to $Q$, and remove all edges from $P$ that lie in the \EMPH{buffer region} of $\chi(e)$; that is, all edges $\chi(f)\in P$ whose projections $\proj_{st}(\chi(f))$ on $st$ are at distance at most $2\beta^{i+1}$ from $\proj_{st}\brac{\chi(e)}$.

From the construction process of $Q$, we see that the edges in $Q$ have pairwise disjoint projections on the line $st$. 
After we have constructed the edge set $Q$, we can order the edges in $Q$ according to their projections from $s$ to $t$ as $Q = \{a_1b_1, a_2b_2, \hdots, a_\ell b_\ell\}$. Next, let us define a candidate path $\gamma$ between $s$ and $t$ in the current $H_2$ (that does not yet include $st$).
\begin{definition}[path stitching]\label{stitch-path}
We define the $st$-path 
\[
    \gamma=(s\rightsquigarrow a_1)\circ a_1b_1\circ (b_1\rightsquigarrow a_2)\circ a_2b_2\circ\ldots\circ
    (b_{\ell-1}\rightsquigarrow a_\ell) \circ a_\ell b_\ell \circ(b_\ell\rightsquigarrow t),
\]
which contains the edges $a_ib_i$ for all $1\leq i\leq \ell-1$, 
and where $(x\rightsquigarrow  y)$ denotes the shortest path from $x$ to $y$ in the current graph $H_2$. 
\end{definition}
To reach a contradiction, our goal is to show that $\|\gamma\|\leq (1+\kappa^2\delta)\cdot\|st\|$.
First, we  show that the total length of the edges in $Q$ is $\Omega(\|st\|)$.
\begin{claim}\label{shortcuts}
	$\|Q\|\geq  \frac{1}{320}\|st\|$.
\end{claim}
\begin{proof}
    By \Cref{proj-len} and our assumption in \Cref{lightness}, we have $\|S\|\geq \frac12\, \|F\|\geq \frac{1}{20}\|st\|$. So it suffices to show that $\|Q\|\geq \frac{1}{16}\|S\|$. According to the construction of $Q$, for every edge $e\in S$, either $\chi(e)\in P$ was added to $Q$, or $\chi(e)$ was removed from $P$ because its projection $\proj_{st}(\chi(e))$ on line $st$ is at most $2\beta^{i+1}$ far away from the projection $\proj_{st}(\chi(f))$ of another edge $\chi(f)\in Q\cap \brac{L_i\cup L_{i+1}\cup\cdots}$ for some level $i\in \{0,\ldots ,O(\log \alpha)\}$. In the latter case, let us associate this edge $\chi(e)$ with $\chi(f)$, and let $P_f$ be the set of all edges associated with $\chi(f)$ under this definition. Since the construction procedure of $Q$ enumerates all indices $i$ from large to small, then $\|\chi(e)\| < \beta\cdot \|\chi(f)\|$ for all $\chi(e)\in P_f$. Since each point on segment $st$ is covered by projections of at most two  edges in $F$, the total length $\sum_{\chi(e)\in P_f}\|\proj_{st}(\chi(e))\|$ is bounded by $14\beta^{i+1}$.

    Noticing that $\angle\brac{e, \chi(e)}, \angle\brac{e, st}\leq 2\sqrt{10\eps}$, we have $\angle(\chi(e), st)\leq 4\sqrt{10\eps}$. Therefore, we have
    $$\|\proj_{st}(\chi(e))\|\geq \cos(4\sqrt{10\eps})\cdot \|\chi(e)\| = \brac{1-80\eps + O(\eps^2)}\cdot \|\chi(e)\| \geq \beta^{-1} \|\chi(e)\|.$$
    The last inequality holds when $\eps$ is below a certain constant threshold. Therefore, we have $14\beta^{i+1}\geq \sum_{\chi(e)\in P_f}\|\proj_{st}(\chi(e))\|\geq \beta^{-1}\sum_{\chi(e)\in P_f}\|\chi(e)\|$, or $\sum_{\chi(e)\in P_f}\|\chi(e)\|\leq 14\beta^{i+2} < 15\|\chi(f)\|$. In other words, each edge $\chi(f)$ added to $Q$ precludes a set of edges in $P$ from joining $Q$ which have total weight at most $15\|\chi(f)\|$, that is,
\begin{align*}
    P&=Q\cup \bigcup_{\chi(e)\in Q} P_{\chi(e)} 
        = \bigcup_{\chi(e)\in Q} 
        \left(\{\chi(e)\}\cup P_{\chi(e)}\right),\\
    \|P\| &= \sum_{\chi(e)\in Q} 
        \big( \|\chi(e)\| + \|P_{\chi(e)}\|\big) 
          \leq \sum_{\chi(e)\in Q} 16\, \|\chi(e)\| 
          =16\, \|Q\|.
\end{align*}
We can conclude that
    $$\|Q\|\geq \frac{1}{16}\|P\|
           \geq \frac{1}{16}\|S\|
           \geq \frac{1}{320}\|st\|.
           \qedhere$$
\end{proof}

To analyze the total length of $\gamma$, let us subdivide the path into multiple parts and analyze each part separately. For each edge $a_i b_i\in Q$, assume that $e_i = z_i w_i$ is the edge in $S$ such that $\chi(e_i) = a_i b_i$. Draw two $(d-1)$-dimensional hyper-planes perpendicular to line $st$ through $z_i$ and $w_i$ which intersects lines $st$ and $a_ib_i$ at $c_i, d_i$ and $p_i, q_i$, respectively. See \Cref{proj} for an illustration.
\begin{figure}
	\centering
	\begin{tikzpicture}[thick,scale=1.2]
	\draw (0, 0) node(1)[circle, draw, fill=black!50,
	inner sep=0pt, minimum width=6pt, label = $s$] {};
	\draw (12, 0) node(2)[circle, draw, fill=black!50,
	inner sep=0pt, minimum width=6pt,label = $t$] {};
	
	\draw (4, 3) node(3)[circle, draw, fill=black!50,
	inner sep=0pt, minimum width=6pt, label = $a_i$] {};
	\draw (8, 4) node(4)[circle, draw, fill=black!50,
	inner sep=0pt, minimum width=6pt,label = $b_i$] {};
	
	\draw (5.5, 2) node(5)[circle, draw, fill=black!50,
	inner sep=0pt, minimum width=6pt, label = {180: {$z_i$}}] {};
	\draw (6.5, 1.7) node(6)[circle, draw, fill=black!50,
	inner sep=0pt, minimum width=6pt,label = {0: {$w_i$}}] {};
	
	\draw (4, 0) node(7)[circle, draw, fill=black,
	inner sep=0pt, minimum width=4pt, label = {-90: {$u_i$}}] {};
	\draw (8, 0) node(8)[circle, draw, fill=black,
	inner sep=0pt, minimum width=4pt,label = {-90: {$v_i$}}] {};
	
	\draw (5.5, 0) node(9)[circle, draw, fill=black,
	inner sep=0pt, minimum width=4pt, label = {-90: {$p_i$}}] {};
	\draw (6.5, 0) node(10)[circle, draw, fill=black,
	inner sep=0pt, minimum width=4pt,label = {-90: {$q_i$}}] {};
		
	\draw [line width = 0.5mm] (1) to (2);
	\draw [line width = 0.5mm] (3) to (4);
	\draw [line width = 0.5mm] (5) to (6);
	
	\draw [line width = 0.2mm] (3) to (7);
	\draw [line width = 0.2mm] (4) to (8);
	
	\node(11)[circle, draw, fill=black,
	inner sep=0pt, minimum width=4pt, label = $c_i$] at (intersection of 5--9 and 3--4){};
	\node(12)[circle, draw, fill=black,
	inner sep=0pt, minimum width=4pt, label = $d_i$] at (intersection of 6--10 and 3--4){};
	
	\draw [line width = 0.2mm] (5) to (11);
	\draw [line width = 0.2mm] (6) to (12);
	\draw [line width = 0.2mm] (5) to (9);
	\draw [line width = 0.2mm] (6) to (10);
\end{tikzpicture}
	\caption{Auxiliary lines and points $c_i$, $d_i$, $p_i$, $q_i$, $u_i$, and $v_i$ to assist our analysis.}\label{proj}
\end{figure}
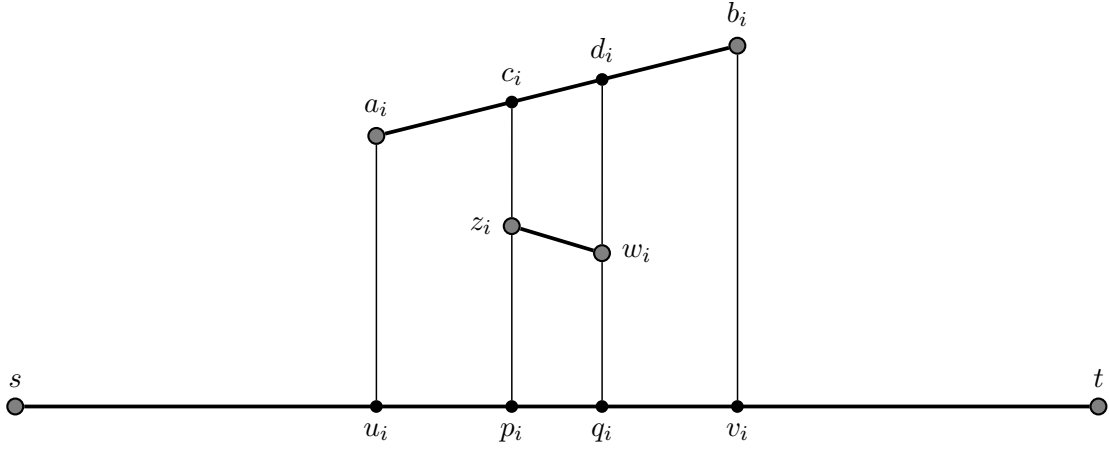

\begin{claim}\label{projection1}
    Both $c_i$ and $d_i$ lie on the segment $a_ib_i$; furthermore, $\min\{\|a_ic_i\|, \|d_ib_i\|\}\geq 0.25\,\|a_ib_i\|$.
\end{claim}
\begin{proof}
	It suffices to focus only on the inequality $\|a_ic_i\|\geq 0.25\,\|a_ib_i\|$. Let $c_i' = \proj_{a_ib_i}(z_i)$ be the projection of $z_i$ on line $a_ib_i$. As $a_ib_i$ is charging to $z_iw_i$, point $c_i'$ should land on the segment $a_ib_i$.
 
    Since $\angle(a_ib_i, st)\leq 4\sqrt{10\epsilon}$ and $z_ic_i$ lies in the hyperplane perpendicular to $st$, we have
	$$\|c_ic_i'\| = \left|\overrightarrow{c_iz_i} \cdot \frac{\overrightarrow{a_ib_i}}{\|a_ib_i\|} \right|\leq \sin(4\sqrt{10\epsilon})\cdot \|c_iz_i\| < 4\sqrt{10\epsilon} \cdot\|c_iz_i\|,$$
	which yields  
    $$\|c_ic_i'\|\leq \frac{4\sqrt{10\epsilon}}{\sqrt{1 - 160\epsilon}}\cdot \|a_ib_i\| .$$
	As $z_i$ lies in the ellipsoid $\left\{z\in \mathbb{R}^d : \|a_iz\| + \|b_iz\|\leq (1+\epsilon)\|a_ib_i\| \right\}$, we have $\|z_ic_i'\| < 2\sqrt{\epsilon}\cdot \|a_ib_i\|$. Therefore, $\|c_ic_i'\| < \frac{4\sqrt{10\epsilon}}{\sqrt{1 - 160\epsilon}}\|a_ib_i\| < 20\epsilon \|a_ib_i\|$. As $\|a_ic_i'\| / \|a_ib_i\| \in [0.355, 0.665]$, we have 
    $$\|a_ic_i\| \geq \|a_ic_i'\| - \|c_ic_i'\| > 0.355\|a_ib_i\| - 20\epsilon \|a_ib_i\| > 0.25\|a_ib_i\| .\qedhere$$
\end{proof}

\begin{claim}\label{projection2}
    Let $u_i = \proj_{st}(a_i)$ and $v_i = \proj_{st}(b_i)$; see \Cref{proj}. Then, both $u_i$ and $v_i$ lie on segment $st$. Furthermore, we have $\|su_i\|, \|v_it\| > \frac{1}{3}\|st\|$ and $\|a_iu_i\|, \|b_iv_i\|\leq 5\sqrt{\epsilon}\cdot\|st\|$.
\end{claim}
\begin{proof}
	By the charging scheme $\Psi_1$, we know that $p_i$ lies on $st$ and $\|sp_i\| \geq \brac{\frac{3}{8} - \frac{1}{50}}\|st\|$. Therefore, we have
	$$\|su_i\| \geq \|sp_i\| - \|a_ic_i\| >  \brac{\frac{3}{8} - \frac{1}{50}}\|st\| - \frac{1}{\kappa}\|st\| > \frac13\,\|st\|.$$
	Similarly we can prove that $\|v_ib\| >\frac13\,\|st\|$, and so both $u_i$ and $v_i$ lie on segment $st$.
	
	As for the length of $a_iu_i$, let $M$ be the projection matrix onto the hyperplane perpendicular to $st$, then the triangle inequality combined with $\|a_ib_i\| \leq \frac{1}{\kappa}\|st\|$ yields
\begin{align*}
		\|a_iu_i\| &= \|M(\overrightarrow{a_ic_i} + \overrightarrow{c_iz_i} + \overrightarrow{z_ip_i} + \overrightarrow{p_iu_i})\| \leq \|M\cdot \overrightarrow{a_ic_i}\| + \|c_iz_i\| + \|z_ip_i\|\\
		&\leq 4\sqrt{10\epsilon}\cdot \|a_ib_i\| + \sqrt{\epsilon}\cdot \|a_ib_i\| + \sqrt{\epsilon}\cdot \|st\| < 5\sqrt{\epsilon}\cdot \|st\|. \qedhere
\end{align*}
\end{proof}

Next, we are going to show that stitching edges $a_ib_i$ and $a_{i+1}b_{i+1}$ via a shortest path $b_i\rightsquigarrow a_{i+1}$ in $H_2$ does not incur too much error compared to $\pi_{s, t}[b_i,a_{i+1}]$ in $E_{\light}$. Refer to \Cref{stitch}.

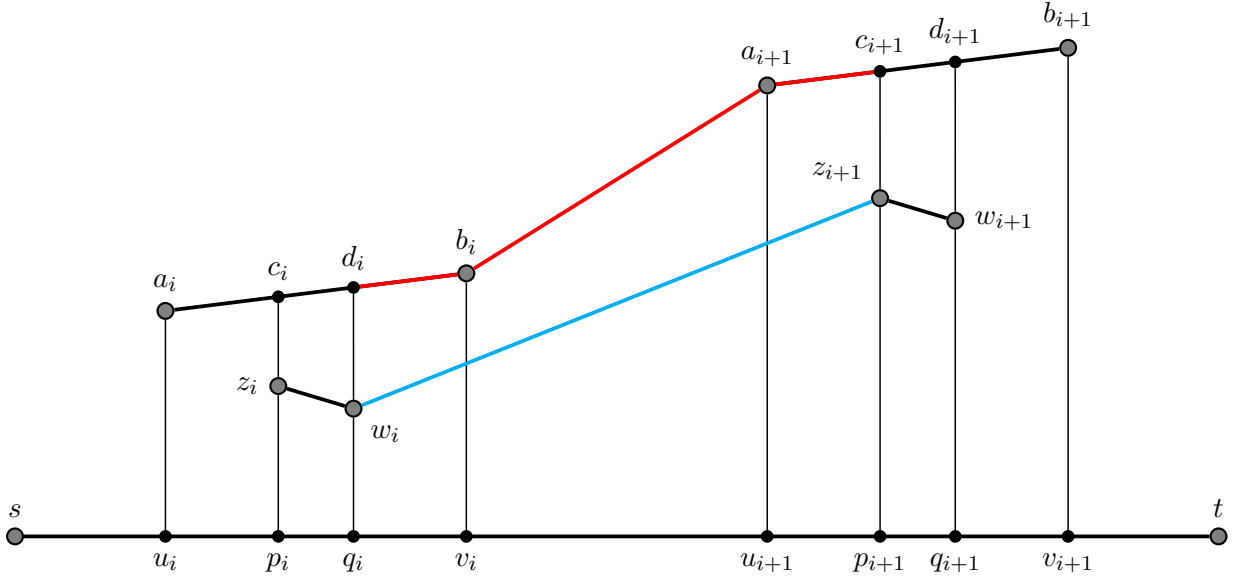
\begin{figure}
	\centering
	\begin{tikzpicture}[thick,scale=1]
	\draw (0, 0) node(1)[circle, draw, fill=black!50,
	inner sep=0pt, minimum width=6pt, label = $s$] {};
	\draw (16, 0) node(2)[circle, draw, fill=black!50,
	inner sep=0pt, minimum width=6pt,label = $t$] {};
	
	\draw (2, 3) node(3)[circle, draw, fill=black!50,
	inner sep=0pt, minimum width=6pt, label = $a_i$] {};
	\draw (6, 3.5) node(4)[circle, draw, fill=black!50,
	inner sep=0pt, minimum width=6pt,label = $b_i$] {};
	
	\draw (3.5, 2) node(5)[circle, draw, fill=black!50,
	inner sep=0pt, minimum width=6pt, label = {180: {$z_i$}}] {};
	\draw (4.5, 1.7) node(6)[circle, draw, fill=black!50,
	inner sep=0pt, minimum width=6pt,label = {-45: {$w_i$}}] {};
	
	\draw (2, 0) node(7)[circle, draw, fill=black,
	inner sep=0pt, minimum width=4pt, label = {-90: {$u_i$}}] {};
	\draw (6, 0) node(8)[circle, draw, fill=black,
	inner sep=0pt, minimum width=4pt,label = {-90: {$v_i$}}] {};
	
	\draw (3.5, 0) node(9)[circle, draw, fill=black,
	inner sep=0pt, minimum width=4pt, label = {-90: {$p_i$}}] {};
	\draw (4.5, 0) node(10)[circle, draw, fill=black,
	inner sep=0pt, minimum width=4pt,label = {-90: {$q_i$}}] {};
	
	\draw [line width = 0.5mm] (1) to (2);
	\draw [line width = 0.5mm] (3) to (4);
	\draw [line width = 0.5mm] (5) to (6);
	
	\draw [line width = 0.2mm] (3) to (7);
	\draw [line width = 0.2mm] (4) to (8);
	
	\node(11)[circle, draw, fill=black,
	inner sep=0pt, minimum width=4pt, label = $c_i$] at (intersection of 5--9 and 3--4){};
	\node(12)[circle, draw, fill=black,
	inner sep=0pt, minimum width=4pt, label = $d_i$] at (intersection of 6--10 and 3--4){};
	
	\draw [line width = 0.2mm] (5) to (11);
	\draw [line width = 0.2mm] (6) to (12);
	\draw [line width = 0.2mm] (5) to (9);
	\draw [line width = 0.2mm] (6) to (10);
	
	\draw (10, 6) node(13)[circle, draw, fill=black!50,
	inner sep=0pt, minimum width=6pt, label = $a_{i+1}$] {};
	\draw (14, 6.5) node(14)[circle, draw, fill=black!50,
	inner sep=0pt, minimum width=6pt,label = $b_{i+1}$] {};
	
	\draw (11.5, 4.5) node(15)[circle, draw, fill=black!50,
	inner sep=0pt, minimum width=6pt, label = {135: {$z_{i+1}$}}] {};
	\draw (12.5, 4.2) node(16)[circle, draw, fill=black!50,
	inner sep=0pt, minimum width=6pt,label = {0: {$w_{i+1}$}}] {};
	
	\draw (10, 0) node(17)[circle, draw, fill=black,
	inner sep=0pt, minimum width=4pt, label = {-90: {$u_{i+1}$}}] {};
	\draw (14, 0) node(18)[circle, draw, fill=black,
	inner sep=0pt, minimum width=4pt,label = {-90: {$v_{i+1}$}}] {};
	
	\draw (11.5, 0) node(19)[circle, draw, fill=black,
	inner sep=0pt, minimum width=4pt, label = {-90: {$p_{i+1}$}}] {};
	\draw (12.5, 0) node(20)[circle, draw, fill=black,
	inner sep=0pt, minimum width=4pt,label = {-90: {$q_{i+1}$}}] {};

	\draw [line width = 0.5mm] (13) to (14);
	\draw [line width = 0.5mm] (15) to (16);
	
	\draw [line width = 0.2mm] (13) to (17);
	\draw [line width = 0.2mm] (14) to (18);
	
	\node(21)[circle, draw, fill=black,
	inner sep=0pt, minimum width=4pt, label = $c_{i+1}$] at (intersection of 15--19 and 13--14){};
	\node(22)[circle, draw, fill=black,
	inner sep=0pt, minimum width=4pt, label = $d_{i+1}$] at (intersection of 16--20 and 13--14){};
	
	\draw [line width = 0.2mm] (15) to (21);
	\draw [line width = 0.2mm] (16) to (22);
	\draw [line width = 0.2mm] (15) to (19);
	\draw [line width = 0.2mm] (16) to (20);
	
	\draw [line width = 0.5mm, color = red] (4) to (13);
	\draw [line width = 0.5mm, color = red] (4) to (12);
	\draw [line width = 0.5mm, color = red] (13) to (21);
	\draw [line width = 0.5mm, color = cyan] (6) to (15);
\end{tikzpicture}
	\caption{Comparing the error of the stitched path with the error of $\rho_{s, t}$.}\label{stitch}
\end{figure}

\begin{claim}\label{stitch0}
	For any index $i$, we have
	$$\|c_id_i\| - \|p_iq_i\|\leq 100\eps\|p_iq_i\| \leq \brac{\|z_iw_i\| - \|p_iq_i\|} + 100\eps \|p_iq_i\|.$$
\end{claim}
\begin{proof}
    As $\angle(c_id_i, p_iq_i) \leq 4\sqrt{10\eps}$, we have
    $$\|c_id_i\| \leq \frac{\|p_iq_i\|}{\cos(4\sqrt{10\eps})}\leq \brac{1 + 80\eps + O(\eps^2)}\|p_iq_i\| < (1+100\eps)\|p_iq_i\|.$$
    Hence, we have
    $$\|c_id_i\| - \|p_iq_i\|\leq 100\eps\|p_iq_i\| \leq \brac{\|z_iw_i\| - \|p_iq_i\|} + 100\eps \|p_iq_i\|.
    \qedhere$$
\end{proof}

\begin{claim}\label{stitch1}
	For any index $i$, we have
	$$\|d_ib_i\circ b_ia_{i+1}\circ a_{i+1}c_{i+1}\| - \|q_ip_{i+1}\|\leq 4\cdot \brac{\|w_iz_{i+1}\| - \|q_ip_{i+1}\|} + 195\epsilon\cdot \|q_ip_{i+1}\|.$$
\end{claim}
\begin{proof}
    Without loss of generality, assume that $\|a_ib_i\|\geq \|a_{i+1}b_{i+1}\|$. Then, by the design of buffer regions, we know that $\|v_iu_{i+1}\|\geq 2\|a_ib_i\|$, and $\|v_iu_{i+1}\|\leq \|q_ip_{i+1}\| < 2\|v_iu_{i+1}\|$.
	
	Since $\angle(a_ib_i, st)$ and $\angle(a_{i+1}b_{i+1}, st)$ are bounded by $4\sqrt{10\epsilon}$, it follows that
\begin{align}\label{stitch1-inequ1}
   \|d_ib_i\| - \|q_iv_i\|
   & \leq \brac{\frac{1}{\cos(4\sqrt{10\eps})}-1}\cdot \|q_iv_i\|\leq \brac{80\eps + O(\eps^2)}\cdot \|q_iv_i\|\leq 41\eps \|q_ip_{i+1}\| ,\\
\|a_{i+1}c_{i+1}\| - \|u_{i+1}p_{i+1}\|&\leq \brac{\frac{1}{\cos(4\sqrt{10\eps})}-1}\cdot\|u_{i+1}p_{i+1}\|\nonumber\\
        &\leq \brac{80\eps + O(\eps^2)}\cdot\|u_{i+1}p_{i+1}\|< 41\eps \|q_ip_{i+1}\|.  \label{stitch1-inequ2}
\end{align}
   It remains to compare $\|b_ia_{i+1}\| - \|v_iu_{i+1}\|$ against $\|w_iz_{i+1}\| - \|q_ip_{i+1}\|$. Let $M$ be the matrix of the orthogonal projection onto the hyperplane perpendicular to line $st$. Define $h_1 = \|M\cdot \overrightarrow{b_ia_{i+1}}\|$ and $h_2 = \|M\cdot \overrightarrow{w_iz_{i+1}}\|$. First, we bound the difference between $h_1$ and $h_2$ using the triangle inequality
\begin{align*}
		|h_1 - h_2| &\leq \left|M\cdot (\overrightarrow{b_ia_{i+1}} - \overrightarrow{w_iz_{i+1}})\right|\\
		&= \left|M\cdot (\overrightarrow{w_id_i} + \overrightarrow{d_ib_i} + \overrightarrow{a_{i+1}c_{i+1}} + \overrightarrow{c_{i+1z_{i+1}}})\right|\\
		&\leq \sqrt{\epsilon}\cdot \|a_ib_i\|+ 4\sqrt{10\epsilon}\cdot \|a_ib_i\| + 4\sqrt{10\epsilon}\cdot \|a_{i+1}b_{i+1}\| + \sqrt{\epsilon}\cdot \|a_{i+1}b_{i+1}\|\\
		&< 15\sqrt{\epsilon}\|v_iu_{i+1}\| .
\end{align*}
    Next, to compare $\|b_ia_{i+1}\| - \|v_iu_{i+1}\|$ against $\|w_iz_{i+1}\| - \|q_ip_{i+1}\|$, we use the Pythagorean theorem combined with the identity $\sqrt{x}-\sqrt{y}=\frac{x-y}{\sqrt{x}+\sqrt{y}}$ and the definition of $h_1$ and $h_2$:
    \begin{align}
		\|b_ia_{i+1}\| - \|v_iu_{i+1}\| &= \sqrt{\|v_iu_{i+1}\|^2 + h_1^2} - \|v_iu_{i+1}\|\nonumber\\
        &\leq \sqrt{\|v_iu_{i+1}\|^2 + 2h_2^2 + 2(h_1-h_2)^2} - \|v_iu_{i+1}\|\nonumber\\
            &=  \frac{2h_2^2 + 2(h_1-h_2)^2}{\sqrt{\|v_iu_{i+1}\|^2 + 2h_2^2 + 2(h_1-h_2)^2} + \|v_iu_{i+1}\|}\nonumber\\
		&\leq  \frac{2h_2^2 + 225\epsilon\cdot\|v_iu_{i+1}\|^2}{\sqrt{\|v_iu_{i+1}\|^2 + 2h_2^2 + 225\epsilon\cdot\|v_iu_{i+1}\|^2} + \|v_iu_{i+1}\|}\nonumber\\
		&< \frac{2h_2^2}{\sqrt{\frac{1}{4}\|q_ip_{i+1}\|^2 + h_2^2} + \frac{1}{2}\|q_ip_{i+1}\|} + 112.5\epsilon\cdot\|v_i u_{i+1}\|\nonumber\\
		&< \frac{4h_2^2}{\sqrt{\|q_ip_{i+1}\|^2 + h_2^2} + \|q_ip_{i+1}\|} + 112.5\epsilon\cdot\|v_i u_{i+1}\|\nonumber\\
		&< 4\cdot (\|w_iz_{i+1}\| - \|q_ip_{i+1}\|) + 113\cdot\eps\|v_iu_{i+1}\|. \label{stitch1-inequ3}
	\end{align}
Taking a summation of \Cref{stitch1-inequ1,stitch1-inequ2,stitch1-inequ3}, we can conclude the proof; see \Cref{stitch} for an illustration.
\end{proof}

Similarly, we can prove an upper bound on the lengths of the prefix and suffix of the path $\gamma$.
\begin{claim}\label{stitch2}
	The following inequalities hold:
	$$\|sa_1\circ a_1c_1\| \leq (1+113\epsilon)\|sp_1\|,$$
	$$\|d_\ell b_\ell\circ b_\ell t\| \leq (1+113\epsilon)\|q_lt\|.$$
\end{claim}
\begin{proof}
	Let us focus on the first inequality; the second one can be proved in a symmetric manner. Since $\angle(a_1b_1, st)\leq 4\sqrt{10\epsilon}$, then we can show that
\begin{equation}\label{eq:end1}
\|a_1c_1\| - \|u_1p_1\|
\leq\brac{\frac{1}{\cos(4\sqrt{10\eps})}-1}\cdot\|u_1p_1\|\leq \brac{80\eps + O(\eps^2)}\|u_1p_1\|<100\epsilon\cdot \|u_1p_1\|.
\end{equation}
	According to \Cref{projection2}, we have $\|a_1u_1\|\leq 5\sqrt{\epsilon}\cdot \|st\|\leq 15\sqrt{\epsilon}\|su_1\|$. Therefore, we obtain
\begin{equation}\label{eq:end2}
	\|sa_1\| - \|su_1\| < 113\epsilon\cdot \|su_1\|.
 \end{equation}
Adding \Cref{eq:end1,eq:end2} finishes the proof.
\end{proof}

Let $\theta$ be the polygonal path defined as:
$$\theta = sa_1\circ a_1b_1\circ b_1a_2\circ a_2b_2\circ \cdots \circ b_{\ell-1} a_\ell\circ a_\ell b_\ell\circ b_\ell t $$
from $s$ to $t$, passing through all edges in $Q=\{a_1b_1,\ldots , a_\ell b_\ell\}$. Taking a summation over all indices $i\in \{1,\ldots ,\ell\}$ and using \Cref{stitch0}, \Cref{stitch1} and \Cref{stitch2}, we have
$$\|\theta\| - \|st\|\leq 4(\|\pi_{s, t}\| - \|st\|) + 195\epsilon\cdot \|st\| < 200\epsilon\cdot \|st\|,$$
or equivalently, 
\begin{equation}\label{eq:theta}
    \|\theta\|\leq (1 + 200\epsilon)\cdot \|st\|. 
\end{equation}
To compare the weight of paths $\theta$ and $\gamma$, we need to argue that each edge of $\theta$ is closely approximated by a shortest path in the current graph $H_2$. (Note that we cannot directly apply \Cref{spanner-stretch} here, because we are analyzing the state of $H_2$ during the execution of the second pruning phase, not at the end.)
\begin{claim}\label{H2-stretch}
    Any edge $e$ of $\theta$ from $\{sa_1, a_1b_1, \cdots, a_\ell b_\ell, b_\ell t\}$ has length at most 
    $$\|e\|\leq 0.9\,\|st\| < \frac{1}{1+\Delta(\kappa, \delta)}\cdot \|st\|.$$
\end{claim}
\begin{proof}
    By the design of our charging schemes, $\Psi_1$ and $\Psi_2$, the projections of all $p_i$ and $q_i$ are on the segment $st$; and more importantly
    $$\|sp_i\|, \|tq_i\|
    \geq \brac{\frac{3}{8} - \frac{1}{50}}\cdot \|st\| 
    > 0.3\,\|st\|.$$
    Therefore, we have
    $$\|sa_1\|\geq \|sp_1\| - \|u_1p_i\|
    \geq \|sp_1\| - \|a_1b_1\| 
    \geq 0.3\,\|st\| - \frac{1}{\kappa}\|st\| 
    > 0.2\,\|st\|.$$
    Similarly, we also have $\|b_\ell t\| > 0.2\,\|st\|$. Since $\|\theta\| < (1+200\eps)\|st\|$ by \Cref{eq:theta}, the length of any edge $e\in \{sa_1, a_1b_1, \ldots, a_\ell b_\ell, b_\ell t\}$ is bounded by 
    $$\|e\| < (1+200\eps - 0.2)\|st\| < 0.9\,\|st\|
    < \frac{1}{1+\Delta(\kappa, \delta)}\cdot \|st\|\qquad,
    $$
    as claimed.
\end{proof}

According to \Cref{spanner-stretch}, all edges in the path $\theta$ will be preserved up to a stretch factor of $(1+\delta)(1+\kappa\delta)(1+\kappa^2\delta) = 1+\Delta(\kappa, \delta)$ in $H_2$ when the second pruning phase finishes. Since the second phase adds edges to $E_2$ in non-decreasing order of
length and by \Cref{H2-stretch}, each edge in $\theta$ is already preserved in $H_2$ at the time when the algorithm examines edge $st$. 

Therefore, using \Cref{shortcuts} and recalling the definition of $\gamma$ (cf.~\Cref{stitch-path}), we have
\begin{align*}
	\|\gamma\|  &= \|\theta\| + \big(\dist_{H_2}(s, a_1) - \|sa_1\|\big) + \big(\dist_{H_2}(t, b_\ell) - \|tb_\ell\| \big)
    + \sum_{i=1}^{\ell-1}\big(\dist_{H_2}(b_i, a_{i+1}) - \|b_ia_{i+1}\|\big)\\
	&\leq \|\theta\| + \Delta(\kappa, \delta)\cdot \brac{\|sa_1\| + \|tb_\ell\| + \sum_{i=1}^{\ell-1}\|b_ia_{i+1}\|}\\
	&= \|\theta\| + \Delta(\kappa, \delta) \cdot \brac{\|\theta\| - \sum_{i=1}^\ell\|a_ib_i\|}\\
	&\leq \|\theta\| + \Delta(\kappa, \delta)\cdot \brac{\|\theta\| - \frac{1}{320}\|st\|}\\
	&\leq \brac{(1+200\epsilon)(1+\Delta(\kappa, \delta)) - \frac{\Delta(\kappa, \delta)}{320}}\|st\|\\
        &= \brac{1 + 200\eps + \brac{1 + 200\eps - \frac{1}{320}}\cdot \Delta(\kappa, \delta)} \|st\|\\
        &\leq \brac{1 + \brac{1 + 200\eps +\frac{200}{(\kappa+1)^2}- \frac{1}{320}}\cdot (\kappa+1)^2\delta} \|st\|\\
	&< (1+\kappa^2\delta)\|st\|.
\end{align*}
Here we have used $\kappa = 10^4$, \Cref{small-delta}, \Cref{eq:theta}, and the inequalities $\eps\leq \delta < \kappa^{-5}$ and $\Delta(\kappa, \delta) < (\kappa+1)^2\delta$. This concludes the proof of \Cref{lightness}. 
\section{Fast Implementation} \label{sec:fast}

In this section, we provide a fast implementation of our algorithm. For the input spanner, we use the construction of Gudmundson et al.~\cite{GLN02}. We follow the original two pruning phases with some modifications. In the first pruning phase, we only add edges in the net-tree spanner of Chan et al.~\cite{CGMZ16}. In the second pruning phase, we modify the technique of Das and Narasimhan in \cite{DN97} when checking each type-(\romannumeral2) edge in $E_1$ and use the hierarchical net structure in \cite{CGMZ16} to find a helper edge.

\subsection{Net-Tree Spanner}

An \EMPH{$r$-net} of a metric $(X,\delta)$ is a subset $N$ of $X$ such that the distance between any two points in $N$ is greater than $r$ and for every point in $X$, its distance to the closest point in $N$ is at most $r$. Let $\Phi$ be the aspect ratio of $X$ and $r_i = 2^i$ with $i$ being any positive integer. A hierarchy $X = N_0 \supseteq N_1 \supseteq N_2 \supseteq \cdots \supseteq N_{\log \Phi}$ is a \emph{hierarchical net} of $X$ if $N_{i + 1}$ is a $r_{i + 1}$-net of $N_{i}$. 

Consider a hierarchical net $X = N_0 \supseteq N_1 \supseteq N_2 \supseteq \cdots \supseteq N_{\log \Phi}$ of $X$. A \EMPH{net tree} $T$ of $X$ is a tree that connects each point in $N_i$ to its closest net point in $N_{i + 1}$. If a point appears in multiple nets, we treat each appearance as a copy. For each point $u \in N_i$, we use the notation $(u, i)$ for the node of the net tree corresponding to $u$ at level $i$. 

From the net-tree $T$ of $X$, Chan et al.~\cite{CGMZ16} construct a $(1+\eps)$-spanner for $X$ by connecting all pairs of net points in $N_i$ at distance at most $\left(\frac{4}{\eps} + 32\right)r_i$. We call such pairs \EMPH{cross edges}. We refer to the spanner construction in \cite{CGMZ16} as a $(1 + \eps)$ net-tree spanner $G$ of $X$. We use edges in $G$ to guide our construction. The key to our construction is the \EMPH{approximate edge} of each pair.

\begin{definition}
    For each $uv \in \binom{X}{2}$, let $i_{uv}$ be the lowest level such that there exists a cross edge $u'v'$ such that $(u', i_{uv})$ and $(v', i_{uv})$ are the ancestors of $(u, 0)$ and $(v, 0)$, respectively. The edge $u'v'$ is called 
    the \EMPH{approximate edge} of $uv$ and $i_{uv}$ is the \EMPH{approximate level} of $uv$.
\end{definition}

We show that the approximate edge of $uv$ has both endpoints close to those of $uv$. For any graph $F$, let $V(F)$ and $E(F)$ denote the vertex set and the edge set of $F$, respectively. For any point $x\in \mathbb{R}^d$ and $r > 0$, let $\ball(x, r)$ be the Euclidean ball centered at $x$ with radius $r$.

\begin{claim}
    \label{clm:approx-edge}
    For an edge $uv$ of length $r$, let $u'v' \in E(G)$ be the approximate edge of $uv$. Then, $u' \in \ball(u, \eps r)$ and $v' \in \ball(v, \eps r)$.
\end{claim}

\begin{proof}
    Let $i$ be the approximate level of $uv$ and let $(u'', i - 1), (v'', i- 1)$ be the ancestors of $(u, 0)$, $(v, 0)$ in the net-tree at level $i - 1$. The net-tree spanner construction yields $u'v' \in E(G)$. Using the triangle inequality, we get
    \begin{equation}
        \begin{split}
            \|uv\| &\geq \|u''v''\| - \|uu''\| - \|vv''\| \geq \left(\frac{4}{\eps} + 32\right)r_{i - 1} - 4\cdot r_{i - 1} \geq \frac{2}{\eps} \cdot r_i .
        \end{split} 
    \end{equation}
    Using geometric series, we obtain $\|uu'\| \leq 2\cdot r_i \leq \eps\|uv\|$, consequently $u' \in \ball(u, \eps r)$. Similarly, we can show that $v' \in \ball(v, \eps r)$.
\end{proof}

\subsection{Fast Implementation}
We are given a set $X\subset \mathbb{R}^d$ be a set of $n$ points, and a sufficiently small $\eps>0$ satisfying \Cref{eq:relationeps}. We start with a spanner with sparsity $\eps^{-O(d)}$ and lightness $\eps^{-O(d)}$. Such a spanner can be constructed in $O_{\eps,d}(n\log n)$ time \cite{GLN02}. Our construction consists of $k \leq \log^*{(d/\eps}) + O(1)$ iterations of two pruning phases.
Assume that at the beginning of the current iteration, we have a $(1 + \delta)$-spanner $H = (X, E)$. We use the same notation for $H_1$, $H_2$, $E_1$ and $E_2$ as in the original algorithm.

\paragraph{Classification of edges in $E$.} For each edge $(s, t)$ in $E$, we define the approximate set $A'_{s, t}$ of $A_{s, t}$ as follow: Let $h$ be the approximate level of $(s, t)$ and $N(A_{s, t})$ be the set containing all net points $w\in N_h$ at level $h$ such that $\ball(w, 2r_h) \cap A_{s, t} \neq \emptyset$. Let $A'_{s, t} = \bigcup_{w\in N(A_{s, t})}\ball(w, 2r_h)$. Since $r_h \leq \frac{\eps}{2} \|st\|$, then for every $x \in A'_{s, t}$, we obtain $\frac{\|s - \proj_{st}(x)\| }{\|st\| } \in \left[\frac{3}{8} - \frac{1}{50} - \eps, \frac{3}{8} + \frac{1}{50} + \eps\right]$ and $\frac{\|t - \proj_{st}(x)\| }{\|st\| } \in \left[\frac{5}{8} - \frac{1}{50} - \eps, \frac{5}{8} + \frac{1}{50} + \eps\right]$. Similarly, we define $B'_{s, t}$ based on the region $B_{s, t}$. An edge $(s, t)$ is type-(\romannumeral1) if  $A'_{s, t}$ or $B'_{s, t}$ is empty. Otherwise, $(s, t)$ is a type-(\romannumeral2) edge. We also denote the set of type-(\romannumeral1) and type-(\romannumeral2) edges by $E^{(\romannumeral1)}$ and $E^{(\romannumeral2)}$, respectively.

\paragraph{First pruning phase. } Recall that for every pair $\{x, y\}\subset X$, if $|P_{x, y}| \geq \frac{\alpha}{2^i\kappa}$ (for each sub-iteration $i$), we delete every edge in $P_{x, y}$. However, since the total number of pairs we need to check is $\Theta(n^2)$, this step would lead to a quadratic running time for each iteration. Thus, instead of considering all pairs of vertices $\{x, y\}\subset X$, we only consider the edges in the net-tree spanner $G$. Assuming that $\|xy\| \geq \beta^j/25$, we define an approximate set 
\begin{equation}
    P'_{x, y} = \left\{st \in L_j \cap E_1 \cap E^{(\romannumeral1)} : \|sx\| + \|xy\| + \|yt\| \leq (1 + 5\eps)\|st\|\right\}.
\end{equation}

The difference between $P_{x, y}$ and $P'_{x, y}$ is the distortion ($1 + \eps$ versus $1 + 5\eps$). If $|P'_{x, y}| \geq \frac{\alpha}{2^j\kappa}$, then add $xy$ to $E_1$ as a new edge, and remove all type-(\romannumeral1) edges in $P'_{x, y}$ from $E_1$. 

\paragraph{Second pruning phase. } For each type-(\romannumeral2) edge in $E_1$, we check whether the distance between two endpoints is already preserved by edges in $E_2$. This can be done in nearly linear time using the technique in \cite{DN97}. Our method is similar to \cite{DN97}, except we only consider each edge for $\Theta(\log{n})$ levels before contracting it. 

Let $E_{1, i} = \{st \in E_1 : 2^i \leq ||st|| \leq 2^{i + 1}\}$ and $E_{2, i} = \{st \in E_2: 2^i \leq ||st|| \leq 2^{i + 1}\}$ for $i \in [0, \log \Phi + 1]$, $E_{2, < i}$ be the union of $E_{2, 0}, E_{2, 1}, \ldots E_{2, i - 1}$ and $H_{2, < i} = (X, E_{2, <i})$. For each $i$, we build a cluster graph $F_i$ with vertex set $X$. We greedily construct the set of balls with radius $\eps 2^i$ in $H_{2, <i}$ (each ball is a cluster) such that the distance between any two centers in $H_{2, <i}$ is at least $\eps 2^i$ and the set of balls covers $X$.
Let $\{C_1, C_2, \ldots\}$ be the set of clusters. Note that one point might belong to multiple clusters.
There are two types of edges in $F_i$: inter-cluster and intra-cluster. Intra-cluster edges are edges between a point $u$ and a center $v$ of a cluster $C$ such that $u \in C$. The weight of $uv$ is $\dist_{H_2}(u, v)$. We add all intra-cluster edges to $F_i$. An inter-cluster edge is an edge between two cluster centers. There is an edge between two cluster centers $v_j$ and $v_{j'}$ if and only if $\dist_{H_2}(v_j, v_{j'}) \leq 2^i$ or there exists an edge in $E_{2, < i}$ from a point in $C_j$ to a point in $C_{j'}$, where $C_j$ and $C_{j'}$ are the clusters corresponding to the centers $v_j$ and $v_{j'}$. The inter-cluster edge $v_jv_{j'}$ has weight $\dist_{H_2}(v_j, v_{j'})$ in the first case and $\min_{u \in C_j, w \in C_{j'}, uw \in E_2}(\dist_{H_2}(u, v_j) + \dist_{H_2}(w, v_{j'}) + \|uw\|)$ in the second case. We use Dijkstra's algorithm to compute all the intra-cluster and inter-cluster edges. 

Our implementation of the second pruning phase
runs in $\log{\Phi} + 1$ iterations. At iteration $i$, we first build the cluster graph $F_i$ from $H_{2, <i}$. We consider edges in $E_{1, i}$ in increasing order of length. For each $st \in E_{1, i}$, if $st$ is not type-(\romannumeral2), we simply add $st$ to $E_2$. Otherwise, we check all paths from $s$ to $t$ within $O(1)$ hops in $F_i$ that contains at most $2$ intra-cluster edges. If the length of the shortest of those paths is less than or equal to $(1 + \kappa^2\delta)(1 + \eps)\|st\|$, skip that edge. Otherwise, we add $st$ and its helper edge to $E_2$. 
After an edge $st$ is added to $E_2$, we add inter-cluster edges between the centers of all clusters containing $u$ to the centers of all clusters containing $v$. For any center $c_s$ of a cluster containing $s$ and any center $c_t$ of a cluster containing $t$, we add an edge $c_sc_t$ of weight $\dist_{F_i}(c_s, s) + \|st\| + \dist_{F_i}(t, c_t)$.

However, computing all intra-cluster edges for all $F_i$ is expensive since one will have to run a single source shortest path (SSSP) algorithm for all cluster centers. Instead, we contract all edges of lengths less than or equal to $\frac{2^i\eps^{2}}{n}$. The total weight of all those edges in any path in $F_i$ is $2^i\eps^2$. Hence, for every path of length at least $2^i$, the total weight of contracted edges is significantly smaller compared to the total length. We then construct the graph $F'_i$ similar to $F_i$. Let $i' = i - \log{(n\eps^{-2})}$ and $E'_{2, < i} = \bigcup_{k = i'}^{i - 1}E_{2, k}$, we compute the graph $F'_i$ as follow: first, contract all edges in $E_{2, < i'}$ and second, compute $F'_i$ from $E'_{2, < i}$ similar to $F_i$. Instead of checking $\dist_{E_2}(u, v) \leq (1 + \kappa^2\delta)\|uv\|$ as in the original algorithm, we check whether $\dist_{F'_i}(s, t) + \eps^22^i \leq (1 + \eps)(1 + \kappa^2\delta)\|st\|$. If true, we skip the edge $st$. Otherwise, we keep $st$ and add its approximate helper edge to $E_2$. The approximate helper edge is any edge between a net point in $N(A_{s, t})$ to a net point in $N(B_{s, t})$.

After the iteration (including two pruning phases), we update  $\delta \leftarrow \Delta'(\kappa, \delta) = (1 + \delta)(1 + \kappa\delta)(1 + 5\eps)(1 + \kappa^2\delta)(1 + \eps) - 1$ and $\alpha \leftarrow  \log \alpha$.
\subsection{Analysis}

\subsubsection{Stretch} 

We prove that \Cref{edge-stretch} still holds in our implementation. After the first pruning phase, for any edge $st \in E$, we are guaranteed that $\dist_{H_1} \leq (1 + \kappa\delta)(1 + 5\eps) \cdot ||st||$. Throughout the execution of the second phase, for any edge $st \in E_1 \cap E$, we have $\dist_{H_2}(s, t) \leq (1 + \kappa^2\delta)(1 + \eps) \cdot \|st\|$.

\begin{claim}
    \label{clm:impl-edge-stretch}
    Throughout the execution of the first pruning phase, for any edge $st\in E$, we are guaranteed that $\dist_{H_1}(s, t)\leq (1+\kappa\delta)(1 + 5\eps)\cdot \|st\|$. Also, throughout the execution of the second phase, for any edge $st\in E_1\cap E$, we have $\dist_{H_2}(s, t)\leq (1+\kappa^2\delta)(1 + \eps)\cdot \|st\| $.
\end{claim}

The bound on the stretch after the first pruning phase can be proven using the same argument as in \Cref{edge-stretch}. We then focus on the second pruning phase. \Cref{clm:good-cluster-path} shows that if there is a good approximation path from $s$ to $t$ in the cluster graph $F_i$, there is a good approximation path from $s$ to $t$ in the current greedy spanner of $E_2$.

\begin{claim}[Lemma 3.2 and Lemma 3.3 \cite{DN97}]
    \label{clm:good-cluster-path} Let $st$ be a type-(\romannumeral2) edge in $E_1$ such that $2^i \leq \|st\| \leq 2^{i + 1}$ for some $i$ and $F_i$ be the cluster graph constructed above. Let $P$ be the shortest path from $s$ to $t$ in $H_{2}$ and $P_C$ be the shortest path from $s$ to $t$ in $F_h$. Then, $1 \leq \frac{\|P_C\|}{\|P\|} \leq 1 + \Theta(\eps)$. 
\end{claim}

We observe that the distance in $F'_i$ is approximately the distance in $F_i$.

\begin{observation}
    \label{obs:approx-contract-dist}
    For all $x,y\in X$, we have
    $\dist_{F'_i}(s, t) \leq \dist_{F_i}(s, t) \leq \dist_{F'_i}(s, t) + 2^i\eps^2$. 
\end{observation}

Now \Cref{clm:good-cluster-path}, combined with \Cref{obs:approx-contract-dist},  implies the following.
\begin{claim}
    \label{clm:good-cluster-path-2} Let $s$ and $t$ be two points such that $2^i \leq \|st\| \leq 2^{i + 1}$, and Let $F'_i$ be the cluster graph constructed above. Let $P$ the shortest path from $s$ to $t$ in $H_{2} = (X, E_{2})$ and $P'_C$ be the shortest path from $s$ to $t$ in $F'_i$. Then $1 \leq \frac{\|P'_C\| + 2^i\eps^2}{\|P\|} \leq 1 + \eps$.
\end{claim}

\begin{proof}
    Let $P_C$ be the shortest path from $s$ to $t$ in $F_i$. By \Cref{obs:approx-contract-dist}, we have $\|P'_C\| \leq \|P_C\| \leq \|P'_C\| + 2^i\eps^2$. Combined with \Cref{clm:good-cluster-path}, we obtain $1 \leq \frac{\|P'_C\| + 2^i\eps^2}{\|P\|} \leq 1 + \Theta(\eps)$. The last term $\Theta(\eps)$ can be reduced to $\eps$ by scaling the cluster radius by a suitable constant. 
\end{proof}

We have the following corollary:

\begin{corollary}
    Let $st$ be a type-(\romannumeral2) edge in $E_1$ such that $2^i \leq \|st\| \leq 2^{i + 1}$. If $\dist_{H_{2}}(s, t) \leq (1 + \kappa^2\delta)\|st\|$, then $st$ is not added to $E_2$. 
\end{corollary}

\begin{proof}
    By \Cref{clm:good-cluster-path-2}, the path from $s$ to $t$ in $F'_i$ has length at most $(1 + \eps)\dist_{H_2}(s, t) - \eps^22^i \leq (1 + \kappa^2\delta)(1 + \eps)\|st\| - \eps^22^i$. Hence, $st$ is not added to $E_2$ by our algorithm. 
\end{proof}

Then, for any edge $st$ not added to $E_2$, the distance from $s$ to $t$ in the spanner is at most $(1 + \kappa^2\delta)(1 + \eps)\|st\| + \eps^22^i$. By proper scale of the cluster radius, we got $\dist_{H_2}(s, t) \leq (1 + \kappa^2\delta)(1 + \eps)\|st\|$. This completes the proof of \Cref{clm:impl-edge-stretch}.

Recall that after each iteration, we update $\delta = \Delta'(\kappa, \delta)$ with  $\Delta'(\kappa, \delta) = (1 + \delta)(1 + \kappa\delta)(1 + 5\eps)(1 + \kappa^2\delta)(1 + \eps) - 1$. Observe that
\begin{equation}
    \begin{split}
        \Delta'(\kappa, \delta) &\leq (1 + \eps)(1 + 5\eps)(1 + \delta)(1 + \kappa\delta)(1 + \kappa^2\delta) - 1\\
        &\leq (1 + 7\eps)(1 + \delta)(1 + \kappa\delta)(1 + \kappa^2\delta) - 1 \leq (\kappa + 1)^2\delta
    \end{split}
\end{equation}
 for sufficiently small $\eps$. Therefore, after the update at the end of the $i$-th iteration, we have $\delta \leq (\kappa + 1)^{2i}\eps$. After $O(\log^*(d/\eps))$ iterations, the stretch is bounded by $(\kappa + 1)^{O(\log^*(d/\eps))}\eps \leq \kappa^{-5}$. 

\subsubsection{Sparsity}
The charging scheme for sparsity is similar to $\Psi_0$. The only change we make for a fast implementation, compared to the original $\Psi_0$, is that we use $A'_{s, t}$ and $B'_{s, t}$ instead of $A_{s, t}$ and $B_{s, t}$. For each edge $st \in E$, let $\pi_{s, t}$ be a spanner path in $E_\sparse$ between $s, t$.  If $st$ is a type-(\romannumeral1) edge, then $A'_{s, t}$ or $B'_{s, t}$ is empty, implying that $A_{s, t}$ or $B_{s, t}$ is empty since $A'_{s, t}$ and $B'_{s, t}$ contain $A_{s, t}$ and $B_{s, t}$, respectively. Therefore, there must be a single edge $e$ in $\pi_{s, t}$ that crosses the region $A_{s, t}$ or $B_{s, t}$. We charge the edge $st$ to $e$. If $st$ is a type-(\romannumeral2) edge, there are two cases:

\begin{enumerate}[leftmargin=*]
    \item If either $\pi_{s, t} \cap A'_{s, t}$ or $\pi_{s, t} \cap B'_{s, t}$ is empty. In that case, there is an edge $e = s't'$ on $\pi_{s, t}$ that crosses either $A'_{s, t}$ or $B'_{s, t}$ (and hence cross either $A_{s, t}$ or $B_{s, t}$). If $e$ only crosses one of the two regions (say $A_{s, t}$), then we have
	$$\|\proj_{st}(s) - \proj_{st}(s')\| < \brac{\frac{3}{8} - \frac{1}{50}}\cdot \|st\|,$$
	$$\brac{\frac{3}{8} - \frac{1}{50}}\cdot \|st\| < \|\proj_{st}(t') - \proj_{st}(t)\| < \brac{\frac{5}{8} - \frac{1}{50}}\cdot \|st\|.$$
	Recall that for each edge $e \in E_\sparse$, we divide $e$ evenly into $\kappa$ sub-segments by adding at most $\kappa - 1$ Steiner points on $e$ and $Y$ is the union of $X$ and the set of added Steiner points. Let $z \in Y\cap e$ be the Steiner points in $A_{s, t}$ on segment $e$ that is closest to $s'$; such a point $z$ must exist since each sub-segment of $e$ has length at most $\frac{\|e\|}{\kappa} < \frac{\|st\|}{25}$. Then, charge $st$ to segment $zt'$ which has length at least $\brac{\frac{1}{25} - \frac{1}{\kappa}}\|st\| > \frac{\|st\|}{26}$.
	
	If $e$ crosses both regions $A_{s, t}, B_{s, t}$, then let $z_1\in Y\cap e$ be the Steiner point in $A_{s, t}$ which is the closest one from $s'$, and let $z_2\in Y\cap e$ be the Steiner point in $B_{s, t}$ which is the closest one from $t'$. Then, charge $st$ to segment $z_1z_2$ which has length at least $\brac{\frac{1}{4} + \frac{1}{25} - \frac{2}{\kappa}}\|st\| > \frac{\|st\|}{4}$.

    \item Otherwise, we charge $st$ fractionally to set of edges in $E_\sparse$. Move along $\pi_{s, t}$ from $s$ to $t$ and let $p$ be the last vertex in $A'_{s, t}$ and let $q$ be the first vertex in $B'_{s, t}$. As $\|\pi_{s, t}\|\leq (1+\epsilon)\cdot \|st\|$, we know that
    \begin{align*}
            \|\pi_{s, t}[p, q]\| &\leq \|\proj_{st}(p)-\proj_{st}(q)\| + \epsilon\cdot \|st\|\\ 
            &\leq \|\proj_{st}(p)-\proj_{st}(q)\| + \eps \cdot \frac{1}{\left(\frac{5}{8} - \frac{1}{50} - \eps\right) - \left(\frac{3}{8} + \frac{1}{50} + \eps\right)} \cdot \|\proj_{st}(p)-\proj_{st}(q)\|\\
            &\leq (1+10\epsilon)\cdot \|\proj_{st}(p)-\proj_{st}(q)\| .
    \end{align*}
	for sufficiently small $\eps$. Therefore,  \Cref{angle-bound}
	yields
	$$\|E(\pi_{s, t}[p, q], st, 2\sqrt{10\epsilon})\|> 0.5\cdot \|\proj_{st}(p)-\proj_{st}(q)\| .$$
	Then, for each edge $e\in E(\pi_{s, t}[p, q], st, 2\sqrt{10\epsilon})$, charge a fraction of $\frac{2\|e\| }{\|\proj_{st}(p)-\proj_{st}(q)\| }$ of edge $st$ to edge $e$.
\end{enumerate}

We then have our modified \Cref{angle}. 
\begin{claim}
    \label{clm:implement-angle}
    If a type-(\romannumeral2) edge $st$ charges to an edge $xy\in E^Y_\sparse$, then angle $\angle(st, xy)$ is at most $15\sqrt{\epsilon}$. Plus, the projection $\proj_{st}(z)$ of $z$ onto line $st$ lies in the segment $st$ and satisfies $\frac{\|s - \proj_{st}(z)\|}{\|st\|}\in \left[\frac{3}{8} -\frac{1}{50} - \eps, \frac{5}{8} +\frac{1}{50} + \eps\right]$ for all $z\in \{x, y\}$
\end{claim}
The difference between \Cref{angle} and \Cref{clm:implement-angle} is the region of the projection of each edge that $st$ charges to. The region changes from $\left[\frac{3}{8} -\frac{1}{50}, \frac{5}{8} +\frac{1}{50}\right]$ to $\left[\frac{3}{8} -\frac{1}{50} - \eps, \frac{5}{8} +\frac{1}{50} + \eps\right]$. This is due to the design of our algorithm.

For sufficiently small $\eps$, \Cref{clm:fully-charge} still holds for our algorithm. We now show that \Cref{sparsity-phase1} also holds.

\begin{claim}
    During the first pruning phase, the number of new edges added to $E_1$ is at most $O(|E_\sparse|\log\alpha)$. After the first pruning phase, the number of type-(\romannumeral1) edges in $E_1$ is at most $O(|E_\sparse|)$.
\end{claim}

\begin{proof}
    We show by induction that during the first pruning phase, by the beginning of the $i$-th sub-iteration, $|E_1 \cap E^{(\romannumeral1)}|$ is at most $|E_\sparse|\alpha/2^{i - 1}$. Suppose otherwise for the sake of contradiction. Using the same argument as in the proof of \Cref{sparsity-phase1}, we obtain that there exists a set $F$ of more than $\alpha/2^{i - 1}$ edges charged to the same edge $e = xy$ in $E_\sparse$ and there is an index $j$ such that at least $\frac{\alpha}{2^i\kappa}$ edges in $L_j$ have been charged to $e$. Let $x'y'$ be the approximate edge of $xy$. We show that $P'_{x', y'} \supseteq F \cap L_j$. For any edge $st \in F \cap L_j$, by the charging scheme, we get
    \begin{equation*}
        \|sx\| + \|xy\| + \|yt\| \leq (1 + \eps)\|st\|.
    \end{equation*}
    By \Cref{clm:approx-edge}, we have $\|xx'\|, \|yy'\| \leq \eps \|xy\|$. Then, using the triangle inequality, we have
    \begin{align*}
            \|sx'\| + \|x'y'\| + \|y't\| &\leq \|sx\| + \|xy\| + \|yt\| + 2\|xx'\| + 2\|yy'\| \\ 
            &\leq (1 + \eps)\|st\| + 4\eps \|xy\| \leq (1 + 5\eps)\|st\| .
    \end{align*}
    Thus, the edge $x'y'$ must be added to our spanner and hence, all edges in $F \cap L_j$ are removed, which is a contradiction. 
\end{proof}

\Cref{helper} also holds in our construction. We then prove an analogue of \Cref{one-edge-per-level}.
\begin{lemma}
    \label{lm:impl-one-edge-per-level}
    Fix any edge $e\in E^Y_\sparse$ and level index $j\geq 0$. Then, after the second pruning phase, there is at most one type-(\romannumeral2) edge in $E_2\cap L_j$ that is charged to $e$.
\end{lemma}

\begin{proof}
    Assume that two type-(\romannumeral2) edges $s_1t_1, s_2t_2 \in E_2 \cap L_j$ are charged to the same edge $e \in E^Y_\sparse$; and w.l.o.g.\ $\|s_1t_1\| \leq \|s_2t_2\|$. We prove that $s_2t_2$ could not have been added to $E_2$ since there is already a good approximation path from $s_2$ to $t_2$ at the time we consider $s_2t_2$. Let $r$ be an endpoint of $e$, and let $c, w$ be the projections of $a, r$ on $s_1t_1$, and let $f, g, h$ be the projections of $a, r, b$ on $s_2t_2$ and $p, q$ be the projections of $c, w$ on $s_2t_2$. Let $D$ be the length of $s_1t_1$, by the design of our algorithm, we still have $\|cd\| \leq (0.29 + 2\eps)D$ and $\angle(s_1t_1, s_2t_2) \leq 30\sqrt{\eps}$. Since $s_2 t_2$ also charges to $e$, we have $\|s_2g\| /\|s_2t_2\| \in \left[\frac{3}{8} - \frac{1}{50} - \eps, \frac{5}{8} + \frac{1}{50} + \eps\right]$. 
    Using the same argument as in the proof of \Cref{one-edge-per-level}, we have that $\|s_2a\| \leq \|s_2f\| + 2170\eps D$, $\|bt_2\| \leq \|ht_2\| + 2170\eps D$ and $\|ab\| \leq \|fh\| + 120\eps D$. Then, we also have:
    \begin{equation*}
        \dist_{H_2}(s_2, a) + ||ab|| + \dist_{H_2}(b, t_2) < (1 + \kappa^2\delta)\|s_2t_2\|.
    \end{equation*}
    Therefore, $(s_2t_2)$ cannot be added to $E_2$ due to the design of our algorithm.
\end{proof} 

\Cref{cor:total-type-ii} is then carries over. We conclude that the total number of edges of our spanner after a single iteration is $O(|E_\sparse|\log{\alpha})$. Therefore, after $k$ iterations, the output spanner has $O(\log^{(k)}(\eps^{-O(d)})|E_{\sparse}|) = O((\log^{(k)}(1/\eps) + \log^{(k - 1)}(d))|E_\sparse|)$ edges.

\subsubsection{Lightness}
For the lightness, we keep the same charging schemes, $\Psi_1$ and $\Psi_2$, only changing regions $A_{s, t}$ and $B_{s, t}$ to $A'_{s, t}$ and $B'_{s, t}$, respectively. Using the same arguments as in the proof of sparsity, \Cref{const-charge}, \Cref{clm:E1-log-size} and \Cref{one-edge-per-level-lightness} still hold. It remains to prove that each edge $e\in E_\light^Y$ 
receives $O(\|e\|)$ charges under the charging scheme $\Psi_2$. For any single edge $st \in E$, we also let $F \subseteq E_\light^Y$ be the set of edges that $st$ charged to under the charging scheme $\Psi_1$.

\begin{claim}
    $\|F\| > \frac{1}{10}\|st\|.$
\end{claim}

\begin{proof}
    The projection of $F$ onto line $st$ is $\proj_{st}(F)$, whose length is at least $\frac{1}{2}\left(\frac{1}{4} - \frac{1}{25} - 2\eps\right)\|st\| > \frac{1}{10}\|st\|$.
\end{proof}

Let $S$ be the subset of edges that are already charged by $\lceil\frac{\alpha}{2^i}\rceil$ times by edges before $st$. If $\|S\| \leq 0.5 \|F\|$, we re-distribute the charge similar to the original proof. We now consider the case when $\|S\| > 0.5\|F\|$. We now prove \Cref{lightness} for our implementation. Using the same proof as \Cref{clm:mid-charge-edge}, for each edge $e$ in $S$, there exists an edge $\chi(e)$ charging to $e$ such that $\frac{1}{\kappa}\|st\| \geq \|\chi(e)\| \geq \kappa\|e\|$. We define the set $P = \{\chi(e) : e \in S\}$ and construct the set $Q = \{a_1b_1, a_2b_2, \ldots, a_\ell b_\ell\}$ similar to the proof of \Cref{lightness}. We obtain $\|Q\| \geq \frac{1}{320}\|st\|$ by using the same argument as in the proof of \Cref{shortcuts}. 

For each edge $a_ib_i$, let $e_i = z_iw_i$ be the corresponding edge in $S$ that $a_ib_i$ charged to and $\chi(e_i) = a_ib_i$. Let $p_i$ ($q_i$) and $c_i$ ($d_i$) be the intersection of the hyperplane perpendicular to $st$ through $z_i$ ($w_i$) with $st$ and $a_ib_i$. We then prove \Cref{projection1}.
\begin{claim}
    \label{clm:implement-projection1}
    Both points $c_i$ and $d_i$ lie on the segment $a_ib_i$; furthermore, $\min\{\|a_ic_i\|, \|d_ib_i\|\}\geq 0.25\|a_ib_i\|$.
\end{claim}

\begin{proof}
    Let $c'_i$ be the projection of $z_i$ on $a_ib_i$. Since $\angle(a_ib_i, st) \leq 2\sqrt{10\eps}$, we also obtain $\|c_ic'_i\| < 10\eps\|a_ib_i\|$ as in the proof of \Cref{projection1}. Since $\|a_ic_i'\|/\|a_ib_i\| \in [0.355 - \eps, 0.655 + \eps]$, we have
    \begin{equation*}
        \|a_ic_i\| \geq \|a_ic'_i\| - \|c_ic'_i\| > (0.355 -\eps)\|a_ib_i\| - 10\eps\|a_ib_i\| > 0.25\|a_ib_i\|.
    \end{equation*}
    Similarly, $\|d_ib_i\| \geq 0.25\|a_ib_i\|$.
\end{proof}

We then prove an analogue of \Cref{projection2}.

\begin{claim}
    \label{clm:implement-projection2}
	Let $u_i = \proj_{st}(a_i), v_i = \proj_{st}(b_i)$. Then, both $u_i$ and $v_i$ lie on segment $st$. Furthermore, $\|su_i\|, \|v_it\| > \frac{1}{3}\|st\|$, and $\|a_iu_i\|, \|b_iv_i\|\leq 5\sqrt{\epsilon}\cdot\|st\|$.
\end{claim}

\begin{proof}
    We have $\|sp_i\| \geq \left(\frac{3}{8} - \frac{1}{50} - \eps\right)\|st\|$. Therefore,
    \begin{equation*}
        \|su_i\| \geq \|sp_i\| - \|a_ic_i\| > \left(\frac{3}{8} - \frac{1}{50} - \eps\right)\|st\| - \frac{1}{\kappa}\|st\| > \frac{1}{3}\|st\|.
    \end{equation*}
    Similarly, $\|v_ib\| > \frac{1}{3}\|st\|$, implying that $u_i$ and $v_i$ lie on segment $st$. Using the same argument as in the proof of \Cref{projection2}, we obtain $\|a_iu_i\| \leq 5\sqrt{\eps} \cdot \|st\|$.
\end{proof}

The proof of \Cref{stitch0}, \Cref{stitch1} and \Cref{stitch2} carry over. Using the proof of \Cref{lightness}, we obtain $\dist_{H_2}(s, t) \leq (1 + \kappa^2\delta)$, implying that $st$ is not added to $E_2$, a contradiction.

Therefore, the total weight of the spanner after a single iteration is $O(\|E_\light\|\log{\alpha})$. After $k$ iterations, the total weight is $O(\log^{(k)}(\eps^{-O(d)}))\cdot \|E_{\light}\| = O(\log^{(k)}(1/\eps) + \log^{(k - 1)}(d))\cdot \|E_\light\|$.

\subsubsection{Running Time}

We first show that for each edge $st$, we can determine whether $st$ is type-(\romannumeral1) or type-(\romannumeral2) in constant time.

\begin{claim}
    \label{clm:type-deter}
    For every edge $st$, determining whether $st$ is type-(\romannumeral1) or type-(\romannumeral2) can be done in $\eps^{-O(d)}$ time. 
\end{claim}

\begin{proof}
    Let $h$ be the approximate level of $st$. Consider the net $N_h$ at level $h$. Recall that $A'_{s, t} = \bigcup_{w \in N(A_{s, t})}\ball(w, r_h)$, where $N(A_{s, t})$ is the set of net points $w \in N_h$ such that $\ball(w, r_h) \cap A_{s, t} \neq \emptyset$. Then, to determine whether an edge is type-(\romannumeral1) or type-(\romannumeral2), we find all the net point $w \in N_h$ such that $\ball(w, r_h) \cap A_{s, t} \neq \emptyset$ or $\ball(w, r_h) \cap B_{s, t} \neq \emptyset$. To find them, we first find the ancestor $(s', h)$ of $(s, 0)$. Then, observe that the ellipsoid $\Gamma_{s, t}$ is in $\ball(s', 2\|st\|)$. By the packing bound, there are $\eps^{-O(d)}$ net points in $\ball(s', 2\|st\|)$. For each net point $w$, checking whether $\ball(w, r_h) \cap A_{s, t} \neq \emptyset$ or $\ball(w, r_h) \cap B_{s, t} \neq \emptyset$ takes constant time. Thus, the total running time of determining the type of an edge is $\eps^{-O(d)}$.
\end{proof}

The net-tree spanner $G$ can be constructed in $O_{\eps, d}(n\log{n})$ time \cite{HPM06}. The first pruning phase can be implemented in $\eps^{-O(d)}n$ time as follows: Start with a spanner that has sparsity $\eps^{-O(d)}$ and lightness $\eps^{-O(d)}$, for each edge $xy$ in $G$ with $\|xy\| \geq \beta^i/25$, we find all type-(\romannumeral1) edges $st$ in $L_j$ such that $\|sx\| + \|xy\| + \|yt\| \leq (1 + 5\eps)\|st\|$. We compute $|P'_{x, y}|$ as follows: For each type-(\romannumeral1) edge $st$, we find all edge $x'y'$ in $G$ such that $\|x'y'\| \geq \beta^j/25$ and $st \in P'_{x', y'}$.
Let $\gamma = \log{\beta^j/25}$, and observe that any cross edge $x'y' \in E(G)$ satisfying $st \in P'_{x', y'}$ can only belong to $O(\log{\eps^{-1}})$ levels of the net-tree, between levels $\gamma - \log{\eps^{-1}}$ and $\gamma + \log{\eps^{-1}}$. Furthermore, any of those cross edges must have two endpoints in the ellipsoid with foci $s$ and $t$ and focal distance $(1 + 5\eps)\|st\|$. By the packing bound, for each level from $\gamma - \log{\eps^{-1}}$ to $\gamma + \log{\eps^{-1}}$, there are $\eps^{-O(d)}$ net points inside such an ellipsoid. Thus, the total number of edges we need to check is at most $\eps^{-O(d)}$. The overall running time of the first pruning phases over the entire algorithm is $\eps^{-O(d)}n \cdot \log^*{(d/\eps)}$.

Then, we prove that the second pruning phase can be implemented in $O_{\epsilon, d}(n\log^2{n})$ time. Assume that we are at the $i$-th checking iteration.
All intra-cluster and inter-cluster edges in $F'_i$ are computed by running Dijkstra's algorithm. For each cluster center $v$, we run Dijkstra's algorithm to find all shortest paths from $v$ with length at most $2^i$. Let $H'_{2, <i}$ be the graph $H_{2, <i}$ after contracting all edges of length at most $\frac{2^i\eps^2}{n}$. For any two vertices $u$ and $v$ for which $uv \in E(F'_i)$, let $\weight(u, v)$ be the weight of $uv$ in $F'_i$.
Then, we bound the time required to compute those single source shortest path trees. We show that each point belongs to at most $\eps^{-O(d)}$ trees of maximum distance $O(2^i)$.

\begin{lemma}
    \label{lm:point-constant-clusters}
    Each point in $F'_i$ belongs to at most $\eps^{-O(d)}$ single source shortest path trees.
\end{lemma}

\begin{proof}
    Recall that $E_{2, < i}$ is the set of edges with weight within $[1, 2^i)$, $E'_{2, < i}$ is the set of edges with weight within $[\frac{\eps^22^i}{n}, 2^i)$ and $H_{2, < i} = (X, E_{2, < i})$. For any vertex $u \in V(H'_{2, <i})$, let $M(u)$ be the set of points in $X$ that are contracted to $u$. For any point $u$, let $N(u)$ be the set of cluster centers $s$ such that $\dist_{H'_{2, <i}}(s, u) \leq 2^i$. We prove that $|N(u)| = \eps^{-O(d)}$ for every $u \in V(H'_{2, <i})$.
    
    Let $s$ be an arbitrary vertex in  $N(u)$. We prove that for every $u' \in M(u)$ and $s' \in M(s)$, $\dist_{H'_{2, <i}}{(s', u')} \leq 2^{i + 1}$. Observe that there exist some $u_0 \in M(u)$ and $s_0 \in M(s)$ such that $\dist_{H_{2, < i}}(u_0, v_0) \leq 2^i$. Furthermore, since the path from $u'$ to $u_0$ has at most $n$ contracted edges, $\dist_{H_{2, <i}}(u', u_0) \leq n \cdot \frac{\eps^22^i}{n} = \eps^22^i$. Similarly, $\dist_{H_{2, <i}}(s', s_0) \leq \eps^22^i$. Thus, $\dist_{H_{2, <i}}(u', s') \leq 2^i + 2\eps^22^i \leq 2^{i + 1}$. Since $\dist_{H_{2, <i}}(u', s') \geq \|u's'\|$, we have $\|u's'\| \leq 2^{i + 1}$.

    We continue with bounding the distance between two points in $N(u)$. 
    \begin{claim}
        \label{clm:bound-cluster-dist}
        Let $s_1$ and $s_2$ be two arbitrary vertices in $N(u)$. For any $s'_1 \in M(s_1)$ and $s'_2 \in M(s_2)$, $\|s_1's_2'\| \geq \frac{\eps 2^i}{1 + (\kappa + 1)^2\delta}$.
    \end{claim}

    \begin{proof}
        From our construction of clusters, we have $\dist_{H_{2, <i}}(s_1', s_2') \geq \dist_{H'_{2, <i}}(s_1, s_2) \geq \eps 2^i$. Suppose, to the contrary, that $\|s_1's_2'\| < \frac{\eps 2^i}{1 + (\kappa + 1)^2\delta}$. By the stretch argument, we have $\dist_{H_2}(s_1', s_2') \leq \frac{\eps 2^i}{1 + (\kappa + 1)^2\delta} \cdot (1 + (\kappa + 1)^2\delta) = \eps 2^i$. Hence, the shortest path from $s_1'$ to $s_2'$ contains only edges of length at most $\eps2^i$, implying that $\dist_{H_2}(s_1', s_2') = \dist_{H_{2, <i}}(s_1', s_2')$. Since the distance in the contract graph $H'_{2, i}$ is dominated by the distance in $H_{2, < i}$, we have $\dist_{H'_{2, <i}}{(s_1, s_2)} \leq \dist_{H_{2, <i}}(s_1', s_2') \leq \eps 2^i$. However, by our construction, $\dist_{H'_{2, <i}}{(s_1, s_2)} > \eps 2^i$, a contradiction.
    \end{proof}
    
    Let $N(u) = \{s_1, s_2, \ldots ,s_l\}$ and $s'_1, s'_2, \ldots ,s'_l$ be arbitrary vertices in $M(s_1), M(s_2), \ldots ,M(s_l)$, respectively. By \Cref{clm:bound-cluster-dist}, $\|s'_{h_1}s'_{h_2}\| \geq \frac{\eps 2^i}{1 + (\kappa + 1)^2\delta}$ for any $h_1 \neq h_2$. On the other hand, $\|us_h\| \leq 2^{i + 1}$ for any $h \in [1, l]$. Therefore, by the packing bound, we have 
    $$|N(u)| = l \leq \left(\frac{2^{i + 1}}{{\eps 2^i}/{(1 + (\kappa + 1)^2\delta)}}\right)^{-O(d)} = \eps^{-O(d)}.$$
This completes the proof of \Cref{lm:point-constant-clusters}.
\end{proof}
\Cref{lm:point-constant-clusters} implies that each edge is considered at most $\eps^{-O(d)}$ times in the construction of all single source shortest path trees. Since the running time of Dijkstra's algorithm for a connected graph of $m$ edges is bounded by $O(m\log{m})$, the total construction time of all graphs $F'_i$ is $\eps^{-O(d)}n\log(n/\eps^2)\log{n} = \eps^{-O(d)}n\log^2{n}$.

For each query $st$, we prove that the shortest path from $s$ to $t$ in $F'_i$ contains a constant number of edges. Here, we abuse the notation of $s$ and $t$ for vertices in $F'_i$.

\begin{claim}
    \label{clm:bounded-hop-diam}
    For any two points $s$ and $t$ in $X$ such that $2^i \leq \|st\| < 2^{i + 1}$, if $\dist_{F'_i}(s, t) \leq (1 + \eps)(1 + \kappa^2\delta)\|st\|$, then there exists a shortest path $P'_C$ from $s$ to $t$ in $F'_i$ comprising a constant number of edges. Furthermore, only the first and the last edge in $P'_C$ are intra-cluster.
\end{claim}

\begin{proof}
    Let $P'_C = \langle s = v_1, v_2, \ldots v_l = s \rangle$ be a shortest path from $s$ to $t$ in $F'_i$ with the least number of intra-cluster edges. Observe that $P'_C$ contains at most $2$ intra-cluster edges. Otherwise, if there is an intra-cluster edge $v_hv_{h + 1}$ such that $1 < h < l$, then $v_h$ or $v_{h + 1}$ is not a cluster center. Assume that $v_{h + 1}$ is not a cluster center. Hence, $v_{h + 1}v_{h + 2}$ is also an intra-cluster edge and $v_{h + 2}$ is a cluster center. However, by our construction, there is an edge from $v_h$ to $v_{h + 2}$ with weight at most $\dist_{H_{2, <i}}(v_h, v_{h + 2}) \leq \dist_{H_{2, <i}}(v_h, v_{h + 1}) + \dist_{H_{2, <i}}(v_h, v_{h + 1}) = \weight(v_h, v_{h + 1}) + \weight(v_{h + 1}, v_{h + 2})$. Replacing $v_hv_{h + 1}$ and $v_{h+1}v_{h + 2}$ by $v_{h}v_{h + 2}$ in $P'_C$, we obtain another shortest path from $s$ to $t$ with fewer intra-cluster edges, a contradiction.
    
    Furthermore, the intra-cluster edges can only be $v_1v_2$ and/or $v_{l - 1}v_l$. Consider the sub-path $Q = \langle v_2, v_3, \ldots v_{l - 1}\rangle$. We have $\weight(v_h, v_{h + 1}) + \weight(v_{h + 1}, v_{h + 2}) > 2^i$ for any $1 \leq h \leq l - 3$, since otherwise there is an edge from $v_h$ to $v_{h + 2}$ by our construction. Then, we have:
    \begin{equation*}
        \|Q\| \geq \sum_{h = 1}^{\lfloor l / 2 \rfloor - 1} \big(\weight(v_{2h},v_{2h + 1}) + \weight(v_{2h + 1}v_{2h + 2})\big) \geq \left(\lfloor l / 2 \rfloor - 1\right) \cdot 2^i.
    \end{equation*}
    On the other hand, since $\|P'_C\| \leq (1 + \eps)(1 + \kappa^2\delta)\|st\| \leq 8 \cdot 2^i$, we get that $(\lfloor l / 2 \rfloor - 1) \cdot 2^i \leq \|Q\| \leq 8 \cdot 2^i$. Thus, $l \leq 19$. 
\end{proof}

For each vertex $u$, let $\mathrm{Clusters}(u)$ be the set of clusters containing $u$. From \Cref{clm:bounded-hop-diam}, for any edge $st \in E^{(\romannumeral1)}_{i}$, among all shortest paths from $s$ to $t$ in $F'_i$, there exists a path $P'_C$ such that all edges that are incident to neither $s$ nor $t$ in $P'_C$ are inter-cluster.
For each pair $(s, t)$, we only need to check the shortest paths from the centers of all clusters in $\mathrm{Clusters}(s)$ to the center of all clusters in $\mathrm{Clusters}(t)$ containing a constant number of inter-cluster edges. By the packing bound, the  total number of inter-cluster edges incident to a vertex is  bounded by $\eps^{-O(d)}$. Thus, the total time complexity of checking each edge $s t$ is also $\eps^{-O(d)}$.

To find a helper edge, we find two net points: one in $N(A_{s, t})$ and one in $N(B_{s, t})$. (Recall that $N(A_{s, t})$ contains all net points $w \in N_h$ such that $\ball(w, \eps r_h) \cap A_{s, t} \neq \emptyset$ with $h$ being the approximate level of $st$.) Then, we add the edge between those net-points to $E_2$. 
Since the number of net points in $N(A_{s, t})$ (resp., $N(B_{s, t})$) is $\eps^{-O(d)}$, one can find a helper edge in $\eps^{-O(d)}$ time per edge.

Therefore, the total time complexity of a single run of the second pruning phase is $\eps^{-O(d)} n\log^2{n}$. This bound absorbs the time complexity of a single run of the first pruning phases. 
After summation over all $\log^*{(d/\eps)}$ iterations,
the construction time of our spanner is $\log^*{(d/\eps)}\cdot \eps^{-O(d)} n\log^2{n} =\eps^{-O(d)} n\log^2{n}$. 
\section{Lower Bounds for the Greedy Spanner}\label{sec:LB-greedy}

In this section, we construct point sets in the plane for which the lightness and sparsity of the greedy spanner far exceed the instance-optimal lightness and sparsity, respectively. In \Cref{ssec:LB-sparsity}, we first construct hard examples against the greedy $(1+\eps)$-spanner and $(1+1.2\,\eps)$-spanner (\Cref{thm:sparsityLB}), and then generalize the ideas to work against the greedy $(1+x\eps)$-spanner for all $x$, $1\leq x\leq o(\eps^{-1/3})$, by refining both the design of the point set and the of the greedy algorithm (\Cref{thm:sparsityLB+}). Since both the greedy and the sparsest spanner use edges of comparable weight, this construction already establishes the same lower bound for lightness (\Cref{cor:weightLB+}). In \Cref{ssec:LB-lightness}, we lower bounds for lightness with a stronger dependence on $\eps$, and a weaker dependence on $x$: The points sets in these constructions are uniformly distributed along a suitable circular arc (\Cref{thm:weightLB,thm:weightLB+}).

\subsection{Sparsity Lower Bounds}
\label{ssec:LB-sparsity}

We begin with lower bound construction in Euclidean plane against the greedy $(1+\eps)$-spanner. The same construction also works well against the greedy $(1+x\eps)$-spanner when $1\leq x\leq 1.2$. 
\begin{theorem}\label{thm:sparsityLB}
For every sufficiently small $\eps>0$, there exists a finite set $S\subset \mathbb{R}^2$ such that 
\[
|E_{\rm gr}| = |E_{{\rm gr}(1.2)}| \geq \Omega(\eps^{-1/2})\cdot |E_{\sparse}|,
\]
where $E_{\rm gr}$ is the edge set of the greedy $(1+\eps)$-spanner, $E_{{\rm gr}(1.2)}$ is the edge set of the greedy $(1+1.2\, \eps)$-spanner, and $E_{\sparse}$ is the edge set of a sparsest $(1+\eps)$-spanner for $S$.
\end{theorem}
\begin{proof} Let $\eps>0$ be given. We construct a point set $S$ as follows; refer to Fig.~\ref{fig:LB}. All points are in an axis-aligned rectangle $R$ of width 1 and height $\tan \alpha$, where $\alpha$ is determined by the equation 
	\begin{equation}\label{eq:LB}
		\frac{1}{\cos \alpha} = 1+\eps.
	\end{equation}
This means, in particular, that the diagonal of $R$ is exactly $1+\eps$.Using the Taylor estimate $\frac{1}{\cos x}= 1+\frac{x^2}{2}+O(x^4)$, we obtain  $\alpha=\sqrt{2\eps}+O(\eps)$. Using the Taylor estimate $\tan x = x+O(x^3)$, this implies that the height of $R$ is $\tan\alpha =\Theta(\alpha)=\Theta(\sqrt{\eps})$.
	
	\begin{figure}[htbp]
		\begin{center}
			\includegraphics[width=0.9\textwidth]{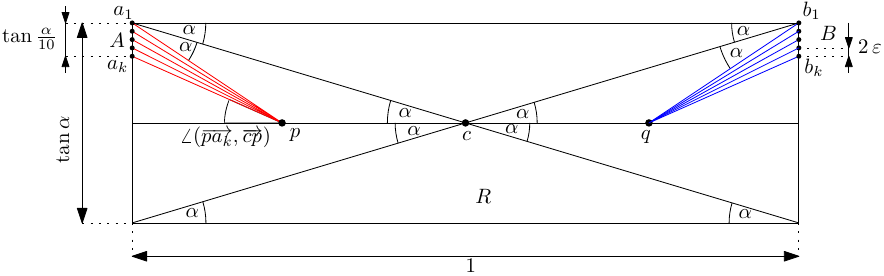}
		\end{center}
		\caption{Lower bound construction for the sparsity of the greedy algorithm. Greedy adds all red and blue edges, but it does not add any edges between $c$ and the points in the left and right sides of $R$. However, it adds all edges between $A$ and $B$.}
		\label{fig:LB}
	\end{figure}

Now we can describe the point set $S$. Along the left and right edges of $R$, resp., place points $A=\{a_1,\ldots , a_k\}$ and $B=\{b_1,\ldots , b_k\}$ such that $a_1$ and $b_1$ are the upper-left and upper-right corners of $R$, the distance between any two points is at least $2\eps$, and $\diam(A)=\diam(B)=\tan(\alpha/10)$. The set $S$ comprises $A\cup B$ and three additional points: the center $c$ of $R$, and points $p$ and $q$ in the interior of $R$ such that $p,c,q$ lie on a horizontal line and $\angle ca_1p=\alpha$ and $\angle cb_1q=\alpha$. This completes the description of $S$.  Note that $|S|=2k+3$. We have $k=\Theta(\frac{\tan(\alpha/10)}{\eps})=\Theta(\alpha/\eps)=\Theta(\eps^{-1/2})$, consequently
	$|S|=\Theta(\eps^{-1/2})$.

	\paragraph{Optimal sparsity.}
	We construct a $(1+\eps)$-spanner $H$ for $S$ with $\Theta(|S|)$ edges (i.e., sparsity $O(1)$). Let $H$ contain the vertical paths $\langle a_1,a_2,\ldots ,a_k\rangle$ and $\langle b_1,b_2,\ldots , b_k\rangle$, the horizontal path $pc\circ cq$, and all edges between $\{p ,c,q\}$ and $A\cup B$. Clearly, $H$ has $\Theta(|S|)$ edges. To show that $H$ is a $(1+\eps)$-spanner, consider a pair of points that are not adjacent in $H$. If both points are in $A$ (or both are in $B$), then they are on a vertical path in $H$. The point pair $\{p,q\}$ is connected by the horizontal path $pc\circ cq$. It remains to consider the pairs $\{a_i,b_j\}$ for $a_i\in A$ and $b_j\in B$. We show that the path $a_ic\circ cb_j$ has weight at most $(1+\eps)\|a_ib_j\|$. On one hand, $A$ and $B$ lie on two parallel lines at distance $1$ apart, hence $\|a_i b_j\|\geq 1$. On the other hand, $\|a_ic\|\leq \|a_1 c\|=\frac{1}{2\, \cos\alpha} = \frac{1+\eps}{2}$ and similarly, $\|cb_j\|\leq \|c b_1\|=\frac{1}{2\, \cos\alpha}=\frac{1+\eps}{2}$. Consequently, $\|a_ic\|+\|b_jc\|\leq \|a_1c\|+\|b_1c\|\leq 1+\eps \leq (1+\eps)\|a_ib_j\|$, as required. 
	
   \paragraph{Greedy sparsity.}
   Now let us consider the greedy algorithm on the point set $S$.
   The same analysis works for the greedy $(1+\eps)$-spanner and the greedy $(1+1.2\, \eps)$-spanner (for short, the greedy spanner). 
    First observe that the greedy spanner contains all edges of the vertical paths $\langle a_1,a_2,\ldots ,a_k\rangle$ and $\langle b_1,b_2,\ldots , b_k\rangle$, and the horizontal path $pc\circ cq$. It also contains all edges between $p$ and $A$, and all edges between $q$ and $B$ (since the distance between any two points in $A$ and any two points in $B$ is at least $2\eps$). However, the greedy algorithm does not add any of the edges between $c$ and $A\cup B$, because for every $a_i\in A$, the path $a_ip\circ pc$ has weight less than $(1+\eps)\|a_i c\|$. Indeed, both segments $a_i p$ and $cp$ make an angle at most $\alpha$ with $a_ic$. Combined with \eqref{eq:LB}, this already implies $\|a_ip\|+\|pc\|\leq (1+\eps)\|a_i c\|$. Similarly for every $b_j\in B$, the path $b_jq\circ qc$ has weight at most $(1+\eps)\|b_j c\|$.
	
    Finally, we show that the greedy algorithm will add all $|A|\cdot |B|=\Omega(\eps^{-1})$ edges between $A$ and $B$. On the one hand, for any $a_i\in A$ and any $b_j\in B$, we have 
	\begin{equation}\label{eq:below}
    \|a_ib_j\|
    \leq \frac{1}{\cos(\alpha/10)}
    =1+\frac{\alpha^2}{200}+O(\alpha^4)
    =1+\frac{\eps}{100}+O(\eps^2).
	\end{equation}
 On the other hand, when the greedy algorithm considers adding edge $a_ib_j$, the shortest $a_ib_j$-path in the current greedy spanner is $a_ip\circ pc\circ cq\circ qb_j$. For $a_i=a_1$ and $b=b_1$, the length of this path is $(1+\eps)^2=1+2\eps+\eps^2>(1+\eps)\|a_1b_1\|$. In general, for $i,j\in \{1,\ldots , k\}$, the length of this $a_ib_j$-path is minimized for $a_i=a_k$ and $b_j=b_k$, that is,
	\begin{equation}\label{eq:path}
		\|a_ip\|+\|pq\|+\|q b_j\|\geq \|a_k p\|+\|pq\|+\|q b_k\|.
	\end{equation}
 
	By construction, $\Delta a_1b_1c$, $\Delta a_1cp$, and $\Delta b_1cq$ are isosceles triangles, and so $\|a_1p\|=\|pc\|=\|cq\|=\|b_1q\|$. Due to the choice of $\alpha$, we have $\|a_1c\|+\|b_1c\|=1+\eps$ and $\|a_1p\|+\|pc\|+\|cq\|+\|b_1q\|=(1+\eps)^2=1+2\eps +\eps^2$. This implies $\|a_1p\|=\|pc\|=\|cq\|=\|b_1q\|=\frac14+\frac{\eps}{2}+O(\eps^2)$.
	
	Note that $a_1p$ and $a_k p$ have the same orthogonal projection to the $x$-axis,	and so 
	\begin{equation}\label{eq:angle}
		\|a_1p\|\cos (2\alpha) = \|a_k p\| \cos\angle (\overrightarrow{pa_k},\overrightarrow{cp}).
    \end{equation}
	Since $\|a_1a_k\|\leq \tan\frac{\alpha}{10} < \frac18\cdot \tan \alpha$ for all sufficiently small $\eps>0$, then $\angle( \overrightarrow{pa_k},\overrightarrow{cp}) <\frac{3\alpha}{2}$. This, combined with \eqref{eq:angle}, the Taylor estimate $\cos x= 1-\frac{x^2}{2}+O(x^4)$, and $\alpha=\sqrt{2\eps}+O(\eps)$, yields
	\begin{align*}
		\|a_k p\| &> \|a_1p\| \frac{\cos (2\alpha)}{\cos (3\alpha/2)}\\
		&=\|a_1p\| \left (1-2\alpha^2+O(\alpha^4)\right) \left(1+\frac{9\alpha^2}{8}+O(\alpha^4)\right)\\
		&=\left(\frac14+\frac{\eps}{2}+O(\eps^2)\right)
		\left (1-4\eps+O(\eps^2)\right) \left(1+\frac{9\eps}{4}\eps+O(\eps^2)\right)\\
		&=\frac14+\frac{\eps}{16}+O(\eps^2).
	\end{align*}
	Overall, using \eqref{eq:below} and \eqref{eq:path}, we obtain 
	\begin{align*}
		\|a_ip\|+\|pq\|+\|q b_j\| 
		&\geq \|a_k p\|+\|pq\|+\|q b_k\|\\
		&=2\|a_k p\| +2 \|pc\|\\
		&> 2\left(\frac14+\frac{\eps}{16}+O(\eps^2)\right) +
		2\left (\frac14+\frac{\eps}{2}+O(\eps^2)\right)\\
		&= 1+\frac{5\eps}{4}+O(\eps^2)\\
		& > (1+1.2\,\eps)\left(1+\eps/100+O(\eps^2)\right)\\
		&\geq (1+1.2\,\eps)\|a_ib_j\|
	\end{align*}
    for sufficiently small $\eps>0$. This shows that the greedy algorithm must add edge $a_ib_j$ for all $a_i\in A$ and $b_j\in B$. Consequently, the greedy $(1+\eps)$-spanner and the greedy $(1+1.2\, \eps)$-spanner both have at least $|A'|\cdot |B'| =\Omega(\eps^{-1})$ edges. 
\end{proof}

The construction above implies a lower bound for the lightness ratio, as well. We obtain a better bound in Section~\ref{ssec:LB-lightness}.

\begin{corollary}\label{cor:weightLB}
For every sufficiently small $\eps>0$, there exists a finite set $S\subset \mathbb{R}^2$ such that 
\[
\|E_{\rm gr}\| = \|E_{{\rm gr}(1.2)}\| \geq \Omega(\eps^{-1/2})\cdot \|E_{\light}\|,
\]
where $E_{\rm gr}$ is the edge set of the greedy $(1+\eps)$-spanner, $E_{{\rm gr}(1.2)}$ is the edge set of the greedy $(1+1.2\cdot \eps)$-spanner, and $E_{\light}$ is the edge set of a minimum-weight spanner for $S$.
\end{corollary}
\begin{proof}
	For $\eps>0$, consider the point set $S$ constructed in the proof of Theorem~\ref{thm:sparsityLB}. We have shown that $S$ admits a $(1+\eps)$-spanner $H$ with $O(\eps^{-1/2})$ edges, each of weight $O(1)$. Consequently, $\|E_{\light}\|\leq \|E(H)\| =O(\eps^{-1/2})$. 
	
	We have also shown that the greedy $(1+\eps)$- and $(1+1.2\eps)$-spanner each contain $\Omega(|S|^2)=\Omega(\eps^{-1})$ edges, each of weight $\Omega(1)$. Consequently, their total weight is $\Omega(\eps^{-1})$, which is $\Omega(\eps^{-1/2}) \|E_{\light}\|$. 
\end{proof}

The lower bound construction generalizes to the case when the stretch is $(1+x\eps)$ for some $1\leq x\leq o(\eps^{-1/2})$, and we compare the sparsest $(1+\eps)$-spanner with the greedy $(1+x \eps)$-spanner.

\lbsparse*
\begin{proof}
Let $\eps>0$ and $1\leq x\leq o(\eps^{-1/3})$ be given. We construct a point set $S$ as follows; refer to Fig.~\ref{fig:LB+}. Let $R$ be an axis-aligned rectangle $R$ of width 1 and height $\tan \alpha$, where $\alpha$ is determined by the equation
	\begin{equation}\label{eq:LB+}
		\frac{1}{\cos \alpha} = 1+\eps.
	\end{equation}
	This means, in particular, that the diagonals of $R$ have length precisely $1+\eps$. Using the Taylor estimate $\frac{1}{\cos x}= 1+\frac{x^2}{2}+O(x^4)$, we obtain  $\alpha=\sqrt{2\eps}+O(\eps)$.
	Using the Taylor estimate $\tan x = x+O(x^3)$, this implies that the height of $R$ is $\tan\alpha =\sqrt{2\eps}+O(\eps)=\Theta(\sqrt{\eps})$.
	
	\begin{figure}[htbp]
		\begin{center}
			\includegraphics[width=0.9\textwidth]{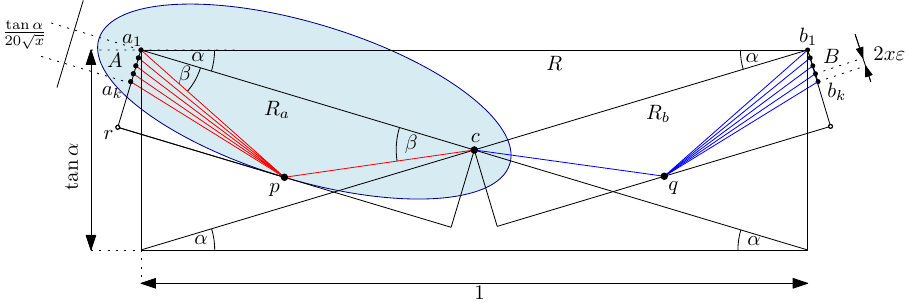}
		\end{center}
		\caption{Lower bound construction for the sparsity of the $(1+x\eps)$-greedy spanner. The locus of all points $p$ such that $\|a_1p\|+\|pc\|=(1+x\eps)\|a_1c\|$ is an ellipse with foci $a_1$ and $c$. The greedy $(1+x\eps)$-spanner contains all red and blue edges, as well as all edges between $A$ and $B$.}
		\label{fig:LB+}
	\end{figure}

Denote by $a_1$ and $b_1$ the upper-left and upper-right corners of $R$, respectively, and let $c$ be the center of $R$. Next, place points $p$ and $q$ below the segments $a_1c$ and $b_1c$, resp., such that
\[
 \|a_1p\|= \|pc\|=\|cq\|=\|qb_1\|=\frac12(1+x\eps) \|a_1c\|=\frac14(1+x\eps)(1+\eps)=\frac14\Big(1+(x+1)\eps+O(\eps^2)\Big).
\]
Note that $\Delta a_1cp$ and $\Delta b_1cq$ are isosceles triangles,
where $\beta:=\angle ca_1 p = \angle pca_1 = \angle b_1cq = \angle qb_1c$ and $\beta=\sqrt{2x\eps}+O(x\eps)$.
Let $R_a$ be the rectangle with one side $a_1c$ and the opposite side containing $p$; and similarly let $R_b$ be the rectangle with one side $b_1c$ and the opposite side containing $q$.

Now we can describe the point set $S$. Along the left side of $R_a$ and the right side of $R_b$, resp., place equally spaced points $A=\{a_1,\ldots , a_k\}$ and $B=\{b_1,\ldots , b_k\}$ such that the distance between any two points is at least $2x\eps$, and $\diam(A)=\diam(B)=\frac{1}{20\sqrt{x}}\tan\alpha$. Our point set is $S=A\cup B\cup \{c,p,q\}$.
Note that $|S|=2k+3$. We have $k=\Theta(\frac{\tan\alpha}{\sqrt{x}\cdot x\eps})=\Theta(\frac{\alpha}{x^{3/2}\eps})=\Theta(\eps^{-1/2}/x^{3/2})$, consequently $|S|=\Theta(\eps^{-1/2}/x^{3/2})$.

    \paragraph{Optimal sparsity.}
    We construct a $(1+\eps)$-spanner $H$ for $S$ with $\Theta(|S|)$ edges (i.e., sparsity $O(1)$). Let $H$ contain the paths $\langle a_1,a_2,\ldots ,a_k\rangle$ and $\langle b_1,b_2,\ldots , b_k\rangle$, the edges of the triangle $\Delta cpq$, and all edges between $\{p,c,q\}$ and $A\cup B$. Clearly, $H$ has $\Theta(|S|)$ edges.

    To show that $H$ is a $(1+\eps)$-spanner, consider a pair of points that are not adjacent in $H$. Any pair of points in $A$ (resp., $B$), are connected by a path of collinear edges.  Consider a point pair $\{a_i,b_j\}$, where $a_i\in A$ and $b_j\in B$. We show that the path $a_ic\circ cb_j$ has weight at most $(1+\eps)\|a_ib_j\|$. This clearly holds for $i=j=1$, where
    $\|a_1c\|+\|cb_1\|=(1+\eps)=(1+\eps)\|a_1b_1\|$ by the choice of $\alpha$. In particular, $c\in \mathcal{E}(a_1,b_1)$, where $\mathcal{E}(f_1,f_2)$ denotes the ellipse with foci $f_1$ and $f_2$ and great axis $(1+\eps)\|f_1f_2\|$. For all $i\in \{1,\ldots ,k\}$, the segment $a_ib_i$ is horizontal and $\|a_ib_i\|\geq \|a_1b_1\|$, therefore $\mathcal{E}(a_i,b_i)$ is obtained from $\mathcal{E}(a_1,b_1)$ by a vertical translation and scaling; hence $c\in \mathcal{E}(a_i,b_i)$. Finally, for all $j\geq i$, we have  $\|a_ib_j\| \geq \|a_ib_i\|$, and $\mathcal{E}(a_i,b_j)$ is obtained from $\mathcal{E}(a_i,b_j)$ by a rotation about $a_i$ and scaling; and so $c\in \mathcal{E}(a_i,b_j)$.
    Overall, we have $c\in \mathcal{E}(a_i,b_j)$ for all $i,j\in\{1,\ldots , k\}$, which implies $\|a_ic\|+\|b_jc\|\leq 1+\eps \leq (1+\eps)\|a_ib_j\|$, as required.
	
   \paragraph{Greedy sparsity.}
   Now let us consider the greedy $(1+x\eps)$-spanner $G_{{\rm gr}(x)} = (S, E_{{\rm gr}(x)})$ on the point set $S$. The greedy algorithm sorts the point pairs in $S$ by weight: It adds all edges of the paths $(a_1,\ldots , a_k)$ and $(b_1,\ldots , b_k)$ to $G_{{\rm gr}(x)}$. It then considers the edges between $p$ and $A$ (resp., $q$ and $B$) by increasing weight: We claim that it adds all these edges to $G_{{\rm gr}(x)}$. By symmetry, it is enough to show that if the greedy algorithm has already added $a_ip$ for some $1<i\leq k/2$, then it also adds $a_{i-1}p$. That is,
   $\|a_{i-1}a_i\| + \|a_i p\|\geq (1+x\eps)\,\|a_{i-1} p\|$.
   Since $\|a_{i-1}a_i\|\geq 2x\eps$ by construction, then $a_{i-1}$ lies outside of the ellipse with foci $a_i$ and $p$, and great axis $(1+x\eps)\|a_ip\|$. Thus the greedy spanner contains all edges between $A$ and $p$, and between $B$ and $q$.

   Next, we show that the greedy algorithm does not add any edge between $A$ and $\{c,q\}$. We claim that for any point $a_i\in A$, the paths $a_ip\circ pc$ and $a_ip\circ pc\circ cq$ have stretch at most $1+x\eps$. This clearly holds for the path $a_1p\circ pc$ by the definition of $p$. An easy calculation shows that it holds for all other point pairs in $A\times \{c,q\}$. Symmetrically, the greedy $(1+x\eps)$-spanner has no edges between $B$ and $\{p,c\}$.

   Finally, the greedy algorithm considers pairs $\{a_i,b_j\}$ for $a_i\in A$ and $b_j\in B$ sorted by weight. We show that for all such pairs, we have $\|a_ip\circ pc\circ cq\circ qb_j\|\geq (1+x\eps)\|a_i b_j\|$, and the greedy algorithm must add the edge $a_ib_j$ to \textbf{$G_{{\rm gr}(x)}$}. This is clear for the pair $\{a_1,b_1\}$, where the definition of $p$ and $q$ gives
\begin{equation}\label{eq:11}
    \|a_1p\circ pc\circ cq\circ qb_1\|
    =(1+x\eps)(1+\eps)
    = 1+(x+1)\eps+O(x\eps^2)
    >(1+x\eps)\|a_1b_1\|.
\end{equation}
   In general, consider an arbitrary point pair $(a_i,b_j)\in A\times B$. The distance between $a_i$ and $b_j$ is maximized for $a_i=a_k$ and $b_j=b_k$, where 
\[\|a_ib_j\|
   \leq \|a_k b_k\|
   = 1+ 2\cdot \frac{\tan\alpha}{20\sqrt{x}} \sin\alpha
   \leq 1+\frac{\sqrt{2\eps}}{10}\cdot \sqrt{2\eps}+O(\eps^{3/2})
   =1+\frac{\eps}{5}+O(\eps^{3/2}).
\]
We give a lower bound for $\|a_ip\circ pc\circ cq\circ qb_j\|$ using \Cref{eq:11} and the differences $\|a_1p\|-\|a_ip\|$ and $\|b_1q\|-\|b_kq\|$. By symmetry and monotonicity, it is enough to give an upper bound for $\|a_1p\|-\|a_kp\|$.
Let $r$ denote the bottom-right vertex of $R_a$; see \Cref{fig:LB+}.
The Pythagorean theorem for the right triangles $\Delta a_1pr$ and $\Delta a_k pr$, combined with the identity $\sqrt{y}-\sqrt{z}=\frac{y-z}{\sqrt{y}+\sqrt{z}}$, yields
\begin{align*}
    \|a_1p\| - \|a_kp\|
    & =\sqrt{ \|a_1r\|^2 + \|pr\|^2} - \sqrt{\|a_kr\|^2 + \|pr\|^2}\\
    & = \frac{ \|a_1r\|^2 - \|a_kr\|^2}{\sqrt{ \|a_1r\|^2 + \|pr\|^2} + \sqrt{\|a_kr\|^2 + \|pr\|^2}}\\
    &<  \frac{ (\|a_1a_k\|+\|a_kr\|)^2 - \|a_kr\|^2}{\|a_1p\|}\\
    &= \frac{\|a_1a_k\|^2 + 2\cdot \|a_1a_k\|\cdot \|a_kr\|}{\|a_1p\|}\\
    &< \frac{2\cdot \|a_1a_k\|\cdot \|a_1r\|}{\|a_1p\|}\\
    &= \frac{2\cdot\frac{\tan\alpha}{20\sqrt{x}}\cdot  \frac14(1+\eps)\tan\beta}{\frac14(1+\eps)(1+x\eps)}\\
    &= \frac{\tan\alpha\cdot \tan\beta}{10\sqrt{x}(1+x\eps)}\\
    &= \frac{1}{10\sqrt{x}} \left(\sqrt{2\eps}+O(\eps)\right)
  \left( \sqrt{2x\eps}+O(x\eps)\right)\Big(1-x\eps+O(x^2\eps^2)\Big) \\
    &\leq \frac{\eps}{5} +O(\eps^{3/2}).
\end{align*}

Overall, for all $a_i\in A$ and $b_j\in B$, we have
\begin{align*}
 \|a_ip\circ pc\circ cq\circ qb_j\|
  &\geq \|a_1p\circ pc\circ cq\circ qb_1\| - \Big(\|a_1p\|-\|a_ip\|\Big)  - \Big(\|b_1q\|-\|b_jq\|\Big) \\
  &\geq \|a_1p\circ pc\circ cq\circ qb_1\| - 2\cdot \Big(\|a_1p\|-\|a_kp\|\Big) \\
  &\geq \Big(1+(x+1)\eps+O(x\eps^2)\Big) - 2\cdot \left(\frac{\eps}{5}+O(\eps^{3/2})\right)\\
  &\geq 1+\left(x+\frac35\right)\eps+O(\eps^{3/2})\\
&> (1+x\eps) \left(1+\frac{\eps}{5}+O(\eps^{3/2})\right)\\
&\geq (1+x\eps)\|a_ib_j\|
\end{align*}
for all sufficiently small $\eps>0$ and $x=o(\eps^{-1/3})$. This shows that the greedy $(1+x\eps)$-spanner contains edges $a_ib_j$ for all $a_i\in A$ and $b_j\in B$. In particular, we have  $|E_{\sparse}|=O(|S|)$ and $|E_{{\rm gr}(x)}|=\Omega(|S|^2)$. Consequently, $|E_{{\rm gr}(x)}|\geq \Omega(|S|)\cdot |E_{\sparse}| \geq \Omega(\eps^{-1/2}/x^{3/2})\cdot |E_{\sparse}|$.

The construction above has $|S|=O(\eps^{-1/2}/x^{3/2})$ points. However, we can obtain arbitrarily large point sets with the same property as a disjoint union of translated copies of the construction above: Both $G_{\sparse}$ and $G_{{\rm gr}(x)}$ would have only one extra edge per copy, so the bound $|E_{{\rm gr}(x)}|\geq \Omega(\eps^{-1/2}/x^{3/2})\cdot |E_{\sparse}|$ carries over.
   \end{proof}

\begin{corollary}\label{cor:weightLB+}
For every sufficiently small $\eps>0$ and $1\leq x\leq o(\eps^{-1/3})$, there exists a finite set $S\subset \mathbb{R}^2$ such that 
\[
\|E_{{\rm gr}(x)}\| \geq 
\Omega\left(\frac{\eps^{-1/2}}{x^{3/2}}\right)\cdot \|E_{\light}\|,
\]
where $E_{{\rm gr}(x)}$ is the edge set of the greedy $(1+x\eps)$-spanner, and $E_{\light}$ is the edge set of a minimum-weight $(1+\eps)$-spanner for $S$.
\end{corollary}
\begin{proof}
    For $\eps>0$, consider the point set $S$ constructed in the proof of Theorem~\ref{thm:sparsityLB+}. We have shown that $S$ admits a $(1+\eps)$-spanner $H$ with $O(|S|)=O(\eps^{-1/2}/x^{3/2})$ edges, each of weight $O(1)$.
	
    We have also shown that the greedy $(1+x\eps)$-spanner has  $\Omega(|S|^2)$ edges, each of weight $\Omega(1)$. Consequently, $\|E_{{\rm gr}(x)}\| /  \|E_{\light}\|\geq \Omega(|S|) \geq \Omega(\eps^{-1/2}/x^{3/2})$. 
\end{proof}

\subsection{Lightness Lower Bounds}
\label{ssec:LB-lightness}

\begin{theorem}\label{thm:weightLB}
For every sufficiently small $\eps>0$, there exists a finite set $S\subset \mathbb{R}^2$ such that 
\[
\|E_{\rm gr}\| \geq \|E_{\rm gr}'\| \geq \Omega(\eps^{-1})\cdot \|E_{\light}\|,
\]
where $E_{\rm gr}$ is  the edge set of a greedy $(1+\eps)$-spanner, $E_{\rm gr}'$ is the edge set of the greedy $(1+1.01\, \eps)$-spanner, and $E_{\light}$ is the edge set of a minimum-weight $(1+\eps)$-spanner for $S$.
\end{theorem}
\begin{proof}
    Let $\eps>0$ be given. We construct a point set $S$ as follows; refer to Fig.~\ref{fig:weightLB}. The points in $S$ lie on a circular arc $C$, which is an arc of length $\alpha+\beta$ of a circle of radius $1$ centered at the origin $o$, for angles $\alpha<\beta$ to be specified later. For any $s,t\in C$, let $C(s,t)$ denote the subarc of $C$ between $s$ and $t$, and let $\mathrm{arc}(s,t)=\|C(s,t)\|$ denote the length of the circular arc $C(s,t)$. First, we place four points $p_1,\ldots ,p_4\in C$ such that $\mathrm{arc}(p_1,p_2)=\alpha$,  $\mathrm{arc}(p_2,p_3)=\beta-\alpha$, and $\mathrm{arc}(p_3,p_4)=\alpha$.
    Now we construct the point set $S$ as follows:
     Place a large number of equally spaced points along $C(p_1,p_2)$ and $C(p_2,p_3)$, and then populate $C(p_3,p_4)$ with the rotated copy of the points in $C(p_1,p_2)$. Note that, by construction, $S$ contains a large number of point pairs $\{s,t\}\subset S$ with $\|st\|=\|p_1p_3\|$.
	       	
	\begin{figure}[htbp]
		\begin{center}
			\includegraphics[width=0.9\textwidth]{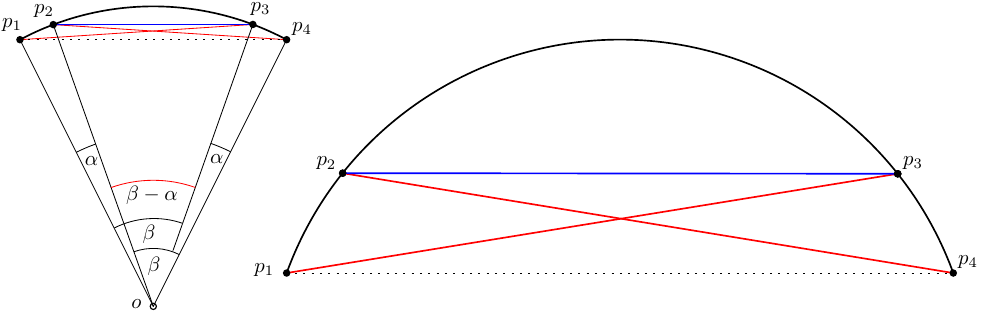}
		\end{center}
		\caption{Lower bound construction for the weight of the greedy algorithm.}
		\label{fig:weightLB}
	\end{figure}
	
    Let $P$ be the spanning path of $S$ that connects pairs of consecutive points along the circle. For $s,t\in S$, let $P(s,t)$ denote the subpath of $P$ between $s$ and $t$. If the points in $S$ are sufficiently dense along $C$, the length of $P(s,t)$ will be arbitrarily close to $\mathrm{arc}(s,t)$.  In the remainder of this proof, we assume that $\|P(s,t)\|=\|\mathrm{arc}(s,t)\|+O(\eps^2)$ for all $s,t\in S$.
    
    Using the Taylor estimates for sine and the bound $\mathrm{arc}(s,t)\leq 2\beta =O(\eps^{1/2})$, the difference between $\mathrm{arc}(s,t)$ and the length of the line segment $st$ is 
\begin{equation}\label{eq:saving}
    \|P(s,t)\|-\|st\|
    =\mathrm{arc}(s,t)-2\sin\left(\frac12 \, \mathrm{arc}(s,t)\right) +O(\eps^2) 
   =\frac{(\mathrm{arc}(s,t))^3}{48} +O(\eps^2).
\end{equation}

  To complete the construction, we specify $\alpha$ and $\beta$.
    We set $\alpha=\frac{\beta}{10}$ and choose $\beta$ such that 
    \begin{equation}\label{eq:beta}
    \|P(p_1,p_3)\|>(1+\eps)\|p_1 p_3\|,
    \end{equation}
    but $\|P(s,t)\|\leq (1+\eps)\|st\|$, for all $s,t\in C$ with $\|P(s,t)\|<\|P(p_1,p_3)\|$. Recall that $\|P(p_1,p_3)\|=\mathrm{arc}(p_1,p_3)+O(\eps^2)=\beta+O(\eps^2)$. The line segment $p_1p_3$ has length $\|p_1p_3\|=2\sin\frac{\beta}{2}$. Using the Taylor estimate $x-\frac{x^3}{6}\leq \sin x\leq x-\frac{x^3}{6}+O(x^5)$, we have $\beta(1-\frac{\beta^2}{48})\leq \|p_1p_3\|\leq \beta(1-\frac{\beta^2}{48}+O(\beta^4))$. Now $\|P(p_1,p_3)\|= (1+\eps)\|p_1 p_3\|+O(\eps^2)$ is attained for $\eps=\beta^2/48+\Theta(\beta^4)$ or $\beta=\sqrt{48\eps}+\Theta(\eps)$.

   \paragraph{Optimal weight.}
   We show that the minimum weight of an $(1+\eps)$-spanner for $S$ is $\|E_{\light}\|=O(\eps^{1/2})$. We claim that the graph $H$ comprised of the path $P$ and the edge $p_2p_3$ is a $(1+\eps)$-spanner for $S$. Consequently, $\|E_\light\|\leq \|E(H)\|=\|P\|+\|p_2p_3\|\leq (\alpha+\beta)+(\beta-\alpha)=2\beta =O(\eps^{1/2})$. 

   To prove the claim, note first that for any pair $s,t\in S$ with $\|st\|<\|p_1p_3\|$, the path $P(s,t)\subset H$ has weight at most $(1+\eps)\|st\|$ due to the choice of $\beta$ (cf., \Cref{eq:beta}). Consider now a pair $s,t\in S$ with $\|st\|\geq \|p_1p_3\|$. Then the points $s$ and $t$ lie in two distinct arcs $C(p_1,p_2)$ and $C(p_3,p_4)$. We may assume w.l.o.g.\ that $s\in C(p_1,p_2)$ and $t\in C(p_3,p_4)$. Then $H$ contains the $st$-path $P(s,p_2)\circ p_2p_3 \circ P(p_3,t)$. The weight of this $st$-path is bounded by:
   
   {\allowdisplaybreaks
   \begin{align}
   \|P(s,p_2)\circ p_2p_3\circ P(p_3,t)\| 
   &\leq  \mathrm{arc}(s,p_2)+\|p_2p_3\|+\mathrm{arc}(p_3,t) +O(\eps^2)\nonumber\\
       &\leq \Big(\mathrm{arc}(s,p_2)+\mathrm{arc}(p_2,p_3)+\mathrm{arc}(p_3,t)\Big) +\big(\|p_2p_3\|-\mathrm{arc}(p_2,p_3)\big)+O(\eps^2)\nonumber\\
       & =\mathrm{arc}(s,t) + \Big(\|p_2p_3\|-\mathrm{arc}(p_2,p_3)\Big)+O(\eps^2)\nonumber\\
       &=\|st\|+\Big(\mathrm{arc}(s,t) - \|st\|\Big) + \Big(\|p_2p_3\|-\mathrm{arc}(p_2,p_3)\Big)+O(\eps^2)\nonumber\\
       &=\|st\| + \frac{(\mathrm{arc}(s,t))^3-(\mathrm{arc}(p_2,p_3))^3}{48}+ O(\eps^{2})\nonumber\\
       &\leq \|st\| + \frac{(\mathrm{arc}(p_1,p_4))^3-(\mathrm{arc}(p_2,p_3))^3}{48}+ O(\eps^{2})\nonumber\\
      &=\|st\| + \frac{(\alpha+\beta)^2-(\beta-\alpha)^3}{48} +(\eps^{2})\nonumber\\
      &=\|st\|+\frac{(1.1^3-0.9^3)\beta^3}{48} +O(\eps^{2})\nonumber\\
      &<\|st\|+\frac{0.7\cdot (\sqrt{48\eps}+\Theta(\eps))^3}{48} +O(\eps^{2})\nonumber\\
      &=\|st\|+0.7\cdot \sqrt{48\eps}\cdot \eps + O(\eps^{2})\nonumber\\
      &=\|st\|+0.7\cdot\beta\cdot \eps + O(\eps^{2})\nonumber\\
      &<\|st\|+0.7\cdot\left(\|p_1p_2\|+\frac{\beta^3}{48}+O(\eps^{2})\right)\cdot \eps + O(\eps^{2})\nonumber\\
     & =\Big(1+0.7\cdot \eps\Big)\|st\|+ O(\eps^{2})\nonumber\\
      &<(1+\eps)\|st\|, \label{eq:shortcut}
   \end{align}
   }
   
   for a sufficiently small $\eps>0$. This confirms that $H$ is a $(1+\eps)$-spanner for $S$, as claimed. 

    \paragraph{Greedy weight.}
    We show that the greedy $(1+\eps)$-spanner and the greedy $(1+1.01\, \eps)$-spanner for $S$ both have weight $\Omega(\eps^{-1/2})$. We argue about the greedy $(1+\eps)$-spanner (greedy spanner, for short), but essentially the same argument holds for the greedy $(1+1.01\, \eps)$-spanner, as well. Let $G_{\rm gr} = (S, E_{\rm gr})$.
    The greedy algorithm adds the entire path $P=P(p_1,p_4)$ to $G_{\rm gr}$, and then considers point pairs sorted by increasing weight. For all point pairs $\{s,t\}\subset S$ with $\|st\|< \|p_1p_3\|$, we have $\|P(s,t)\|< (1+\eps)\|st\|$, and so none of these edges is added to $G_{\rm gr}$. By construction, $S$ contains a large number of point pairs $\{s,t\}$ with $\|st\|=\|p_1p_3\|$. We show that the greedy algorithm adds $\Omega(\alpha/\eps)$ such pairs to $G_{\rm gr}$.  
	
	Specifically, we claim that every circular arc of length $2\eps\beta$ of $C(p_1,p_2)$ contains an endpoint of some edge of weight $\|p_1p_3\|$ in $G_{\rm gr}$. Suppose, for the sake of contradiction, that this is not the case. Then there is a point pair $\{s,t\}$ with $\|st\|=\|p_1p_3\|$ such that $G_{\rm gr}$ does not contain any edge of this weight whose endpoints are within arc distance $\eps\,\beta$ from $s$ or $t$. This means that when the greedy algorithm considers the point pair $\{s,t\}$, all $st$-paths have length more than $(1+\eps)\|st\|$: The path $P(s,t)$ has weight more than $(1+\eps)\|st\|$ due to \eqref{eq:beta}. Furthermore, any $st$-path that goes thought a previously added edge of weight $\|p_1p_3\|$ must include subpaths of length at least $\eps\,\beta$ in the neighborhood of $s$ and $t$, resp., and so its total length is at least $2\eps\beta+\|p_1p_2\|>(1+\eps)\|st\|$. Consequently, the greedy algorithm must add $st$ to the spanner, which is a contradiction. 
    This completes the proof of the claim. 
	
    Since $\mathrm{arc}(p_1,p_2)=\alpha$, then $G_{\rm gr}$ contains $\Omega(\alpha/(\eps\beta))=\Omega(\eps^{-1})$ edges of weight $\|p_1p_3\|=\Theta(\beta)=\Theta(\eps^{1/2})$, and so the weight of the greedy spanner is $\|E_{\rm gr}\|\geq \Omega(\eps^{-1}\cdot \eps^{1/2})=\Omega(\eps^{-1/2})$.
\end{proof}

The lower bound generalizes to the case where we relax the stretch to $a+x\eps$ for $2\leq x\leq O(\eps^{-1/2})$, and compare the lightest $(1+\eps)$-spanner with the greedy $(1+x \eps)$-spanner.

\lblight*

\begin{proof}
We use the point set in the proof of \Cref{thm:weightLB}, with $x\eps$ in place of $\eps$. That is, $S$ is a set of points on a circular arc of radius 1 and angle $\alpha_x+\beta_x$, where $\alpha_x=\beta_x/10$ and $\beta_x=\sqrt{48x \eps}+O(x\eps)$. In the proof of \Cref{thm:weightLB}, we saw that $\|G_{{\rm gr}(x)}\| = \Omega((x\eps)^{-1/2})$. 

\paragraph{Optimal weight.}
We show that $S$ admits a $(1+\eps)$-spanner of weight $O(\eps^{1/2}\cdot x^2\log x)$. Specifically, we construct a graph $H$ on $S$ that comprises the path $P$ and a set of chords. For $i=-2,-1,0,1,\ldots , 2\,\lceil 10x \log \sqrt{x}\rceil$, we augment $H$ with a maximal collection of chords, each of length $\sqrt{48\eps}\cdot \left(1+\frac{1}{10x}\right)^{i/2}$, such that the arc distance between the left endpoints of any two chords is $\frac{1}{20}\cdot \sqrt{48\eps}\cdot \left(1+\frac{1}{10x}\right)^i$. 

We claim that $H$ is a $(1+\eps)$-spanner for $S$. Consider a point pair $s,t\in S$. If $\mathrm{arc}(s,t)\leq \sqrt{48\eps}$, then $\|P(s,t)\|\leq (1+\eps)\|st\|$. Otherwise, there exists a chord $ab\in H$ such that $a,b\in C(s,t)$ and 
\[
    \left(1+\frac{1}{10x}\right)^{-1}\mathrm{arc}(s,t)\leq\mathrm{arc}(a,b)\leq \left(1+\frac{1}{10x}\right)^{-1/2}\mathrm{arc}(s,t).
\]
\Cref{eq:shortcut} now shows (substituting $\beta_x=\sqrt{48 x\eps}+O(x\eps)$ instead of $\beta =\sqrt{48 \eps}+O(\eps)$) 
that $\|P(sa)\|+\|ab\|+\|P(bt)\|\leq (1+\eps)\|st\|$. 

It remains to bound the weight of $H$. The weight of the path $P$ is $\|P\|=O(\beta)=O(\sqrt{x\eps})$. For every $i\in \{-2,-1,\ldots ,2\,\lceil 10x \log \sqrt{x}\rceil \}$, the total weight of the chords of length $\sqrt{48\eps}\cdot \left(1+\frac{1}{10x}\right)^i$ is $O(20\cdot \beta_x)=O(\sqrt{x\eps})$. Consequently, $\|E(H)\|= O(\sqrt{x\eps}\cdot x \log x) = O(\eps^{1/2}\cdot x^{3/2}\log x)$.
\end{proof}

\paragraph{Acknowledgements.~}
The second-named author thanks Mike Dinitz for interesting discussions.  We thank anonymous reviewers of FOCS 2024 for thorough comments, especially on approximating spanners of general graphs. Hung Le and Cuong Than are supported by the NSF CAREER award CCF-2237288, the NSF grant CCF-2121952, and a Google Research Scholar Award.
Research by Csaba D.\ T\'oth was supported by the NSF award DMS-2154347.
Shay Solomon and Tianyi Zhang are funded by the European Union (ERC, DynOpt, 101043159).  Views and opinions expressed are however those of the author(s) only and do not necessarily reflect those of the European Union or the European Research Council.  Neither the European Union nor the granting authority can be held responsible for them.  Shay Solomon is also supported by the Israel Science Foundation (ISF) grant No.1991/1 and by a grant from the United States-Israel Binational Science Foundation (BSF), Jerusalem, Israel, and the United States National Science Foundation (NSF).
\vspace{5mm}
\bibliographystyle{alphaurl}
\bibliography{ref}

\end{document}